\documentclass[12pt]{article}

\usepackage{amssymb, amsmath, amsthm}
\usepackage{graphicx,subfig,tikz}
\usepackage{cite}
\usepackage[left=2cm,right=2cm,top=2cm,bottom=2cm]{geometry}

\newtheorem{theorem}{Theorem}
\newtheorem{lemma}{Lemma}

\newtheorem{prop}{Proposition}
\newtheorem{conj}{Conjecture}

\usepackage[colorlinks=true,linktocpage=true,allcolors=blue]{hyperref}

\def\be{\begin{equation}}
\def\ee{\end{equation}}
\def\bea{\begin{eqnarray}}
\def\eea{\end{eqnarray}}

\newcommand{\td}{\text{d}}

\title{ \bf{Moduli space of stationary vacuum black holes from integrability}}

\author{James Lucietti\footnote{j.lucietti@ed.ac.uk}\,  and Fred Tomlinson\footnote{f.tomlinson@ed.ac.uk}
\\ \\ \small \sl School of Mathematics and Maxwell Institute for Mathematical Sciences, \\ \small \sl    University of Edinburgh, King's Buildings, Edinburgh, EH9 3JZ, UK }

\date{}

\begin{document}

\maketitle

\begin{picture}(0,0)(0,0)
\put(350, 240){}
\put(350, 225){}
\end{picture}

\begin{abstract}
We consider the classification of asymptotically flat, stationary, vacuum black hole spacetimes in four and five dimensions, that admit one and two commuting axial Killing fields respectively. It is well known that the Einstein equations reduce to a harmonic map on the two-dimensional orbit space, which itself arises as the integrability condition for a linear system of spectral equations. We integrate the Belinski-Zakharov  spectral equations along the boundary of the orbit space and use this to fully determine the metric and associated Ernst and twist potentials on the axes and horizons. This is sufficient to derive the moduli space of solutions that are free of conical singularities on the axes, for any given rod structure.  As an illustration of this method we obtain constructive uniqueness proofs for the Kerr and Myers-Perry black holes and the known doubly spinning black rings.
\end{abstract}

\newpage 
\tableofcontents

\newpage 
\section{Introduction}

The classification of equilibrium black hole solutions  is a fundamental problem  in General Relativity (GR). In four spacetime dimensions this is essentially answered by the celebrated black hole uniqueness theorem, which roughly states that the only asymptotically flat, stationary black hole solution to the vacuum Einstein equations is the Kerr solution~\cite{Chrusciel:2012jk}.   A striking consequence of this is that any such solution is simply labelled by two parameters, the mass $M$ and angular momentum $J$ which must obey 
\be
|J| \leq M^2.   \label{kerrbound}
\ee  
One of the key underlying structures which allows one to establish this theorem is the remarkable fact that the vacuum Einstein equations for stationary and axisymmetric spacetimes reduce to a harmonic map on the two-dimensional orbit space.

In higher dimensional spacetimes, it has been known for some time that there can be no such simple uniqueness theorem.  This was revealed by the striking discovery of the black ring, a five-dimensional, asymptotically flat black hole solution with horizon topology $S^2\times S^1$~\cite{Emparan:2001wn}. Alongside the spherical topology black hole discovered by Myers-Perry~\cite{Myers:1986un}, this explicitly shows that even vacuum black holes are not uniquely specified by their mass and angular momenta.  A natural problem which then presents itself  is to classify all higher-dimensional stationary black hole solutions to the Einstein equations.  This is a central open problem in higher-dimensional GR and is largely unsolved, see the reviews~\cite{Emparan:2008eg, Hollands:2012xy}.

Nevertheless, substantial progress has been made for $D$-dimensional stationary spacetimes which admit $D-3$ commuting axial Killing fields that commute with the stationary Killing field~\cite{Emparan:2001wk, Harmark:2004rm}. These generalise the four-dimensional stationary and axisymmetric spacetimes that contain the Kerr solution. Crucially, the vacuum Einstein equations for spacetimes with such symmetry  reduce to an integrable harmonic map on the two-dimensional orbit space. However, asymptotic flatness is only compatible with such a symmetry assumption for $D=4,5$. This is because if $D>5$ the rank of the rotation group $SO(D-1)$ is less than $D-3$ (however, this symmetry assumption is compatible with Kaluza-Klein asymptotics).  For this reason, most advances in constructing and classifying higher-dimensional black hole solutions has been for the class of $D=5$ asymptotically flat, stationary spacetimes admitting two commuting axial Killing fields. Indeed, both the Myers-Perry black hole and the black ring belong to this class of solutions. 

In this paper we will consider $D=4,5$ asymptotically flat stationary vacuum spacetimes that  admit $D-3$ commuting axial Killing fields. For such spacetimes it has been shown that the two-dimensional orbit space can be identified with a half plane $\{ (\rho, z) \; | \; \rho>0 \} \subset \mathbb{R}^2$, where $(\rho, z)$ are the global Weyl-Papapetrou coordinates~\cite{Hollands:2007aj, Hollands:2008fm}.  The following uniqueness theorem for black hole spacetimes in this class has been previously established by Hollands and Yazadjiev (an analogous result also holds for the $D$-dimensional asymptotically Kaluza-Klein case): 

 \begin{theorem}[\hspace{-.1cm}~\cite{Hollands:2007aj, Hollands:2008fm}] There is at most one $D=4$ or $5$-dimensional  asymptotically flat, stationary, vacuum spacetime with $D-3$ commuting axial Killing fields,  containing a  non-degenerate\footnote{An analogous theorem can be established for degenerate horizons, i.e. for extreme black holes in this class~\cite{Figueras:2009ci}. In this paper we will only consider non-degenerate horizons.} event horizon, for a given rod structure and a given set of horizon angular momenta. \label{th1}
 \end{theorem}

Roughly speaking, the rod structure is data that encodes the fixed points of the axial Killing fields and the topology of the horizons. More precisely, the boundary of the orbit space can be identified with the $z$-axis of the half-plane which divides into a set of intervals, called rods, each of which either corresponds to a component of the axis or horizon. Each axis rod is defined by the vanishing of a certain periodic linear combination of the axial Killing fields, called the rod vector (of course, for $D=4$ there is only one axial Killing field and hence only one type of axis rod). The rod structure  corresponds to this set of rods together with their lengths and the axis rod vectors.  

For $D=4$ and a connected horizon the above theorem reduces to the classic black hole uniqueness theorem for the Kerr black hole: it says that any solution is uniquely parameterised by the horizon rod length $\ell_H$ and angular momentum $J$ (there are no finite axis rods). The (nonextreme) Kerr solution realises all possible values of this data, $\ell_H>0, J \in \mathbb{R}$, and hence the classification for this case is complete (in this case one can  of course also use the $M, J$ to label solutions as is traditionally done).  As in the classic $D=4$ case, the proof of Theorem \ref{th1} is nonconstructive and involves a nonlinear divergence identity (Mazur identity) which characterises the `difference' of two solutions to the corresponding harmonic map problem.   Therefore this theorem does not address the crucial question of existence:  {\it for what rod structures and horizon angular momenta do regular solutions actually exist?}

Indeed, the existence question is largely open even for $D=4$. In this case the other possible rod structures correspond to black holes with multiple horizons, with finite axis rods separating the disjoint horizon rods.  There is a general expectation that equilibrium configurations describing such solutions in the vacuum cannot exist due to their mutual gravitational attraction. In fact, by adapting existence results for harmonic maps to this problem, Weinstein has shown that a unique $N$-component black hole solution exists given any rod structure and horizon angular momenta, which is regular everywhere away from the axis~\cite{Weinstein1990, Weinstein1992, Weinstein1994}. However, such solutions may still suffer from conical singularities on the finite axis components (i.e. those not connected to infinity).  Physically, these singularities are related to the force of attraction between the black holes and it is conjectured that for $N>1$ such solutions always do possess conical singularities. Evidence that this force is always attractive has been obtained by studying various special cases~\cite{TY, Weinstein1994}.  

Candidate multi-black hole solutions, known as the multi-Kerr-NUT solutions, have been known for some time~\cite{KN, N, BZ2}, although an analysis of the potential conical singularities has proven to be essentially intractable. Naturally, the $N=2$ case corresponding to a double-black hole has been the most extensively studied. From the above theorem this solution depends on five-parameters (two horizon rod lengths, one axis rod length and the angular momentum of each horizon), that are related by the equilibrium condition (i.e. the condition for removal of the conical singularity on the finite axis rod).  The study of the equilibrium condition for the double-Kerr-NUT solution has been the subject of much work, see e.g.~\cite{Tomimatsu:1981bc, MR}. However, even if one can give a general proof that the equilibrium condition for the double-Kerr-NUT solution is never satisfied, this would still not give a proof of the nonexistence of a regular double-black hole, since it is not a priori clear that it contains the general solution with these boundary conditions.   
Recently, this conjecture has been settled by Hennig and Neugebauer~\cite{Neugebauer:2011qb}: a regular double-black hole solution does not exist. The proof consists of two steps: (i) employing the inverse scattering method from integrability theory to prove that the general solution with such boundary conditions is contained in the known double-Kerr-NUT solution (this was already shown in earlier work by Varzugin~\cite{Varzugin:1997ee} and Meinel and Neugebauer~\cite{Neugebauer:2003qe}); (ii) showing that the equilibrium conditions are incompatible with the area-angular momentum inequality for a marginally trapped surface~\cite{Hennig:2008yw, Dain:2011pi, Chrusciel:2011iv}. 

The $D=5$ case is more complicated for two principle reasons. Firstly, there are more horizon topologies compatible with biaxial symmetry: $S^3$, $S^2\times S^1$ and lens spaces $L(p,q)$. Secondly, for every horizon topology (including multi-horizons) there can be an arbitrary number of finite axis rods on which different linear combinations of the two axial Killing fields  vanish -- these correspond to nontrivial 2-cycles in the domain of outer communication (DOC).  Recently, a theorem which partially addresses the existence question in this context has been established by Khuri, Weinstein and Yamada~\cite{Khuri:2017xsc}.   It is a five-dimensional analogue of Weinstein's theorem for $D=4$ multi-black holes. Indeed, the proof involves the theory of harmonic maps adapted to this setting, although it requires one to make a certain technical assumption on the rod structure.
 
  \begin{theorem}[\hspace{-.04cm}\cite{Khuri:2017xsc}] Given any admissible rod structure obeying a compatibility condition, and given any set of horizon angular momenta, there exists exactly one $5$-dimensional, asymptotically flat, stationary, vacuum, spacetime with two commuting axial Killing fields and containing a non-degenerate event horizon, if and only if the metric is smooth at the axes. \label{th2}
   \end{theorem}
   
The rod structure is required to be admissible in order to avoid potential orbifold singularities at the fixed points of the biaxial symmetry; on the other hand the compatibility condition on the rod structure appears to be a technical condition required for the proof (see Section \ref{sec:prelim} for details).  While Theorem \ref{th2} does not settle the classification of {\it regular} solutions, it greatly simplifies the problem. 
 In particular, it reduces it to a regularity analysis of the axes that requires two conditions to be met: (i) the metric components must be smooth and even functions of $\rho$ up to the axes, (ii) there are no conical singularities at the inner axis rods (it has  been shown that there are no conical singularities at the two semi-infinite axis rods~\cite{Khuri:2018udf}).  

 In contrast to the $D=4$ case, it is known there are a number of regular solutions with nontrivial rod structure. In addition to the black ring, several remarkable multi-black hole solutions have been constructed: the black Saturn~\cite{Elvang:2007rd} --  an equilibrium configuration comprising of a black hole surrounded by a black ring -- and various double-black ring configurations~\cite{Iguchi:2007is, Izumi:2007qx, Elvang:2007hs}. These were constructed using the inverse scattering method of Belinski-Zakharov (BZ) which is  based on their spectral equations~\cite{BZ1, BZ2}.  Notably, however, the existence of a regular vacuum black lens, i.e. a black hole with lens space horizon topology, has remained an open problem. Several attempts at constructing such solutions have been made, again using the BZ method~\cite{Evslin:2008gx, Chen:2008fa, Tomizawa:2019acu}. These have all resulted in singular solutions, the mildest being a conical singularity on an inner axis rod. Unfortunately, the BZ method is not fully systematic and requires guesswork at various stages and therefore does not necessarily reproduce the most general solution for a given rod structure (cf. the  $D=4$  multi-Kerr black hole discussed above). Therefore,  these works cannot be taken as proof of  the nonexistence of regular vacuum black lenses.

The purpose of this paper is to use the spectral equation of BZ to {\it systematically} investigate all possible solutions for any given rod structure.  In particular, we explicitly integrate the BZ spectral equations along the axes and horizons, and around infinity. Then, using this we show that one can determine the metric everywhere on the axes and horizons for any given rod structure purely algebraically. Our main result can be summarised  as follows (see Theorem \ref{th:main} for a precise statement):

\begin{theorem}
\label{th3}
Consider a $D=4,5$ stationary vacuum spacetime as in Theorem \ref{th1}.  
On every component of the axis and horizon, the general solution for the metric components and the associated Ernst or twist potentials are rational functions of $z$.  These functions are explicitly determined in terms of the rod structure, horizon angular momenta, horizon angular velocities and certain gravitational fluxes, which are subject to a set of nonlinear algebraic constraints.
\end{theorem}

Thus the solution depends on a number of continuous parameters which are geometrically defined: the rod lengths,  the angular momenta and angular velocities of each horizon, and certain gravitational fluxes. The gravitational fluxes are invariants associated to each finite axis rod. In the spacetime the finite axis rods correspond to noncontractible $(D-3)$-cycles and the fluxes are integrals of certain closed $(D-3)$-forms constructed from the Killing fields. For every axis rod one can define an associated Ernst potential from the Killing fields which are nonzero on that rod. The change in  Ernst potential across the associated rod is then precisely the gravitational flux through the corresponding 2-cycle. It is worth noting that similar gravitational fluxes arise in the recently found thermodynamic identities for $D=5$ black holes in this class~\cite{Kunduri:2018qqt}. 

As mentioned in  our theorem, the parameters in the general solution must obey certain nonlinear algebraic equations. These arise from integrating the BZ spectral equations along the $z$-axis and around the `semi-circle' at infinity. Furthermore, imposing  the metric is free of conical singularities on the axes and horizons typically imposes further constraints on the parameters. Thus we are able to address part (ii) of the regularity problem left open by Theorem \ref{th2}. Hence, our method is particularly useful for ruling out regular solutions with a prescribed rod structure. For example, one can prove that a $D=5$ solution with no horizon and one finite axis rod must be conically singular at the finite axis rod; this is of course guaranteed by the no-soliton theorem for vacuum solutions (even without biaxial symmetry), although it illustrates that our method is capable of showing that certain rod structures must lead to conically singular solutions.

In principle, using our theorem one can determine the full moduli space of regular black holes in this class, up to the regularity problem (i).  Indeed, specialising Theorem \ref{th3} to the rod structures of the known nonextremal black holes with connected horizons allows us to show that the full moduli space of regular solutions coincides with that of the Kerr black holes, the Myers-Perry black holes and the known doubly spinning black rings.  For the latter case, this proves that the Pomerasky-Sen'kov doubly spinning black ring~\cite{Pomeransky:2006bd} is indeed the most general solution with that rod structure, a fact which does not seem to have been addressed in the previous literature. In a subsequent paper we will apply our method to investigate the (non)existence of new types of regular black hole solutions, most notably a black lens; in the Discussion we present our preliminary findings.

Our method may be thought of as a higher-dimensional analogue of the $D=4$ methods of Varzugin~\cite{Varzugin:1997ee,Varzugin:1998wf} and Meinel and Neugebauer~\cite{Neugebauer:2003qe}, which both lead to simple constructive uniqueness proofs for Kerr.  In particular, Varzugin integrated the  BZ spectral equations along the axis and horizons and used this to show that the $N$-black hole solution is contained in the $2N$-soliton solution of BZ~\cite{BZ2}. We also integrate the BZ spectral equations along the boundaries, although our analysis of its solution differs, and we give a simple method to extract the spacetime metric, so even for $D=4$ it offers an alternative approach. On the other hand, Meinel and Neugenbauer integrated a different spectral equation along the axis, whose integrability condition gives the Ernst equations, and used this to determine the Ernst potential on the axis.  It would be interesting to investigate the precise relationship between these various methods.

This paper is organised as follows. In Section \ref{sec:stataxis} we recall well-known properties of  stationary vacuum spacetimes with $D-3$ commuting axial Killing fields and introduce various Ernst potentials which will feature later (this section also serves to set our notation). In  Section \ref{sec:lax} we derive the general solution to the BZ spectral equations on the axes, horizons and around infinity, and use this to construct the general solution to the Einstein equations on the axes and horizons. In Section \ref{sec:4d} we specialise to $D=4$: we compute the asymptotic charges of the general solution and derive the moduli space of Kerr black holes. In Section \ref{sec:5d} we specialise to $D=5$: we compute the asymptotic charges of the general solution and derive the moduli spaces of the Myers-Perry black hole and the doubly spinning black ring. In Section \ref{sec:disc} we discuss our results and future work. We relegate various results to the Appendix.

\section{Stationary spacetimes with $D-3$ axial Killing fields}
\label{sec:stataxis}

\subsection{Einstein equations and rod structure}
\label{sec:prelim}

Let $(M, \mathbf{g})$ be a $D$-dimensional stationary spacetime with $D-3$ commuting axial Killing vector fields that also commute with the stationary Killing field. We denote the stationary Killing field $k$ and the remaining $D-3$ axial Killing fields $m_i, i = 1, \dots, D-3$, and assume these generate an isometry group $G=\mathbb{R}\times U(1)^{D-3}$. We define coordinates $(t, \phi^i)$ adapted to  the stationary and axial symmetries, so $k=\partial_t$ and $m_i= \partial_{\phi^i}$, and choose $m_i$ to be generators with $2\pi$-periodic orbits, i.e. the angles $\phi^i$ are $2\pi$-periodic. We also assume that there is at least one point in spacetime that is a fixed point of the axial symmetry (as is the case for asymptotically flat spacetimes).

As is well known, under such assumptions the spacetime metric can be written in Weyl-Papapetrou coordinates~\cite{Emparan:2001wk, Harmark:2004rm}
\be
\mathbf{g} = g_{AB}(\rho,z) \td x^A \td x^B + e^{2\nu(\rho,z)}(\td \rho^2+ \td z^2), \label{WP}
\ee
where $A \in \{ 0, 1, \dots, D-3 \}$, $\partial_A$ are the Killing fields and 
\be
\det g_{AB} = -\rho^2.  \label{detg}
\ee  
The vacuum Einstein equations reduce to
\be
\partial_\rho U+ \partial_z V=0,  \label{einsteineq}
\ee
where 
\be
U= \rho \partial_\rho g g^{-1}, \qquad V= \rho \partial_z g g^{-1}  \label{UV}
\ee  
and the conformal factor, $e^{2 \nu}$, is then determined by
\be
\partial_\rho \nu  =-\frac{1}{2\rho}+ \frac{1}{8\rho}\text{Tr}(U^2-V^2),\,  \qquad \partial_z \nu = \frac{1}{4\rho} \text{Tr}\, UV \; .  \label{conf}
\ee
Indeed, the integrability condition for (\ref{conf}) is (\ref{einsteineq}).

We now turn to  global assumptions. We will restrict to asymptotically flat spacetimes, i.e., asymptotically Minkowski such that $k=\partial_t$ in the standard Cartesian chart. This is only compatible with our symmetry assumption if $D=4,5$, which we assume throughout.  In fact, it is necessary to make a number of further global assumptions (see review~\cite{Hollands:2012xy}). In particular we assume: (i) there exists a spacelike hypersurface $\Sigma$ with asymptotically flat end that intersects the event horizon (if there is one) on a compact cross-section $H$; (ii) the stationary Killing field $k$ is complete; (iii) the domain of outer communication (DOC) is globally hyperbolic; (iv) the horizon is non-degenerate.  
   
Under these assumptions, a number of global results have been derived~\cite{Hollands:2007aj, Hollands:2008fm, Hollands:2012xy}. In particular, it has been shown that Weyl-Papapetrou coordinates (\ref{WP}) provide a global chart in the DOC away from the horizon and axes of symmetry. Furthermore,  the orbit space of $M$ under the isometry group $\hat{M}= M/G\cong \Sigma/U(1)^{D-3}$, is a 2d simply connected manifold with boundaries and corners, which may therefore be identified with the  half-plane 
\be
\hat{M}= \{(\rho, z)\,  | \, \rho>0 \} \; .  \label{HP}
\ee
The boundary of the orbit space $\rho=0$ corresponds to the $z$-axis and this splits into intervals, called rods,  $(-\infty, z_1), (z_1, z_2), \dots, (z_n, \infty)$, with $z_1<z_2<\dots <z_n$, each of which corresponds to a connected component of the horizon orbit space $\hat{H}=H/U(1)^{D-3}$, or an axis where an integer linear combination of the axial Killing fields -- called the rod vector -- vanishes. For $D=4$ there is only one axial Killing field and therefore only one type of axis rod.  The endpoints of the rods  $z_a, a=1, \dots, n$, correspond to the corners of the orbit space, each of which corresponds   to where an axes intersects a horizon, or for $D>4$, a fixed point of the $U(1)^{D-3}$-action (i.e. $m_i=0$ for all $i=1, \dots, D-3$, which occurs precisely where two axes intersect). 

 Let us denote the rods by $I_a$,  for $a= 1, \dots, n+1$, and the length of the finite rods  $I_a=(z_{a-1}, z_a)$  by  $\ell_a=z_a-z_{a-1}$ for $a=2, \dots, n$. Given any axis rod $I_a$ the corresponding  rod vector takes the form
 \be
 v_a=v_{a}^i m_i  \label{vaxis}
 \ee
where $(v_{a}^{i})_{i=1, \dots, D-3}$ are coprime integers.  If $D=5$, for any adjacent axis rods $I_a$ and $I_{a+1}$ separated by the corner $z=z_a$ the associated rod vectors must satisfy the condition
\be
\det \left(\begin{array}{cc} v^1_a & v_a^2 \\ v^1_{a+1} & v^2_{a+1}\end{array} \right) = \pm 1  \; .   \label{admissible}
\ee
Following~\cite{Khuri:2017xsc}, we will call any rod structure satisfying (\ref{admissible}) {\it admissible}.  We will denote the union of all axis rods by $\hat{A}$ and all horizon rods by $\hat{H}$.
 The collection of all this data  
 \be
 \{ (\ell_a, v_a) \, |  \, I_a \subset \hat{A}\} \cup \{ \ell_a \, | \, I_a\subset \hat{H} \}\label{rodstruct}
 \ee
 is known as the {\it rod structure}. We will often denote the semi-infinite axis rods by $I_L= I_1= (-\infty, z_1)$ and $I_R= I_{n+1}= (z_n, \infty)$. For definiteness, in the $D=5$ case we will choose the $m_i$ such that $m_2=0$ on $I_L$ and $m_1=0$ on $I_R$, i.e., the rod vectors $v_L= (0,1)$ and $v_R=(1,0)$ relative to the basis $(m_1, m_2)$.
 
 For $D=5$ any finite axis rod $I_a$ lifts to a 2-cycle in the spacetime. Explicitly this is given by the surface $C_a$ obtained from the fibration of the nonzero $U(1)$ Killing field $u_a=u_a^i m_i$ over the closure of $I_a$ (recall $v_a=0$ on $I_a$). If the adjacent rods  are both axis rods then $u_a$ must vanish at the endpoints of $I_a$ and $C_a$ has the topology of $S^2$;  if only one adjacent rod is an axis rod (and hence the other a horizon) then $u_a$ only vanishes at the corresponding endpoint so $C_a$ is topologically a 2-disc; finally if both adjacent rods are horizon rods then $u_a$ does not vanish at either endpoint and  $C_a$ is topologically a cylinder. 
 
Another important set of invariants for such solutions are the Komar angular momenta of each connected component of the horizon $H_a$ defined by
\be
J^a_i = \frac{1}{16\pi} \int_{H_a}  \star \td m_i,
\ee
where we fix the orientation $\epsilon_{0 1\dots D-3 \rho z}>0$.
From a standard argument, invoking Stokes' theorem and the Einstein equation, these are related to the total angular momenta of the spacetime $J_i=\sum_{a} J_i^a$.
Due to the assumed symmetry these can be reduced to integrals over the horizon rods using 
\be
\int_{H_a} \star \alpha =(2\pi)^{D-3} \int_{I_a} \star (m_1 \wedge \dots \wedge m_{D-3} \wedge \alpha)   \label{Hint}
\ee  where $\alpha$ is any $U(1)^{D-3}$-invariant 2-form. This gives
\be
J_i^a= \frac{1}{8} (2\pi)^{D-4} (\chi_i(z_{a})-\chi_i(z_{a-1}))  \; ,   \label{Jhorizon}
\ee
where $\chi_i$ are the twist potentials defined by 
\be
\td \chi_i =  \star ( m_1\wedge \dots m_{D-3}  \wedge \td m_i)  \; .  \label{twist}
\ee
The existence of globally defined twist potentials follows from the fact the vacuum Einstein equations imply the RHS of (\ref{twist}) is a closed 1-form and that under our assumptions the DOC is simply connected. Observe that the twist potentials are constant on any axis rod. Therefore, they can only vary across a horizon rod and the above shows that the change in twist potential across any horizon rod is precisely the angular momenta of the corresponding horizon in spacetime.
 
 We now are now in a position to consider the uniqueness and existence theorems mentioned in the Introduction in more detail. Theorem \ref{th1} guarantees that there is at most one solution for a given rod structure (\ref{rodstruct}) and horizon angular momenta (\ref{Jhorizon}).   However, as highlighted in the Introduction, the main limitation of this theorem is that it does not address the crucial question of existence: for what rod structure and angular momenta do there exist regular black hole solutions? This is not an issue for $D=4$ as the uniqueness theorem reduces to the classic no-hair theorem for the Kerr black holes (although for multi-black holes this is largely open, as explained in the  Introduction).
  
 However, for $D=5$ the uniqueness theorem is less powerful as even for a connected horizon an arbitrary number of axis rods are  allowed in principle.  To this end, Theorem \ref{th2} has been recently established, which guarantees the existence of a solution for any admissible rod structure that obeys the following {\it compatibility} condition:  if there are three consecutive axis rods $I_{a-1}, I_a, I_{a+1}$, then the compatibility condition states that
 \be
 v_{a-1}^1 v_{a+1}^1\leq 0  \;  , \label{comp}
 \ee
 whenever  the admissibility condition (\ref{admissible}) between the pairs $I_{a-1}, I_a$ and $I_a, I_{a+1}$ are obeyed with positive determinant.
As explained in the Introduction, this theorem guarantees the solution is regular in the DOC away from the axes. Therefore it does not address regularity of the solution on the axes, which generically will possess conical singularities on the finite axis rods. 
 It is instructive to consider certain special cases of Theorem \ref{th1} and \ref{th2}.
 
  First, consider the case of no horizon.  Then, it is well known from the no-soliton theorem that the only regular solutions in this class of spacetimes  is Minkowski spacetime (indeed, this result does not even assume biaxial symmetry). Hence, it must be that the only regular solution with the same rod structure as Minkowski spacetime is Minkowski spacetime itself. Furthermore, any solution with non-Minkowski rod structure  must be singular on some component of the axis. For example,  consider the Eguchi-Hanson soliton
  \be
\td s^2_{\text{EH}} = -\td t^2+  \frac{\td  R^2}{1-\frac{a^4}{R^4}}+ \tfrac{1}{4} R^2 \left( 1- \frac{a^4}{R^4} \right) (\td \psi+ \cos\theta \td \phi)^2+ \tfrac{1}{4} R^2 (\td \theta^2+ \sin^2\theta \td \phi^2)  \; ,\label{EH}
\ee
where $R \geq a$.
 As is well known, if $\psi$ is periodically identified with period $2\pi$ this gives a smooth metric with a bolt at $R=a$ which is asymptotically locally Euclidean with $S^3/\mathbb{Z}_2$ topology for large $R$.  However, if instead we take $(\theta, \psi, \phi)$ to be Euler angles on $S^3$, we get an asymptotically Minkowski spacetime, except now with a conical singularity at the bolt. This example then gives a nontrivial rod structure with one finite axis rod corresponding to the bolt $R=a$ separating the two semi-infinite rods. In particular, relative to the basis $(m_1, m_2)$ introduced above, the rod vectors are $v_L=(0,1)$, $v_B=(1,1)$ and $v_R=(1,0)$, where $v_B$ is the rod vector on the bolt, thus giving an admissible rod structure (\ref{admissible}). It is a one parameter family of solutions, where the parameter can be taken to be length of the axis rod, in line with the above theorems (since there is no horizon the only moduli are the rod lengths). One might wonder whether the more general Gibbons-Hawking solitons similarly give solutions with multiple axis rods in Theorem \ref{th2}. In Appendix \ref{apen:GH} we show that in fact these do not possess an admissible rod structure (instead they possess orbifold singularities at the corners $z_2, \dots, z_{n-1}$ and thus correspond to solutions of a different theorem in~\cite{Khuri:2017xsc}).

 Now consider black hole solutions with a single horizon.  First, suppose that the angular momenta $J_i=0$.  Then it can be shown that the solution must be static~\cite{Figueras:2009ci} and hence by the static uniqueness theorem the solution must be the Schwarzschild black hole~\cite{Gibbons:2002bh}.  This implies that any regular solution in this class must have the same rod structure as Schwarzschild, i.e. one horizon rod separating the two semi-infinite axis rods. In other words, any solution with a single horizon, $J_i=0$ and finite axis rods, must be conically singular on the axis rods. This shows that for single black holes, not all rod structures and angular momenta lead to regular solutions.
   
 Next, consider the rod structure of the Myers-Perry solution: a single horizon rod and two semi-infinite axis rods (i.e. this is the same as that of Schwarzschild). Then, since the Myers-Perry solution realises all possible data $\ell_H>0$ and $(J_1, J_2) \in \mathbb{R}^2$ it is the only solution in this class. A self-contained proof of this was given in earlier work~\cite{Morisawa:2004tc}. This case is analogous to the Kerr solution in four dimensions.
 
 For a black ring something more interesting happens. Consider the rod structure of the known black ring, i.e. one horizon rod and one finite axis rod (and two semi-infinite axis rods). In this case there are four parameters in the uniqueness theorem, namely the horizon and finite axis rod lengths $\ell_H, \ell_A$ and  the angular momenta $J_i$.  However, the most general known regular black ring solution is the three parameter doubly spinning solution~\cite{Pomeransky:2006bd}.  Thus, in this case the known regular solutions do not occupy all parts of the possible parameter space. A way to understand this is that generically one has a conical singularity at the finite axis rod and its removal imposes a constraint on the four available parameters thus leaving a three parameter subset. Nevertheless, this raises the question: are there other regular black rings which occupy different parts of the possible moduli space? A definitive answer requires constructing the most general solution with such a rod structure. We answer this question in the negative in this work.\footnote{In fact, a four parameter family of `unbalanced' doubly spinning back rings has been constructed~\cite{Chen:2011jb}, i.e., these suffer from a conical singularity at the axis rod.  It is possible that these do fill out the whole moduli space, although as far as we aware this has not been checked in the literature. If so, then by the uniqueness theorem this would have to be the general solution and hence the known three-parameter family of black rings would have to be the most general regular solution with this rod structure.}
  
Remarkably,  regular multi-black hole solutions do exist in five dimensions. The first such example constructed was the black Saturn, an equilibrium configuration of a spherical black hole surrounded by a black ring that is balanced by angular momentum~\cite{Elvang:2007rd}. This solution is a four parameter family corresponding to the horizon rod lengths and one angular momentum for each black hole (the rod length of the finite axis rod between the black holes is fixed by removing the associated conical singularity). There should be a more general six-parameter family where both the spherical black hole and black ring are doubly spinning which is yet to be constructed.  Similarly, regular four-parameter multi-black rings have been constructed: di-rings are concentric rings rotating in the same plane~\cite{Iguchi:2007is}, and bi-rings rotate in orthogonal planes~\cite{Izumi:2007qx, Elvang:2007hs}. Again, these should be part of a more general six-parameter family of two doubly spinning black rings that remains to be constructed.

Let us now consider the dimension of the moduli space of solutions in the above theorems.  Suppose we have $h$ horizon rods and $a$ finite axis rods. Then, counting the number of continuous parameters appearing in the above theorems (i.e. $h+a$ rod lengths and $(D-3)h$ horizon angular momenta), shows the dimension of the moduli space of solutions with $h$ horizons and $a$ finite axis rods that are potentially  singular on the axis is,
\be
\text{dim} \;  \mathcal{M}_{\text{sing}}^{h,a} = (D-2)h+a \;  \label{dimmodsp}.
\ee
From experience with the known solutions one expects that removal of the conical singularities on each finite axis rod reduces the number of parameters by one, thus reducing the total by $a$.  Hence, a natural conjecture, which agrees with the known solutions, is that provided regular solutions actually exist, the dimension of the moduli space of regular solutions  $\mathcal{M}_\text{reg}^{h,a}$  with $h$ horizon rods and $a$ finite axis rods is simply
\be
\text{dim}\, \mathcal{M}_{\text{reg}}^{h,a} \overset{?}{=} (D-2)h \; .
\ee

\subsection{Geometry of axes and horizons}

In this section we write down a general form for the metric near $\rho=0$, i.e., near any axis or horizon, which will be useful for our purposes. The analysis of the geometry near an axis  and near a horizon is very similar, although for clarity of presentation we will use different notations for the metric in these two cases.  Most of the material in this section is  well-known. In Appendix \ref{app:corners} we also include a regularity analysis at the corners of the orbit space which is perhaps less well-known.

\subsubsection{Axes}
\label{sec:axes}

First consider an axis rod $I_a$.  For simplicity of notation we temporarily drop the labelling of each rod.  It is convenient to introduce an adapted basis for the $D-2$ commuting Killing fields $\tilde{E}_A=(e_\mu, v)$ where $\mu=0, \dots D-4$ and $v=v^im_i$ is the rod vector corresponding to $I_a$. For $D=4$ we simply take $e_0=k$. For $D=5$ we take $e_\mu=(k,u)$ where $u$ is an axial Killing field 
\be
u= u^im_i, \qquad \text{such that} \qquad A=\left( \begin{array}{cc} u^1 & u^2 \\ v^1 & v^2  \end{array} \right) \in GL(2, \mathbb{Z}) \; ,   \label{uAdef}
\ee
i.e. $(u, v)$ are $2\pi$-periodic generators of the $U(1)^2$-action.  It is worth emphasising that $u$ is defined only up to an additive integer multiple of the rod vector $v$. Then,  relative to the adapted basis the metric on the orbits of the isometry can be written as
\be
\tilde{g} = \left( \begin{array}{cc} h_{\mu\nu} - \rho^2 h^{-1}w_\mu w_\nu & h^{-1} \rho^2 w_{\mu} \\  h^{-1} \rho^2 w_{\nu} & - h^{-1} \rho^2  \end{array} \right)  \; .  \label{gaxisadapted}
\ee
Note that the normalisation (\ref{detg}) is automatically imposed in this basis. Here,  $h_{\mu\nu}$ is an invertible  $(D-3)\times (D-3)$ matrix and its determinant $h = \det h_{\mu\nu}<0$.      A regular axis requires $h_{\mu\nu}, w_\mu$ to be smooth functions of $(\rho^2, z)$ and 
\be
\lim_{\rho \to 0, \;  z\in  I_a} \frac{\rho^2 e^{2\nu}}{| v |^2}=  1  \label{smoothaxis}  \; .
\ee
This ensures the absence of a conical singularity at $I_a$~\cite{Harmark:2004rm}. 

The inverse metric in this adapted basis is
\be
\tilde{g}^{-1} = \left( \begin{array}{cc}  h^{\mu\nu} & w^\mu \\ w^\nu & - h \rho^{-2} + w^\rho w_\rho \end{array} \right)
\ee
where $h^{\mu\nu}$ is the inverse matrix of $h_{\mu\nu}$ and $w^\mu= h^{\mu\nu} w_\nu$. The requirement of a smooth axis  implies the following limits exist
\be
\mathring{U} = \lim_{\rho \to 0} U  \; ,\qquad \mathring{V}= \lim_{\rho\to 0} \frac{V}{\rho} \; ,\label{U0V0}
\ee
where here and throughout we denote quantities evaluated in the limit $\rho \to 0$ by a circle above. Explicitly, relative to the adapted basis we find
\be
\mathring{\tilde{U}} = \left( \begin{array}{cc} 0 &- 2w_\mu \\ 0 & 2  \end{array} \right)  \; , \qquad   \label{UVtilde0}
 \mathring{\tilde{V}}= \left( \begin{array}{cc}   (\partial_z h_{\mu\rho} ) h^{\rho \nu}  & - h h_{\mu\nu} \partial_z(h^{-1} w^\nu)  \\  0  &    - (h^{-1} \partial_z h ) \end{array} \right)   
\ee
where here, and in what follows, all quantities on the RHS are understood to be evaluated at $\rho=0$. Taking the $\rho \to 0$ limit of the second equation in  (\ref{conf}) it follows that the conformal factor  on the axis obeys
\be
\partial_z \mathring{\nu} = - \frac{\partial_z h }{2 h}
\ee
which integrates to
\be
e^{2\mathring{\nu}}= -\frac{c^2}{h}  \label{nu0}
\ee
where $c$ is a constant.

Collecting these results, we deduce that the metric induced on the axis component associated to $I_a$ is
\be
\mathbf{g}_a= -\frac{c_a^2 \td z^2}{h^a(z)} + h^a_{\mu\nu}(z) \td x^\mu \td x^\nu  \; , \label{axisrodmetric}
\ee
where $x^\mu$ are adapted coordinates so that $e_\mu= \partial_\mu$, $\mu=0,1$, and we have reinstated the rod labels.  This is a $(D-2)$-dimensional smooth Lorentzian metric for $z \in I_a$. The condition for the absence of a conical singularity in the spacetime at $I_a$ (\ref{smoothaxis})  is
\be
c_a=1 \; , \label{balance}
\ee
which is sometimes referred to as the equilibrium or balance condition.

For $D=5$ one or both of the adjacent rods to $I_a$ may be another axis rod (for $D=4$ it must be the case that any adjacent rod is a horizon rod). If $I_{a+1}$ is another axis rod then $u= \partial/\partial x^1=0$ at $z=z_a$ and the above metric will have a conical singularity at this endpoint unless
\be
\frac{{h^a}'(z_a)^2}{h^a_{00}(z_a)}=-4c_a^2 \; , \label{regupper}
\ee
in which case the metric extends smoothly at this point. Note that we used ${h^a}'(z_a)= h^a_{00}(z_a){h^a_{11}}'(z_a)$ to simplify the above expression which in turn comes from $h_{\mu1}^a(z_a)  = 0 $ and smoothness. Similarly, if $I_{a-1}$ is an axis rod then $u= \partial/\partial x^1=0$ at $z=z_{a-1}$ and the above metric extends smoothly at this endpoint iff
\be
\frac{{h^a}'(z_{a-1})^2}{ h^a_{00}(z_{a-1})}=-4c_a^2 \; . \label{reglower}
\ee
Therefore,  if $I_a$ is a finite axis rod and provided these regularity conditions are met, the axis metric extends to a smooth Lorentzian metric on $\mathbb{R}\times C_a$. The 2-cycle $C_a$ is topologically $S^2$, a 2-disc or a 2-cylinder depending on if $I_{a-1}, I_a$ are either both axis rods, one axis rod and one horizon, or both horizon rods, respectively. In Appendix \ref{app:corners} we analyse the geometry where two axis rods meet and derive further relations that follow from the above regularity analysis. In particular, we find that for two axis rods $I_a$ and $I_{a+1}$ the function $| z-z_a| e^{2\mathring{\nu}}$ is continuous at $z=z_a$, a result that has been previously proven in~\cite{Tod:2013ska}. 

\subsubsection{Horizons}
\label{sec:horizons}

The analysis of the metric near a horizon is very similar. Consider a component of the horizon, $H_a$, with corresponding rod $I_a$ (again, for simplicity of notation we will temporarily drop the labelling of each rod).  The Killing field null on the horizon is 
\be
\xi= k+ \Omega_i m_i \; , \label{corotKVF}
\ee
where $\Omega_i$ are the angular velocities of the black hole.  Now, working in an adapted basis  for the $D-2$ commuting Killing fields, $\tilde{E}_A= (m_i, \xi)$, the metric can be written as
\be
\tilde{g} = \left( \begin{array}{cc} \gamma_{ij} - \rho^2 \gamma^{-1}\omega_i \omega_j & \gamma^{-1} \rho^2 \omega_{i} \\  \gamma^{-1} \rho^2 \omega_{j} & - \gamma^{-1} \rho^2  \end{array} \right)   \label{ghoradapted}  \; ,
\ee
where $\gamma_{ij}$ is an  an invertible  $(D-3)\times (D-3)$ positive definite matrix with determinant  $\gamma = \det \gamma_{ij}$ (again the normalisation (\ref{detg}) is automatically imposed in this basis). A regular non-degenerate horizon requires $\omega_i, \gamma_{ij}$ to be smooth functions of $(\rho^2, z)$ and
\be
\lim_{\rho \to 0, \;  z\in  I_a} \frac{\rho^2 e^{2\nu}}{| \xi |^2}=  -\frac{1}{\kappa^2}  \label{smoothhor}  \; ,
\ee
where $\kappa \neq 0$ is the surface gravity~\cite{Harmark:2004rm}.

The analysis of the metric induced on the horizon proceeds in an essentially identical fashion to the axis metric analysis above. The inverse metric in this adapted basis is
\be
\tilde{g}^{-1} = \left( \begin{array}{cc}  \gamma^{ij} & \omega^i \\ \omega^j & - \gamma \rho^{-2} + \omega^i \omega_i \end{array} \right)
\ee
where $\gamma^{ij}$ is the inverse matrix of $\gamma_{ij}$ and $\omega^i=\gamma^{ij} \omega_j$. The requirement of a smooth horizon then implies the limits (\ref{U0V0}) exist, which relative to the adapted basis are
\be
\mathring{\tilde{U}} = \left( \begin{array}{cc} 0 &- 2\omega_i \\ 0 & 2  \end{array} \right)  \; , \qquad 
 \mathring{\tilde{V}}= \left( \begin{array}{cc}   (\partial_z \gamma_{i k} ) \gamma^{k j}  & - \gamma \gamma_{ij} \partial_z(\gamma^{-1} \omega^j)  \\  0  &    - (\gamma^{-1} \partial_z \gamma ) \end{array} \right)  \; .  \label{U0V0hor}
\ee
Then the second equation in  (\ref{conf})  integrates to
\be
e^{2\mathring{\nu}}= \frac{\tilde{c}^2}{\gamma}  \label{nu0hor}
\ee
where $\tilde{c}$ is a constant and imposing the smoothness condition (\ref{smoothhor}) gives
\be
\tilde{c}=\kappa^{-1} \label{chor}  \; .
\ee
We deduce that the metric induced on the horizon component $H_a$ associated to the rod $I_a$ is
\be
\mathbf{g}|_{H_a}= \frac{\td z^2}{\kappa^2\gamma(z)} + \gamma_{ij}(z) \td \phi^i \td \phi^j   \; , \label{hormetric}
\ee
which is a $(D-2)$-dimensional smooth Riemannian metric for $z \in I_a$ (recall the axial Killing fields $m_i= \partial_{\phi_i}$).

Given the metric on a horizon $H_a$, one can determine the surface gravity as follows. In general there are conical singularities in the metric (\ref{hormetric}) at the endpoints $z=z_{a-1}, z_a$ and demanding that they are absent  will fix $\kappa$.  For $D=4$ we have $m=\partial_\phi$ vanishing at each endpoint so the condition for no conical singularities is simply
\be
\kappa= \frac{2}{\gamma'(z_{a-1})} =-\frac{2}{\gamma'(z_{a})} \; .  \label{4dkappa}
\ee
In order to fix the sign we have used the fact that ${\gamma}'(z_{a-1})>0$ and ${\gamma}'(z_a)<0$ (these follow from $\gamma>0$ in the interior of $I_a$). Observe this gives two ways of calculating $\kappa$ and hence in principle can provide a nontrivial constraint on the parameters of the solution.
For $D=5$ the adjacent rods $I_{a-1}$ and $I_{a+1}$ are axis rods with rod vectors $v_{a-1}$ and $v_{a+1}$. 
In particular $v_{a-1}=0$ at $z=z_{a-1}$ and $v_{a+1}=0$ at $z=z_a$, so that the horizon metric has conical singularities at the endpoints of $I_a$. The horizon metric extends to a smooth metric at these end points iff the surface gravity
\be
\kappa^2 = \frac{4}{{\gamma}'(z_{a-1}) {\gamma}'_{ij}(z_{a-1})v^i_{a-1}v^j_{a-1} } =\frac{4}{{\gamma}'(z_{a}) {\gamma}'_{ij}(z_{a})v^i_{a+1}v^j_{a+1} }  \; . \label{5dkappa}
\ee
Therefore, again, in principle this gives two independent expressions for $\kappa$ and hence may provide a constraint on the parameters of the solution.  In  Appendix \ref{app:corners} we obtain further relations for the surface gravity by studying the geometry near where an axis rod meets a horizon rod. Similarly to the analysis of a corner between two axes described in the previous section, we find that if an axis rod and horizon rod meet at $z=z_a$ then $|z-z_a| e^{2\mathring{\nu}}$ is continuous at $z=z_a$.

Using (\ref{hormetric}) one can also compute the area of a cross-section of the horizon  
\be
A = \int_{H_a} \kappa^{-1} \td z \td \phi^1 \cdots \td \phi^{D-3}=  \frac{(2\pi)^{D-3} \ell_a}{\kappa} \; ,   \label{area}
\ee
a relation which has been previously derived~\cite{Hollands:2008fm}.

\subsubsection{Standard basis}

In order to compare the solutions on each rod it is useful  to write them in a common basis of Killing fields.   For definiteness we will take a basis adapted to the  semi-infinite rod $I_L$, i.e. the standard basis $E_A=(k, m_1, \dots, m_{D-3})$.   We can relate the adapted bases $\tilde{E}_A$ associated to each rod $I_a$ to the standard basis by $\tilde{E}_A= (L_a^{-1})_A^{~B} E_B$ where $L_a$ is a change of basis matrix.  The metric $\tilde{g}$ in the adapted basis $\tilde{E}_A$,  relative to the standard basis  is thus
\be
g= L_a \tilde{g} L_a^T \; ,  \label{gstandard}
\ee
where $\tilde{g}$ is given by (\ref{gaxisadapted}) or (\ref{ghoradapted}) for an axis rod or horizon rod respectively.

 If  $I_a$ is a horizon rod then $\tilde{E}_A=(m_i, \xi_a)$ where $\xi_a$ is the corotating Killing field (\ref{corotKVF}) for the component of the horizon $H_a$, so
\be
L_a = \left( \begin{array}{cc}  -\Omega^a_j & 1 \\ \delta_j^{i} & 0 \end{array} \right) \; .   \label{Lhorrod}
\ee
On the other hand, now suppose $I_a$ is an axis rod. In 4d there is of course only one axial Killing field and so there is only one type of axis rod and hence the transformation matrices $L_a$ are the identity matrix for all axis rods. In 5d we take the basis $\tilde{E}_A=(k, u_a, v_a)$, where $(u_a,v_a)$ is a basis of $U(1)^2$ Killing fields such that $v_a$ is the rod vector, which gives
\be
L_a = \left( \begin{array}{cc} 1 & 0 \\ 0 &  A_a^{-1} \end{array} \right)
 \; ,  \label{Laxisrod}
\ee
with $A_a$ a $GL(2, \mathbb{Z})$ matrix given by (\ref{uAdef}). 
In particular, in 5d the right semi-infinite rod $I_R$ has rod vector $v_R= m_1$ and choosing $u_R=m_2$  gives 
\be
L_R=  \left( \begin{array}{ccc} 1 & 0  & 0 \\ 0 &0 & 1 \\ 0 & 1 & 0  \end{array} \right)  \; .  \label{LRtrans}
\ee
It is worth noting that for any horizon and axis rods $\det L_a=\pm 1$. Therefore, using (\ref{gstandard}) we deduce that the normalisation (\ref{detg}) is also obeyed in the standard basis.

\subsection{Ernst potentials and gravitational fluxes}

We will need to introduce the following Ernst potentials $b^a_\mu$ associated to each axis rod $I_a$,
\be
\td b^a_\mu = (-1)^{D-1}\tilde{\star} (e_0 \wedge \dots \wedge e_{D-4} \wedge \td e_\mu) \; ,  \label{ernst}
\ee
where $\tilde{E}_A=(e_\mu, v_a)$ is the adapted basis defined above and we fix an orientation $\tilde{\epsilon}_{0 \cdots  D-3  \rho z}>0$. Therefore $\tilde{\star}= (\det L_a) \star$ where $\star$ is the Hodge dual with respect to the standard orientation (defined above) and $L_a$ is the transformation matrix between the adapted basis and the standard basis. Closure of the 1-form on the RHS of (\ref{ernst}) follows by the vacuum Einstein equations and simple connectedness ensures the potentials are globally defined. Explicitly, in Weyl coordinates we have 
\be
\partial_\rho b^a_\mu = \rho \tilde{g}^{D-3 A} \tilde{g}_{A \mu, z} \; , \qquad \partial_z b^a_\mu =-  \rho \tilde{g}^{D-3 A} \tilde{g}_{A \mu, \rho}   \; .
\ee
From the explicit form of the metric in the adapted basis (\ref{gaxisadapted}) it follows that near each axis rod $I_a$
\be
\partial_z b^a_\mu =  2 w_\mu+ O(\rho) \; ,\qquad \partial_\rho b^a_\mu= O(\rho)   \label{baxis}
\ee
as $\rho \to 0$.

The above Ernst potentials associated to each axis rod depend on the corresponding rod vector. For $D=4$ there is only one type of axis rod and the corresponding Ernst potential is simply
\be
\td b = - \star ( k \wedge \td k)\; .  \label{4dernst}
\ee
For $D=5$,  there are many possible axis rods,  although there are two rods which appear in any asymptotically flat solution: the two semi-infinite axis rods $I_L$ and $I_R $ on which $m_2=0$ and $m_1=0$ respectively. The Ernst potentials (\ref{ernst}) associated to $I_L$ and $I_R$ are
\bea
&&\td b^L_{\mu}= \star ( k \wedge m_1 \wedge \td e^L_{\mu}),  \qquad e^L_{\mu}=(k, m_1) \; ,  \label{bL} \\
&& \td b^R_{\mu}= -\star ( k \wedge m_2 \wedge \td e^R_{\mu}) , \qquad e^R_{\mu}=(k, m_2)  \; . \label{bR}
\eea
where the sign in the latter arises from the transformation  (\ref{LRtrans}) between the adapted basis and the standard basis being orientation reversing, $\det L_R=-1$. 

We will also need similar potentials associated to any horizon rod $I_a$.  We define these analogously to the Ernst potentials (\ref{ernst}). Thus, given our adapted basis for a horizon rod $\Tilde{E}_A=(m_i, \xi)$,  these potentials are precisely the usual twist potentials (\ref{twist}) (observe our choice of orientation in these two formulas is consistent).  Therefore, similarly to the Ernst potentials, we find the twist potentials obey 
\be
\partial_\rho \chi_i = \rho \tilde{g}^{0 A} \tilde{g}_{A i, z} \; , \qquad \partial_z \chi_i =-  \rho \tilde{g}^{0 A} \tilde{g}_{A i, \rho}  
\ee
and using  (\ref{ghoradapted}) we find that near a horizon rod $I_a$ 
\be
\partial_z \chi_i=  2 \omega_i+ O(\rho) \; ,\qquad \partial_\rho \chi_i= O(\rho)   \label{chiaxis}
\ee
as $\rho \to 0$.

As shown earlier, the change in twist potential over a horizon rod is related to the Komar angular momenta of the horizon (\ref{Jhorizon}). Similarly, one can relate the change in the Ernst potentials (\ref{ernst}) across their associated axis rods $I_a$ to  certain gravitational fluxes. For $D=4$ we can define the flux
\be
\mathcal{G}[I_a] =-  \int_{I_a} \star (k \wedge \td k)
\ee
for any finite axis rod. Since the integrand is closed by the vacuum Einstein equations these fluxes may be evaluated over any curve homotopic to $I_a$. Clearly, from (\ref{4dernst}) we deduce
\be
\mathcal{G}[I_a]= b(z_a)-b(z_{a-1}),  \label{4dflux}
\ee
which gives a geometric interpretation to the change in Ernst potential over an axis rod. 

Similarly, for $D=5$, given any finite axis rod $I_a$ we may define gravitational fluxes on the corresponding 2-cycle $C_a$.  Explicitly, for each 2-cycle $C_a$ one can define a set of fluxes
\be
\mathcal{G}_{\mu}[C_a]= \frac{1}{2\pi} \int_{C_a} \tilde{\star} ( e_0 \wedge \td e_\mu),
\ee
where $e_\mu=(k, u_a)$, $\mu=0,1$, is our adaped basis of Killing fields on $C_a$ (recall $\tilde{E}_A= (k ,u_a, v_a)$ is the adapted basis of Killing fields in the full spacetime).  The integrand is closed by the vacuum Einstein equations so one can evaluate these fluxes over any 2-surface homologous to $C_a$ so it only depends on the homology class $[C_a]$.  Thus these fluxes define  gravitational topological charges.  Due to the invariance under the Killing fields these integrals can be reduced to ones over the corresponding axis rods,\footnote{Here we are using the identity $\int_{C_a} \tilde{\star} \alpha =-2\pi \int_{I_a} \tilde{\star} ( u_a \wedge \alpha)$, valid for any $U(1)^2$-invariant 3-form $\alpha$. This also shows that $\int_{C_a} \tilde{\star} (e_1 \wedge \td e_\mu)=0$ so that these quantities do not give rise to new charges. } 
\be
\mathcal{G}_{\mu}[C_a]= \int_{I_a} \tilde{\star} (e_0 \wedge  e_1 \wedge \td e_\mu) =  b_\mu^a(z_a)-b_\mu^a(z_{a-1})  \; , \label{5dflux}
\ee
where we have used the definition of the Ernst potentials (\ref{ernst}).  Thus we see that the fluxes $\mathcal{G}_{\mu}[C_a]$ precisely correspond to the change in the Ernst potential $b^a_\mu(z)$ over the associated axis rod $I_a$ giving it a geometric interpretation. A similar set of topological charges have appeared in recent identities that relate the thermodynamic variables to the topology of solutions in this class~\cite{Kunduri:2018qqt}.

Finally, it is worth noting  that one can also relate the changes in Ernst potentials $b^a_\mu(z)$ (\ref{ernst}) over a horizon rod $I_a$ to the standard thermodynamic quantities. We give these expressions in Appendix \ref{apen:ernstHorizon}.

\section{Integrability of Einstein equations}
\label{sec:lax}

\subsection{Belinski-Zakharov spectral equations}
As shown by Belinski and Zakharov (BZ), the Einstein equations (\ref{einsteineq}) are the integrability conditions for the following auxiliary linear system~\cite{BZ1, BZ2},
\be
D_z \Psi = \frac{\rho V - \mu U}{\mu^2+\rho^2} \Psi, \qquad D_\rho \Psi = \frac{\rho U+ \mu V}{\mu^2+\rho^2} \Psi  \; , \label{BZLP}
\ee
where 
\be
D_z= \partial_z-\frac{2\mu^2}{\mu^2+\rho^2} \partial_\mu, \qquad D_\rho= \partial_\rho+\frac{2\mu \rho}{\mu^2+\rho^2} \partial_\mu
\ee
are commuting differential operators, $\mu$ is a complex `spectral' parameter and $\Psi$ is an invertible $(D-2)\times (D-2)$ complex matrix function of $(\rho, z, \mu)$. 

We will work with a slightly different version of the BZ linear system~\cite{MY, BZM, Varzugin:1997ee}.  This is obtained by a change of spectral parameter defined by the coordinate change $(\rho, z, \mu) \to (\rho, z, k)$ where 
\be
k = z+ \frac{\mu^2-\rho^2}{2\mu} \; ,  \label{k}
\ee
which in particular implies $D_z \to \partial_z, D_\rho\to \partial_\rho$. This results in the linear system 
\be
\partial_z \Psi = \frac{\rho V - \mu U}{\mu^2+\rho^2} \Psi, \qquad \partial_\rho \Psi = \frac{\rho U+ \mu V}{\mu^2+\rho^2} \Psi  \; , \label{LP}
\ee
where $\mu=\mu(k)$ is defined implicitly by (\ref{k}) and $k$ is the new complex spectral parameter. We will assume $\Psi$  is a smooth function of $(\rho, z)$ and meromorphic in $k$ (in a suitable domain).   Henceforth we will work exclusively with this alternate form of the BZ linear system (\ref{LP}). It turns out to be more useful for our purposes since, as (\ref{k}) shows, the spectral parameter $k$ is defined on a two-sheeted Riemann surface.

Independently of (\ref{BZLP}), one can check directly from (\ref{LP})  that $\partial_z\partial_\rho \Psi= \partial_\rho \partial_z \Psi$ iff
\be
\partial_\rho \left(\frac{V}{\rho} \right)- \partial_z\left(\frac{U}{\rho} \right)-  \frac{1}{\rho^{2}}[ U,V ]=0   \label{flat}
\ee
and
\be
\partial_\rho \mu = \frac{2 \rho \mu}{\mu^2+\rho^2} , \qquad \partial_z \mu = - \frac{2\mu^2}{\mu^2+\rho^2}   \label{mupde}
\ee
and the Einstein equations (\ref{einsteineq}) are satisfied. Equation (\ref{flat}) is in fact identically satisfied as it is  the integrability condition for the existence of a matrix $g$ such that (\ref{UV}), whereas the general solution to (\ref{mupde}) is given by (\ref{k}) where $k$ is the integration constant. For some purposes it will be convenient to write the linear system in the equivalent form
\be
(\rho \partial_\rho - \mu\partial_z) \Psi = U \Psi, \qquad  (\mu \partial_\rho +\rho \partial_z) \Psi = V \Psi. \label{LPalt}
\ee
In particular, this form will be useful when evaluating on the boundary of the half-plane.

Although solving for $\Psi$ in general is complicated, it is straightforward to solve for the general form of $\det \Psi$.  Right multiplying (\ref{LP}) by $\Psi^{-1}$ and taking the trace gives
\be
\partial_\rho \det \Psi = \frac{2\rho  \det \Psi}{\mu^2+\rho^2}, \qquad \partial_z \det \Psi = -\frac{2\mu \det\Psi}{\mu^2+\rho^2}\; ,
\ee
where we have used $\text{Tr}\,  U=2$ and $\text{Tr} \,V=0$.
Comparing to (\ref{mupde}) it follows that 
\be
\det \Psi = \mu f(k),   \label{det}
\ee
where $f(k)$ is an arbitrary function of $k$ (i.e independent for $\rho,z$).

As we will take $k$ to be a complex parameter we need to take care to treat the implicitly defined function $\mu$ in (\ref{k}) properly. Locally, we may solve for $\mu$ to get
\be
\mu =k-z\pm \sqrt{\rho^2+(k-z)^2}   \; .  \label{mu}
\ee 
Thus there are branch points at $k= w$ and $k = \bar{w}$ where $w= z+i \rho$ and so we take the branch cut to be the finite line in the complex $k$-plane between these points. Hence we consider the linear system (\ref{LP}) on the two sheeted Riemann surface $\Sigma_w \subset \mathbb{C}^2$ defined by 
\be
y^2= (k- w)(k-\bar{w}) \; ,  \qquad (k, y) \in \mathbb{C}^2 \; .
\ee
The square root function (\ref{mu}) is then defined by $\mu : \Sigma_w \to \mathbb{C}$ where
\be
\mu(k, y) = k- z + y 
\ee
 We will denote $y$ on the two sheets (i.e. the two square roots) by $y_\pm(k)$ and use $k$ as a local coordinate on each sheet. For definiteness we define $y_+$ by having positive real part for $\text{Re} \, (k-w)>0$. We also define $\mu_\pm = \mu(k, y_\pm)$ and note the useful identity $\mu_+ \mu_- = - \rho^2$. 
  
We will also denote the corresponding $\Psi$ on the two sheets by $\Psi_\pm$ and similarly for any other quantity on $\Sigma_w$. Since $\Psi_\pm$ corresponds to $\Psi$ evaluated on the two sheets of the same Riemann surface we must require a continuity condition at the branch points:
\be
\Psi_+(\rho, z, k)= \Psi_-(\rho, z, k)  \qquad \text{at} \quad k= z\pm i \rho  \; .   \label{bpcont}
\ee
This condition will be important in our later analysis. Taking the determinant of this and comparing to (\ref{det}) shows that $f_+(k)= f_-(k)$ (for $\text{Im} \; k \neq 0$, and by  continuity, for all $k$ except perhaps at isolated points) and so we drop the subscript on this quantity.

The spectral equations have an important involution symmetry which allow one to map solutions on one Riemann sheet to the other. The matrices defined by
\be
\tilde{\Psi}_{\pm} = g \Psi_{\mp}^{T-1}  \; ,  \label{tildedef}
\ee
obey the same equations as $\Psi_{\pm}$, i.e.,
\be
(\rho \partial_\rho - \mu_\pm \partial_z) \tilde{\Psi}_{\pm} = U \tilde{\Psi}_{\pm} , \qquad  (\mu_{\pm} \partial_\rho +\rho \partial_z) \tilde{\Psi}_{\pm} = V \tilde{\Psi}_{\pm}  \; .
\ee
It is easy to show that given two solutions $\Psi_\pm$ and $\tilde{\Psi}_\mp$ to the above equations their `difference'  $B_{\pm}=\tilde{\Psi}_\pm^{-1} \Psi_\pm$ must be independent of $(\rho, z)$. Therefore, it follows from (\ref{tildedef}) that
\be
\label{symBpm}
\Psi_\pm = g \Psi_{\mp}^{T-1}B_\pm    \; ,
\ee
where $B_\pm=B_\pm(k)$ are invertible matrices. It immediately follows from this that $B_\pm = B_\mp^T$.  Furthermore, for $\rho>0$ we can write (\ref{symBpm}) as $B_\pm= \Psi^T_\mp g^{-1} \Psi_\pm$ and evaluating this at the branch points $k=z\pm i \rho$ and using the continuity condition (\ref{bpcont}) shows that $B_\pm(k)$ is symmetric (for $\text{Im}\; k\neq 0$, and by continuity, for all $k$ except perhaps at isolated points). Putting all this together we deduce that $B_+= B_-^T=B_-$ so we may drop the subscript on $B$.   Thus this symmetry may be simply written as
\be
\Psi_\pm = g \Psi_{\mp}^{T-1}B \label{sym}
\ee
where $B=B(k)$ is an invertible symmetric matrix. Taking the determinant shows
\be
\det B(k) =     f(k)^2 \; .  \label{detB}
\ee

\subsection{Spectral equations on semi-circle at infinity}
\label{subsec:infinity}

We will consider asymptotically flat spacetimes in four and five dimensions.  In both cases the asymptotic region corresponds to the semi-circle at infinity in the half-plane (\ref{HP}). Thus it is convenient to introduce polar coordinates $(r, \theta)$ on the half-plane where
\be
\rho = r \sin \theta, \qquad z = r \cos\theta   \label{polarcoords}
\ee
and $0\leq \theta \leq \pi$.  In terms of the complex coordinate $w=z+ i \rho$ we have $w= r e^{i \theta}$. The semi-circle at infinity then simply corresponds to $r  \to \infty$.  More precisely, we introduce the contour $C_r = \{ r e^{i \theta} : 0 \leq \theta \leq \pi \}$ in the half-plane with anticlockwise orientation and consider large $r$.

Now, fix a sheet of $\Sigma_w$ with local coordinate $k$ and consider traversing $C_r$ starting at $\theta=0$.  The branch points $w, \bar{w}$ trace out corresponding semi-circles in the upper and lower half of the complex $k$-plane with a moving branch cut between the upper and lower semi-circle.  Any fixed value of $k$ on the sheet must pass through the  moving branch cut as we traverse $C_r$ for large enough $r$ (i.e. $r>|k|$). This occurs at an angle given by $\text Re (k-w)=0$,  i.e. 
\be
\cos \theta_* = \frac{\text{Re} (k) }{r}  \implies  \theta_*= \frac{\pi}{2}  - \frac{\text{Re} (k) }{r} + O(r^{-3})  \; .
\ee
Now, passing through the branch cut corresponds to changing sheet of $\Sigma_w$. Therefore, in effect, traversing $C_r$ imposes a change of sheet as we pass through $\theta= \theta_*$. In particular, given a solution to the linear system $\Psi_\pm (r, \theta, k)$ on the two sheets, this implies the following continuity conditions on the semi-circle at infinity 
\be
\lim_{\epsilon \to 0^+} \Psi_\pm(r, \theta_*- \epsilon, k) = \lim_{\epsilon \to 0^+} \Psi_\mp( r, \theta_*+\epsilon,k)  \; .   \label{BCinfty}
\ee
Notice this provides a relation between the $\Psi_+$ and $\Psi_-$ fields at infinity.

 The above considerations  also affect the asymptotic expansion of quantities defined on each sheet along infinity.  For instance, consider $\mu_+$ on the $+$ sheet. Traversing $C_r$ from $\theta=0$, it is easy to see that the branch cut approaches a fixed $k$ from the right (where $y_+(k)$ has negative real part) so
\be
\mu_+(r, \theta, k) = (k-r) (1+ \cos \theta) + O(r^{-1})  \qquad 0 \leq \theta < \theta_*  \; ,
\ee
whereas traversing $C_r$ from $\theta = \pi$, the branch cut approaches $k$ from the left so
\be
\mu_+(r, \theta, k) = (k+r) (1- \cos \theta) + O(r^{-1})  \qquad \theta_*  < \theta \leq \pi  \; .   \label{asympmu}
\ee
A similar argument for $\mu_-$ shows that
\bea
\mu_- (r, \theta, k) = \left\{  \begin{array}{c} (k+ r) (1- \cos \theta) + O(r^{-1})  \qquad 0 \leq \theta < \theta_*  \\ 
 (k-r) (1+  \cos \theta) + O(r^{-1})  \qquad \theta_*  < \theta \leq \pi  \end{array}  \right.  \; .
\eea
Observe that the continuity conditions  $\lim_{\epsilon \to 0^+} \mu_\pm(r, \theta_*- \epsilon, k) = \lim_{\epsilon \to 0^+} \mu_\mp(r, \theta_*+\epsilon,k)$ are indeed satisfied.

It is convenient to write our  linear system (\ref{LP}) in  polar coordinates, which gives, 
\bea
&&\partial_r \Psi = Y_r \Psi,  \qquad   Y_r = \frac{r\sin^2 \theta S -  \mu T }{ \mu^2+ r^2 \sin^2\theta}  \; \\ 
 &&\partial_\theta \Psi = Y_\theta \Psi \;, \qquad Y_\theta = \frac{r\sin \theta(\mu  S + r T)}{\mu^2+r^2\sin^2\theta} \label{LPpolar}
\eea
where $S = r \partial_r g g^{-1}$ and $T=\sin \theta \partial_\theta gg^{-1}$.  We now consider the solution to the spectral equations in the limit $r\to \infty$.

The explicit solution depends on the dimension, although it has some common features which will be key in our analysis. Let $\bar{g}$ denote the Minkowski metric and $\bar{\Psi}$ a corresponding solution to the spectral equation (\ref{LPpolar}).  Now define the `difference',
\be
\Delta = \bar{\Psi}^{-1} \Psi  \; ,
\ee
 between a Minkowski solution $\bar{\Psi}$ and a solution $\Psi$ to (\ref{LPpolar}) for any asymptotically flat metric $g$.
Then, it easily follows that
\bea
&&(\partial_r \Delta )\Delta^{-1} =\Upsilon_r , \qquad \Upsilon_r \equiv   \bar{\Psi}^{-1} (Y_r-\bar{Y}_r) \bar{\Psi},  \nonumber \\ 
 &&(\partial_\theta \Delta )\Delta^{-1}=\Upsilon_\theta\; , \qquad \Upsilon_\theta \equiv \bar{\Psi}^{-1} (Y_\theta-\bar{Y}_\theta) \bar{\Psi}  \; .  \label{Deltaeq}
\eea
The matrices $\Upsilon$ depend on the explicit solution in Minkowski spacetime and the definition of asymptotic flatness, which for $D=4,5$ will be given later.  All that we need at this stage is that for both dimensions, all matrix entries of $\Upsilon_r$ and $\Upsilon_\theta$ are $O(r^{-2})$ and $O(r^{-1})$ respectively, as $r\to \infty$.  Thus, asymptotically, $\Delta$ must be only a function of $k$. In other words, the solution to the spectral equation for an asymptotically flat spacetime is asymptotic to that for Minkowski spacetime, as one would expect.

More precisely, consider the solution on the $+$ sheet of $\Sigma_w$
\be
\Psi_+ = \bar{\Psi}_+ \Delta_+  \; . \label{Psiinfty}
\ee
From the above it follows that
\be
\Delta_+ =\left\{ \begin{array}{cc} N_R(k)+ O(r^{-1}) \qquad 0 \leq \theta < \theta_* \\  N_L(k)+ O(r^{-1})   \qquad \theta_*  < \theta \leq \pi \end{array} \right.  \; ,   \label{Deltainfty}
\ee
where $N_{R,L}(k)$ are invertible matrices and  $R, L$ denote the right and left segment (these in general are different since  $\Upsilon_{r+}, \Upsilon_{\theta+}$ are discontinuous on $C_r$ at $\theta=\theta_*$).
Using the involution symmetry (\ref{sym}) we find that 
\be
\Psi_- = g \bar{\Psi}_+^{T-1} \Delta_+^{T-1} B  \; ,
\ee
and hence imposing the continuity conditions (\ref{BCinfty}) we deduce that
\bea
C &\equiv& N_R^{T-1}(k) B(k) N_L(k)^{-1}  \label{Cdef} \\
 &=& \lim_{r \to \infty} \bar{\Psi}^T_+(r, \theta_*^-, k) g(r, \theta_*)^{-1} \bar{\Psi}_+(r, \theta_*^+, k)  \; .  \label{C}
\eea
The relation (\ref{C}) allows one to compute $C$ given the asymptotics of the Minkowski solution. It is worth remarking that although (\ref{BCinfty}) consists of two continuity equations, the fact that $B$ is a symmetric matrix (\ref{sym}) ensures that they are equivalent.

There is a certain freedom in the choice of $\bar{\Psi}_+$ corresponding to right-multiplication by a matrix function of $k$. Since the asymptotic expansion (\ref{asympmu}) for $\mu_+$ to leading order is independent of $k$, we may choose $\bar{\Psi}_+(r, \theta, k)$ such that as $r \to \infty$ the leading term in each entry  is independent of $k$.  Making this choice, one then expects from (\ref{C}) that $C$ is independent of $k$ and hence is a constant matrix (we will confirm this explicitly later).

\subsection{General solution on the axes and horizons}

We will now show that the linear system simplifies when evaluated on the boundary of the half-plane. Recall that smoothness of the axes and horizons requires the metric must be a smooth function of $(\rho^2, z)$. Therefore we may assume $\Psi$ is a smooth function of $(\rho^2, z)$. 

First we make a few general remarks. In order to evaluate limits to the boundary we will need the following useful relations
\be
\mu_+ \sim 2 (k-z), \qquad \mu_- \sim -\frac{\rho^2}{2(k-z)}
\ee
as $\rho \to 0$.  Thus taking the limit of the determinant $\det \Psi_\pm$ and using (\ref{det}) shows that $\Psi_+$ is generically a nonsingular matrix on the boundary whereas $\Psi_-$ is singular. Therefore we will only consider $\Psi_+$ and use (\ref{sym}) to deduce $\Psi_-$.

We are now in a position to evaluate the limit of the linear system (\ref{LPalt}) for $\Psi_+$ as $\rho \to 0$. It is easy to see this system reduces to an ODE
\be
(z-k)\partial_z \mathring{\Psi} = \tfrac{1}{2}\mathring{U} \mathring{\Psi},  \label{LPbdry}
\ee
where we define $\mathring{\Psi}(z, k) = \lim_{\rho \to 0} \Psi_+(\rho, z, k)$ and the second equation  vanishes identically due to our assumption that $\Psi_+$ is a smooth function of $\rho^2$. We will  explicitly solve the linear system along the boundary $\rho=0$.

First consider an axis rod $I_a$. In the corresponding adapted basis the metric is given by (\ref{gaxisadapted}). The general solution to the linear system (\ref{LPbdry}) on $I_a$ in this basis can be written as
\be
 \tilde{X}_a(z, k) \tilde{M}_a(k), \qquad \tilde{X}_a(z, k) = \left( \begin{array}{cc} -\delta_{\mu}^{~\nu} & b^a_\mu(z) \\ 0 & 2(k-z) \end{array} \right), \quad z \in I_a   \label{LPaxissoladap}
\ee
where we have used (\ref{UVtilde0}, \ref{baxis}) and $\tilde{M}_a(k)$ is an arbitrary integration matrix. The particular solution $\tilde{X}_a(z, k)$ satisfies 
\be
\partial_z \tilde{X}_a= - \mathring{\tilde{U}}_a.
\ee
We note there is a lot of freedom in the choice of particular solution $\tilde{X}_a(z, k)$. In particular, the integration constant for the Ernst potential $b^a_\mu(z)$ may be set to any value we like by right multiplying the particular solution by a constant upper triangular matrix with unit diagonals (which can then be absorbed into a redefinition of $\tilde{M}_a(k)$). For convenience we will choose the potentials to vanish at the lower endpoint of the finite rods 
\be
b^a_{\mu}(z_{a-1})=0   \label{ErnstBC}
\ee
for $a=2, \dots, n$ and 
\be
\lim_{z\to -\infty} b^L_{\mu}(z) = 0 \; ,   \qquad  \lim_{z\to \infty} b^R_{\mu}(z)=0 \; .  \label{bLRinfty}
\ee
The latter are consistent with the asymptotics $b^L_\mu \to 0$ and $b^R_\mu\to 0$ at infinity (even off axis).

In order to compare the solutions on each rod we will write them all relative to the standard basis. The metric near each axis rod $I_a$ relative to the standard basis is given by (\ref{gstandard}),  which implies $U= L_a \tilde{U} L_a^{-1}$.  Hence, from (\ref{LPaxissoladap}), we deduce that the general solution to the linear system (\ref{LPbdry}) on an axis rod $I_a$ relative to the standard basis takes the form
\be
\mathring{\Psi}_a(z, k)= X_a(z, k) M_a(k) \;,  \quad z \in I_a \; ,\label{LPgensol} 
\ee
where 
\be
X_{a}(z, k) =L_a \left( \begin{array}{cc} -\delta_{\mu}^{~\nu} & b^a_{\mu}(z) \\ 0 & 2(k-z) \end{array} \right) L_a^{-1}
\label{Xaxis}
\ee
and  $M_a(k)$ are arbitrary matrices. It is also worth recording that the metric on $I_a$ relative to the standard basis (\ref{gstandard}) is simply
\be
\mathring{g}(z) = L_a \left( \begin{array}{cc} h_{\mu\nu}^a(z) & 0 \\ 0 & 0 \end{array} \right) L_a^{T}  \; . \label{gaxis}
\ee
Recall that in these formulas, if $D=4$ the matrix $L_a$ is the identity matrix, whereas if $D=5$ it is given by  (\ref{Laxisrod}).

Now consider a horizon rod $I_a$. An entirely analogous derivation of the solution can be given in this case using (\ref{U0V0hor}, \ref{chiaxis}).  Thus we find the general solution to the linear system (\ref{LPbdry}) on a horizon rod $I_a$ relative to the standard basis can be again written as (\ref{LPgensol}) where 
\bea
X_{a}(z, k) =L_a \left( \begin{array}{cc} -\delta_{i}^{~j} & \chi^a_i(z) \\ 0 & 2(k-z) \end{array} \right) L_a^{-1}
\label{Xhor}
\eea
and $\chi^a_i(z)= \chi_i(z)- \chi_i(z_{a-1})$ (which corresponds to a choice of integration constant), the matrix $L_a$ is given by (\ref{Lhorrod}) and $M_a(k)$ are arbitrary matrices. The metric on $I_a$  relative to the standard basis (\ref{gstandard}) is simply
\be
\mathring{g}(z) = L_a \left( \begin{array}{cc} \gamma_{ij}(z) & 0 \\ 0 & 0 \end{array} \right) L_a^{T}  \; .  \label{ghor}
\ee
We  now have the general solution to the linear system on all components of the boundary $\rho=0$.

Before moving on it is worth noting that for both axis and horizon rods we have
\be
\det X_a(z,k)= 2 (-1)^{D-3} (k-z)   \label{detX}
\ee
and combining this with (\ref{det}) implies
\be
\det M_a(k) =(-1)^{D-3} f(k) \; ,  \label{detM}
\ee
for all $a=1, \dots, n+1$.

  Clearly we must impose continuity of $\mathring{\Psi}(z,k)$ at $z=z_a$ for $a= 1, \dots, n$, where adjacent rods $I_a$ and $I_{a+1}$ touch, i.e.,
\be
\mathring{\Psi}_{a}(z_a, k)= \mathring{\Psi}_{a+1}(z_a, k)   \; .
\ee
Upon using the general solution this gives
\be
M_a(k) = P_a(k) M_{a+1}(k)
\ee
where we have introduced the matrices
\be
P_a(k)= X_a(z_a, k)^{-1}X_{a+1}(z_a,k)   \; , \label{Padef}
\ee
for each $a=1, \dots n$.  Observe that from (\ref{detX}) it follows  that $\det P_a(k)=1$ automatically.
Iterating we find
\bea
&&M_a(k) = Q_{a}(k)  M_{R}(k), \label{MaIdentity}\\ 
&&Q_a(k) \equiv  P_a(k)P_{a+1}(k) \cdots P_n(k)  \label{Qadef}
\eea
for $a=1, \dots, n+1$ with $Q_{n+1}(k)$ understood as the $(D-2)$-dimensional identity matrix. In particular
\be
M_L(k) = Q_1(k) M_R(k) \; . \label{MLR2}
\ee
Note the fact $P_a(k)$ is unit determinant implies $\det Q_a(k)=1$ automatically.

We may now match the solution on the semi-infinite axis rods to the solution for an asymptotically flat spacetime near infinity (\ref{Psiinfty}) and (\ref{Deltainfty}). Firstly, the solutions for Minkowski spacetime on the semi-infinite axes can be deduced from the above by setting $b^{L,R}_\mu(z)=0$. A convenient choice,  such that these solutions are independent of $k$ to leading order as $|z| \to \infty$, is
\be
\mathring{\bar{\Psi}}_L (z, k) =  \left( \begin{array}{cc} -\delta_{\mu}^{~\nu} & 0 \\ 0 & 2 (k-z) \end{array} \right), \qquad \mathring{\bar{\Psi}}_R (z, k)=  L_R \left( \begin{array}{cc} -\delta_{\mu}^{~\nu} & 0 \\ 0 & 2(k-z) \end{array} \right) L_R^{-1}  \; .  \label{barPsiaxis}
\ee
Thus from (\ref{Psiinfty}) we get
\be
\mathring{\Delta}_{+L}  = M_L(k)+ O(z^{-1}), \qquad \mathring{\Delta}_{+R}  = M_R(k)+ O(z^{-1}) \; ,  \label{Deltaaxis}
\ee
where we have used (\ref{bLRinfty}) and further assumed   the asympotic expansion for $b_\mu^{L,R}(z)=O(z^{-1})$ (this follows from the definition of asymptotic flatness as we will see later). Therefore, comparing to (\ref{Deltainfty}) we deduce that
\be
N_{R}(k)= M_{R}(k)  \; , \qquad N_{L}(k)= M_{L}(k)  \; .
\ee
We may use this to eliminate the matrices $N_{L/R}$ in favour of $M_{L/R}$ and thus from (\ref{Cdef}) we obtain
\be
M_{L}= C^{-1} (M_{R})^{T-1} B  \; .  \label{MLR}
\ee
Recall that the choice of asymptotic solutions corresponds to a choice of  matrix $C$ (\ref{C}). Later we will see that our choice (\ref{barPsiaxis}) fixes $C$ to be a dimension dependent constant matrix.  In any case, taking the determinant of (\ref{MLR}) and using (\ref{detB}) and (\ref{detM}) implies
\be
\det C=1   \label{detC}
\ee
independently of the dimension.

It is convenient to define the following matrix
\be
\tilde{Q}_1(k) = C Q_1(k)  \;  .  \label{Qtilde}
\ee
We are now ready to state our first result.
\begin{prop}
\label{prop:sym} The matrices
\be
\label{Fadef}
F_a(k) = - Q_a(k) \tilde{Q}_1(k)^{-1} Q_a(k)^T  \;,
\ee
are symmetric for $a=1, \dots, n+1$.
\end{prop}

\begin{proof} Clearly, if $\tilde{Q}_1(k)$ is symmetric then $F_a(k)$ is also symmetric for all $a=1, \dots, n+1$.  Combining the condition arising from asymptotic flatness (\ref{MLR}) with the continuity condition (\ref{MLR2})  gives
\be
M_R(k)  B(k)^{-1}  M_R(k)^{T}= \tilde{Q}_1(k)^{-1}  \label{Q1tilde}
\ee
which immediately implies the result (recall $B$ is symmetric).  Symmetry of $F_a(k)$ for $a=1, \dots, n$,  also directly follows from the relation
\be
M_a(k)  B(k)^{-1}  M_a(k)^{T}= - F_a(k)   \; .  \label{Fa}
\ee
which can be established by combining (\ref{Q1tilde}) with (\ref{MaIdentity}).
\end{proof}

\noindent {\bf Remarks.}
\begin{enumerate}
\item  The matrices $F_a$ can be rewritten explicitly in terms of $P_a(k)$ to give
    \begin{equation}
    \begin{aligned}
        F_L &= - C^{-1} Q_1^T = -C^{-1} P_n^T \cdots P_1^T,\\
        F_a &= -P_{a-1}^{-1} \cdots P_1^{-1} C^{-1} P_n^T \cdots P_a^T
        ,\\   \label{Faalt}
        F_R &= - \tilde{Q}_1^{-1} = - P_n^{-1} \cdots P_1^{-1}C^{-1},
    \end{aligned}
    \end{equation}
    where $a = 2,\dots, n$.
\item In general the determinant of $F_a$ is
\be
\det F_a(k)=(-1)^{D-2}   \label{detFa}
\ee
as a consequence of $Q_a(k)$ being unit determinant and (\ref{detC}).
\end{enumerate}

We are now ready to state the main result of this section.
\begin{prop}
The metric data on each rod  satisfies the algebraic equation
\be
\mathring{g}(z)= X_a(z,z) F_a(z) \; , \qquad z \in I_a  \label{g0sol}
\ee
where $F_a(z)$ is given by (\ref{Fadef}), whereas $\mathring{g}(z)$ and $X_a(z, z)$ are given by (\ref{gaxis}), (\ref{Xaxis}) for an axis rod and (\ref{ghor}), (\ref{Xhor}) for a horizon rod.
\end{prop}

\begin{proof}  We impose continuity at the branch points (\ref{bpcont})  on the axis $\rho=0$:
\be
\lim_{k \to z} {\Psi}_+(0, z, k) =\lim_{k \to z} {\Psi}_{-}(0, z, k)  \; . \label{bpaxiscont}
\ee
Using (\ref{sym}) to write $\Psi_{-}$ in terms of $\Psi_{+}$, the continuity condition (\ref{bpaxiscont}) reads
\be
\mathring{\Psi}(z,z) = \lim_{k\to z} \mathring{g}(z) \mathring{\Psi}(z,k)^{T-1} B(k)  \; .  \label{bpaxis}
\ee
Evaluating on each rod and using the general solution (\ref{LPgensol}), equation (\ref{Fa}) and the elementary  identity  $\mathring{g}(z) X_a(z, k)^{T-1}= - \mathring{g}(z)$, gives  (\ref{g0sol}) as claimed.
\end{proof}
We emphasise that, crucially, equation (\ref{g0sol})  does not depend on the arbitrary matrices $M_a(k)$ and hence  provides a constraint on the spacetime geometry. In fact, (\ref{g0sol})  fully determines the metric on each rod $I_a$. Indeed, both $\mathring{g}(z)$ and $X_a(z,z)$ for $z\in I_a$ are rank-$(D-3)$ so (\ref{g0sol}) gives $\tfrac{1}{2}(D-3)(D-2)+ D-3$ algebraic equations for the $\tfrac{1}{2}(D-3)(D-2) +D-3$ unknowns, either $(h^a_{\mu\nu}(z), b^a_\mu(z))$ or $(\gamma_{ij}(z), \chi_i(z))$ (depending on if $I_a$ is an axis or horizon rod).

\subsection{Classification theorem and moduli space of solutions}
\label{sec:modspace}

We now show that (\ref{g0sol}) fully determines the metric on each rod. The explicit solution is summarised by the following theorem which is the main result of this paper.

\begin{theorem} \label{th:main}
 Consider a $D=4$ or $5$-dimensional vacuum spacetime as in Theorem \ref{th1}.
\begin{enumerate}
\item The general solution $(h^a_{\mu\nu}(z), b^a_\mu(z))$ on any axis rod $I_a$ is 
\be
h^a_{\mu\nu}(z) =-\tilde{F}_{a\mu\nu}(z) +\frac{\tilde{F}_{a\mu N} (z) \tilde{F}_{a N \nu}(z)}{\tilde{F}_{a NN}(z)}, \qquad b^a_\mu(z)= \frac{\tilde{F}_{a \mu N} (z)}{\tilde{F}_{a NN}(z)}  \; , \label{gensolaxis}
\ee
where $\mu=0,\dots, D-4$ and $N=D-3$ and the matrices $\tilde{F}_a(k)$ are defined by 
\be
F_a(k) =L_a \left( \begin{array}{cc} \tilde{F}_{a\mu\nu}(k) & \tilde{F}_{a\mu N}(k) \\ \tilde{F}_{a N \nu}(k) & \tilde{F}_{a N N}(k) \end{array} \right)L_a^T  \; .    \label{Facompaxis}
\ee
In particular, this implies
\be
\det h^a_{\mu\nu}(z) = -\frac{1}{\tilde{F}_{aNN}(z)}  \label{detaxis}
\ee
and
\be
\tilde{F}_{a NN}(z)>0 \; , \qquad z \in I_a \; .   \label{FNNpos}
\ee
\item The general solution $(\gamma_{ij}(z), \chi^a_i(z))$ on any horizon rod $I_a$ is 
\be
\gamma_{ij}(z) =-\tilde{F}_{aij}(z) +\frac{\tilde{F}_{a i 0} (z) \tilde{F}_{a 0 j}(z)}{\tilde{F}_{a 00}(z)}, \qquad \chi^a_i(z)= \frac{\tilde{F}_{a i 0} (z)}{\tilde{F}_{a 00}(z)}  \; ,\label{gensolhor}
\ee
where $i=1, \dots, D-3$ and $\tilde{F}_a(k)$ is defined by
 \be
F_a(k) =L_a \left( \begin{array}{cc} \tilde{F}_{aij}(k) & \tilde{F}_{a i 0}(k) \\ \tilde{F}_{ a 0 j}(k) & \tilde{F}_{a 00}(k) \end{array} \right)L_a^T  \; .   \label{Facomphor}
\ee
In particular, 
\be
\det \gamma_{ij}(z) = -\frac{1}{\tilde{F}_{a 00}(z)}  \label{dethor}
\ee
and 
\be
\tilde{F}_{a 00}(z)<0 \;, \qquad z \in I_a \; .   \label{F00neg}
\ee
\end{enumerate}
In both cases $F_a(k)$ are the matrices defined by (\ref{Fadef}). The solution depends on the `moduli'
\begin{equation}
    \{b^L_\mu(z_1),b_\mu^R(z_n)\} \cup \{(\ell_a,v_a,b_\mu^a(z_a)|I_{a\neq L,R} \subset \hat{A}\} \cup \{(\ell_a,\Omega_i^a,\chi_i^a(z_a)|I_a \subset \hat{H}\},\label{moduli}
\end{equation}
where $\hat{A}$ and $\hat{H}$ are the union of axis and horizon rods respectively, subject to algebraic constraints arising from Proposition \ref{prop:sym} 
and the inequalities (\ref{FNNpos}), (\ref{F00neg}).
\end{theorem}

\begin{proof}
First consider an axis rod $I_a$ and let us write $F_a(k)$ as (\ref{Facompaxis}). Then, using (\ref{gaxis}) and (\ref{Xaxis}) reveals that (\ref{g0sol}) is equivalent to  $h^a_{\mu\nu}=- \tilde{F}_{a\mu\nu}+ b^a_\mu \tilde{F}_{a N\nu}$ and $\tilde{F}_{a \mu N}=F_{ a NN} b^a_\mu$.  We can solve this for $b^a_\mu= \tilde{F}_{a \mu N} /\tilde{F}_{a NN}$, since  $\tilde{F}_{aNN}\neq 0$ for any $z \in I_a$; to see this latter condition simply note that if $\tilde{F}_{aNN}=0$ then $\tilde{F}_{a\mu N}=0$ which contradicts the fact $F_a(k)$ must be unimodular (\ref{detFa}).  Thus we find the unique solution on an axis rod $I_a$ is (\ref{gensolaxis}). Then, recalling that $\det L_a=\pm 1$ for any rod,  (\ref{detFa}) implies that  (\ref{detaxis}). Finally, since $h^a_{\mu\nu}(z)$ must be a smooth Lorentzian metric on $I_a$ we must require (\ref{FNNpos}).

A completely analogous analysis holds for any horizon rod $I_a$, with the only difference being that $\gamma_{ij}(z)$ must be a smooth positive definite metric on $I_a$ so we must require (\ref{F00neg}).

The matrices $F_a(k)$ are given by (\ref{Faalt}), where the matrices $P_a(k)$ are defined by (\ref{Padef}). From the explicit form for $X_a(z, k)$ on each axis rod (\ref{Xaxis}) or horizon rod (\ref{Xhor}), it is clear that the set of matrices $P_a(k)$ depend on the parameters $z_a,v_a,  b_{\mu}^a(z_a), b_\mu^R(z_n), \chi_i^a(z_a), \Omega_i^a$. However, due to the translation freedom in the choice of origin of the $z$-axis the solution can only depend on the constants $z_a$ via the rod lengths $\ell_a=z_a-z_{a-1}$ and therefore the solution  depends on (\ref{moduli}).
\end{proof}
\noindent {\bf Remarks.}
\begin{enumerate}
\item  If $D=4$ the determinant fully fixes the metric $h^a(z)= - \tilde{F}_{aNN}(z)^{-1}$ and $\gamma(z) = - \tilde{F}_{a00}(z)^{-1}$  if $I_a$ is an axis or horizon rod.  If $D=5$ symmetry of $F_a(k)$ implies symmetry of the metric $h^a_{\mu\nu}$ and $\gamma_{ij}$ (but not vice-versa).
\item Alternate forms of the general solution can be obtained by replacing $F_a(k)$ with $F_a(k)^T$ for some $a \in \{ 1, \dots, n+1 \}$. Of course, these are all equivalent since  $F_a(k)$ must be symmetric by Proposition \ref{prop:sym}. In fact the symmetry of $F_a(k)$  implies the moduli (\ref{moduli}) satisfy a complicated set of  algebraic constraint equations which will be discussed  below. 
\item The horizon moduli $\chi_i^a(z_a)$ are (up to a constant) the horizon angular momenta $J_i^a$ (\ref{Jhorizon}). We will recover this result from an asymptotic analysis of the general solution later.  On the other hand, the axis moduli $b_\mu^a(z_a)$ are equal to the gravitational fluxes (\ref{4dflux}) and (\ref{5dflux}).
 \item From the explicit form of the matrices (\ref{Faalt}), (\ref{Padef}), (\ref{Xaxis}), (\ref{Xhor})  it is easy to see that the metric components and potentials on each rod are rational  functions of $z$.
\item In general, regularity of the axes imposes further constraints on these moduli from the conditions for the removal of conical singularities (\ref{balance}), (\ref{regupper}), (\ref{reglower}), (\ref{4dkappa}), (\ref{5dkappa}) (see also Appendix \ref{app:corners}). Observe that these regularity conditions also require that $\det h^a_{\mu\nu}$ vanishes at the endpoints of the associated axis rod $I_a$ and  that $\det \gamma_{ij}$ vanishes at the endpoints of a horizon rod.
\end{enumerate}

The parameters (\ref{moduli}) that the general solution for the Ernst and twist potentials $\{ b^a_{\mu}(z), \chi_i^a(z)\}$ on the finite rods depend on include $\{b^a_\mu(z_a), \chi_i^a(z_a)\}$, so there are potential constraints on these from the obvious consistency relations: $b^a_\mu(z)|_{z \to z_{a-1}} = 0$ (recall (\ref{ErnstBC})) and  $b^a_\mu(z)|_{z \to z_a} = b^a_\mu(z_a)$ and the corresponding constraints for horizon rods. 
In total these amount to $2(D-3)(n-1)$  conditions, $(D-3)(n-1)$ of which are automatically satisfied by our solution as the following shows.

\begin{prop}
\label{prop:consistency}
For the general solution (\ref{gensolaxis}), (\ref{gensolhor}) the following identities are satisfied for generic values of the moduli:
\bea
&& \lim_{z \to z_{a-1}} b^a_{\mu}(z) = 0, \label{balower} \\
&& \lim_{z \to z_{a-1}} \chi^a_i(z) = 0 \label{chialower},
\eea
if $I_a$ is a finite axis rod or horizon rod respectively.

On the other hand, for the general solution with $F_a(k)$ replaced by $F_a(k)^T$ the following identities are satisfied for generic values of the moduli:
\bea
&&  \lim_{z \to z_{a}} b^a_{\mu}(z) = b_{\mu}^a(z_a) , \label{baupper}\\
&&  \lim_{z \to z_{a}} \chi^a_i(z) = \chi_i^a(z_a),  \;   \label{chiaupper}
\eea
if $I_a$ is a finite axis rod or horizon rod respectively.
\end{prop}

\begin{proof}
First, using (\ref{Faalt}), we can write  $F_a(k)= X_a(z_{a-1}, k)^{-1} G_a(k)$, where $G_a(k)$ is a matrix with a finite limit as $k \to z_{a-1}$, for $a=2, \dots, n$. Then, if $I_a$ is an axis rod, from (\ref{Xaxis}) we get
\be
\tilde{F}_a(k) = \left(\begin{array}{cc}  -\delta_\mu^{~\nu} & 0 \\ 0 & \frac{1}{2(k-z_{a-1})}  \end{array} \right) \tilde{G}_a(k),    \label{Falower}
\ee
where $\tilde{F}_a(k)$ is defined in Theorem \ref{th:main}, and $G_a \equiv L_a \tilde{G}_a L_a^T$ is defined similarly. Using (\ref{gensolaxis}) implies the solution
\be
b_{\mu}^a(z) = - \frac{2 (z-z_{a-1})\tilde{G}_{a\mu N}(z)}{\tilde{G}_{aNN}(z)}   \;.
\ee
Therefore, if $\lim_{k \to  z_{a-1}} \tilde{G}_{aNN}(k) \neq 0$ for generic parameter values, the claim (\ref{balower}) follows. This is proved in Appendix \ref{apen:ghneq0}. The analysis for a horizon rod is essentially identical.

Next, we can write $F_a(k)^T = X_a(z_a, k)^{-1} H_a(k)$, where $H_a(k)$ is a matrix with a finite limit as $k \to z_a$. Using (\ref{Xaxis}) we find
\be
\tilde{F}_a(k)^T = \left(\begin{array}{cc}  -\delta_\mu^{~\nu} & \frac{b_\mu^a(z_a)}{2(k-z_{a})} \\ 0 & \frac{1}{2(k-z_{a})}  \end{array} \right) \tilde{H}_a(k),    \label{Faupper}
\ee
where $H_a \equiv L_a \tilde{H}_a L_a^T$. Therefore the general solution (\ref{gensolaxis}) with $F_a(k)$ replaced with $F_a(k)^T$ gives
\be
b_{\mu}^a(z) =b_\mu^a(z_a) - \frac{2 (z-z_{a})\tilde{H}_{a \mu N}(z)}{\tilde{H}_{a NN}(z)}   \;,
\ee
which implies (\ref{baupper}), since  $\lim_{k \to  z_{a}} \tilde{H}_{aNN}(k)  \neq 0$ for generic parameter values (again, see Appendix  \ref{apen:ghneq0}). The analysis for a horizon rod is completely analogous.
\end{proof}

\noindent {\bf Remarks}.
\begin{enumerate}
\item
  Conversely, for the general solution the conditions (\ref{baupper}) and (\ref{chiaupper}) generically provide nontrivial constraints on the moduli (\ref{moduli}). Similarly, for the solution with $F_a(k)$ replaced by $F_a(k)^T$, the conditions (\ref{balower}) and (\ref{chialower}) generically give nontrivial constraints. Thus, in either case these consistency relations on the finite rods generically provide $(D-3)(n-1)$ constraints on the moduli (\ref{moduli}).
  \item  There are analogous relations that are satisfied automatically on the semi-infinite rods, i.e.  
  for the solution (\ref{gensolaxis}) using $F_R$ on $I_R$ and $F_L^T$ (rather than $F_L$) on $I_L$ one finds that 
  \begin{equation}
  \label{bLRpoints}
 \lim_{z \to z_{n}}b^R_{\mu}(z) = b^R_{\mu}(z_n)   \; , \qquad  \lim_{z \to z_{1}} b^L_{\mu}(z) = b_{\mu}^L(z_1) \; . 
  \end{equation}
  \item An important consequence of this Proposition is that {\it if} $\tilde{Q}_1(k)$ (and hence $F_a(k)$) is symmetric, then both sets of consistency conditions (\ref{balower}, \ref{chialower}) and  (\ref{baupper}, \ref{chiaupper}) are satisfied and thus provide no further constraint on the moduli.
    \end{enumerate}

We now consider the constraints on the parameters (\ref{moduli}) that arise from the symmetry of  the matrices $F_a(k)$. As can be seen from their definition (\ref{Fadef}), the symmetry of $F_a(k)$ is equivalent to the symmetry of the single matrix $\tilde{Q}_1(k)$. To this end, we establish the following result.

\begin{lemma} 
\label{lemma}
$\tilde{Q}_1(k)$ is a rational function of the form
\be
\tilde{Q}_1(k) = \frac{\sum_{l=0}^{n+1} q_{p} k^p}{\prod_{a=1}^n (k-z_a)} ,
\ee
where $q_{n+1}=0$ for $D=4$ and  $q_{n+1}=-2C \text{diag}( 0, 1, 0)$ for $D=5$. 
\end{lemma}

\begin{proof}
  To see this, it is convenient to rewrite (\ref{Qadef}) for $a=1$, as
\be
{Q}_1(k) =  X_L(z_1, k)^{-1} R(k) X_R(z_n, k),   \label{Q1decomp}
\ee
where we have defined 
\bea
&&R(k)= R_2(k) \dots R_n(k), \label{Rdef} \\
&&R_a(k)= X_a(z_{a-1}, k) X_a(z_a,k)^{-1}  \label{Radef}
\eea
for $a=2, \dots, n$.
Using our solution (\ref{Xaxis}) we find that for any axis rod $I_a$ (excluding $I_L, I_R$)
\be
R_a(k)
=  I_{D-2} +\frac{S_a}{k-z_a}, \qquad   R_a(k)^{-1}= I_{D-2} - \frac{S_a}{k-z_{a-1}},    \label{RIS}
\ee
where
\be
S_a\equiv  L_a \left( \begin{array}{cc} 0& - \tfrac{1}{2} b_\mu^a(z_a) \\ 0 & \ell_a \end{array} \right) L_a^{-1}   \label{Sdef}
\ee
and $I_{D-2}$ is the $(D-2)$-dimensional identity matrix. The same expression holds for any horizon rod upon the obvious replacement of $b^a_\mu(z_a)$ with $\chi_i^a(z_a)$.  The lemma now follows straightforwardly from (\ref{Q1decomp}) and the definition (\ref{Qtilde}).  
\end{proof} 

\noindent {\bf Remarks.} 

\begin{enumerate}
\item For $D=4$,  $q_{n+1}$ is trivially symmetric. For $D=5$,  the explicit form of $C$  is computed in Section \ref{sec:5d}, see (\ref{C5d}), which also ensures that $q_{n+1}$ is automatically symmetric.  Therefore, symmetry of $\tilde{Q}_1(k)$ is equivalent to symmetry of the coefficient matrices 
\be
q^T_p= q_p  \; ,   \label{modspace}
\ee
for $p=0,1, \dots, n$.  (\ref{modspace}) are a set of nonlinear algebraic constraints for the moduli (\ref{moduli}) which together with the inequalities $\ell_a>0$ and (\ref{FNNpos}, \ref{F00neg}) define the {\it moduli space of solutions}. The moduli space equations (\ref{modspace}) can impose up to $\tfrac{1}{2}(D-3)(D-2)(n+1)$ constraints on the parameters. For $D=4$ these must be equivalent to the $n+1$ equations obtained by Varzugin~\cite{Varzugin:1997ee}.
\item Consider the special case where all the continuous moduli (\ref{moduli}) are set to zero, except for the rod lengths $\ell_a$.  Also, for $D=5$, suppose that any finite axis rods have rod vectors  $v_L$ or $v_R$.  Then it is straightforward to see that $\tilde{Q}_1(k)$ is diagonal (the matrix $C$ turns out to be diagonal for $D=4,5$, see (\ref{C4d}, \ref{C5d})). Thus, in particular, $\tilde{Q}_1(k)$ is automatically symmetric and there are no constraints on the remaining moduli $\ell_a$, i.e. we obtain a solution to (\ref{modspace}). This class  corresponds to the (generalised) Weyl solutions which are defined by the additional requirement that the $D-2$ commuting Killing fields are hypersurface-orthogonal~\cite{Emparan:2001wk} (so all Ernst/twist potentials must be constants which can be fixed to zero).

\item
The matrices $F_L(k)$ and $F_R(k)$ which determine the general solution on $I_L$ and $I_R$ respectively can be written in terms of $Q_1(k)$ and $C$ using (\ref{Faalt}). Therefore, the asymptotics of the general solution can be deduced from the asymptotic expansion for $\tilde{Q}_1(k)$ for $k\to \infty$, which from the Lemma takes the form
\be
\tilde{Q}_1(k)= q_{n+1} k + q_n + q_{n+1} \sum_{a=1}^n z_a +O(k^{-1})  \; .  \label{Qtildeasympt}
\ee
The coefficients can be easily extracted from the decomposition (\ref{Q1decomp}) together with
\be
R(k) = I_{D-2} + \frac{S}{k}+ O(k^{-2}) , \qquad S= \sum_{a=2}^n S_a  \; .  \label{S}
\ee
The computation of the matrix $C$ is dimension dependent so we present this and the coefficients in the asymptotic expansion in later sections. \end{enumerate}

We are now ready to consider the moduli space of solutions with $n+1$ rods and $h$ horizons (thus there are $n-1-h$ finite axis rods) that are potentially singular on the axis. The general solution on the $z$-axis we have found depends on a number of moduli (\ref{moduli}):
 the rod structure, the change in Ernst and twist potentials across each axis and horizon rod, and the horizon angular velocities. Thus, the number of continuous parameters is given by $n-1+(n+1+h)(D-3)$.  On the other hand, from the uniqueness and existence Theorems \ref{th1} and \ref{th2} we know that the solutions can be specified by the rod structure and the change in twist potentials across each horizon rod (recall by (\ref{Jhorizon}) these are equal to the horizon angular momenta $\{J^a_i \}$), which consists of $n-1+(D-3)h$ parameters (see (\ref{dimmodsp})). Thus we expect $(D-3)(n+1)$ relations on the moduli (\ref{moduli}); these may be thought of as determining  $\{ \Omega^a_i, b^a_\mu(z_a)
, b^L_\mu(z_1), b^R_\mu(z_n)  \} $ in terms of the fundamental moduli $\{ \ell_a, v_a, \chi^a_i(z_a) \}$ (although in practice these may not be the best parameters to express the solution with).

For $D=4$ we see that this coincides with the number of conditions that symmetry of $\tilde{Q}_1(k)$ can impose, i.e.  (\ref{modspace}),  which gives $n+1$ relations. However, for $D=5$ we find that symmetry of $\tilde{Q}_1(k)$ imposes too many conditions, i.e. it imposes $3(n+1)$ rather than $2(n+1)$ conditions. Hence, for $D=5$, there must exist $n+1$ independent redundancies in the symmetry relations (\ref{modspace}). Therefore, we conclude that while for $D=4$ equations (\ref{modspace}) provide a good description of the moduli space of solutions, for $D=5$ imposing symmetry of $\tilde{Q}_1(k)$ leads to a redundant description of the moduli space.  In Section \ref{sec:5d} we will discuss an alternate description for the $D=5$ moduli space.

\section{Four dimensions}
\label{sec:4d}

\subsection{General solution and physical parameters}

In four spacetime dimensions the general solution on each components of the axis and horizon simplifies. It is therefore worth recording some of the key formulas and the solution again in this case. 
The main simplification arises because there is only one axial Killing field and hence the rod vector which vanishes on any axis rod is always $m=\partial_\phi$ (this of course includes the semi-infinite rods $I_L$ and $I_R$).

Near any axis rod $I_a$, the metric (\ref{gaxisadapted}) relative to the standard basis $(k,m)$  is simply
\be
g = \left( \begin{array}{cc} h - h^{-1} \rho^2 w^2 & \rho^2 h^{-1} w\\ h^{-1} \rho^2 w & - h^{-1} \rho^2 \end{array} \right),  \label{4dg}
\ee
where $h <0$.  The general solution to the linear system (\ref{LPbdry})  on the each axis rod  can be written as (\ref{LPgensol}) where 
\be
X_a(z, k) =  \left( \begin{array}{cc} -1 & b^a(z) \\ 0& 2(k-z) \end{array} \right)   \label{4dXaxis}
\ee
and $b^a(z)=b(z)-b(z_{a-1})$ for $a=2, \dots, n$, $b^L(z)=b^R(z)=b(z)$, and $b(z)$ is the Ernst potential (\ref{4dernst}) fixed by imposing that $b \to 0$ at infinity. 

On the other hand, near a horizon rod $I_a$  the metric (\ref{ghoradapted}) relative to the standard basis is
\be
g = L_a  \left( \begin{array}{cc} \gamma - \gamma^{-1} \rho^2 \omega^2 & \rho^2 \gamma^{-1} \omega \\ \gamma^{-1} \rho^2 \omega & - \gamma^{-1} \rho^2 \end{array} \right)L_a^T,
\ee
where $\gamma>0$ and
\be
L_a = \left( \begin{array}{cc} -\Omega^a & 1 \\1& 0 \end{array} \right)  \; .
\ee
The general solution to the linear system on $I_a$ is (\ref{LPgensol}) where
\be
X_{a} (z,k)= L_a\left( \begin{array}{cc} -1 & \chi^a(z) \\ 0& 2(k-z) \end{array} \right) L_a^{-1}  \;   \label{4dXH}
\ee
and $\chi^a(z)= \chi(z)-\chi(z_{a-1})$ is the twist potential defined by (\ref{twist}).

We  now consider the general solution with rods $I_{a=1, \dots, n+1}$.   This is given by Theorem \ref{th:main} in terms of the matrices $F_a(k)$. In turn, the matrices $F_a(k)$ are constructed from the matrices $P_a(k)$ and a constant matrix $C$ arising from the solution to the linear system at infinity using (\ref{Faalt}).  To fix $C$ we need to explicitly compute asymptotic solutions to the linear system (\ref{Psiinfty}), (\ref{Deltainfty}) which match on to the axis solution (\ref{barPsiaxis}), (\ref{Deltaaxis}). Then, from the definition (\ref{Padef}) for matrices $P_a(k)$, we deduce that the general solution on the axis and horizons depends only on the following constants: the rod lengths $\ell_a= z_{a}-z_{a-1}$, the angular velocity of each horizon $\Omega^a$, the jump in Ernst potentials $b(z_a)-b(z_{a-1})$ over each axis rod and jump in twist potentials $\chi(z_a)-\chi(z_{a-1})$ over each horizon rod. 

We now turn to the computation of the constant matrix $C$.  Firstly, Minkowski spacetime  in polar coordinates (\ref{polarcoords}) is given by
\be
\bar{g} = \text{diag}(-1, r^2\sin^2\theta), \qquad \bar{\nu}=0 \; ,    \label{4dflat}
\ee
which implies $\bar{S}= \text{diag}(0, 2)$ and $\bar{T}= \text{diag}(0,2 \cos\theta)$, where $S,T$ are defined in (\ref{LPpolar}). The general solution to (\ref{LPpolar}) in Minkowski space, which agrees with the axis solution (\ref{barPsiaxis}), is
\be
\bar{\Psi}_+ =\text{diag}(-1 , \;  \mu_+ )   \; .   \label{4dPsiflat}
\ee
Thus, using the  asymptotic expansion for $\mu_+$ in polar coordinates, given
 in section \ref{subsec:infinity}, we find that  
 \be
\bar{\Psi}_+ (r, \theta, k) =\left\{ \begin{array}{c}  \text{diag} \left(-1 , \;  -r (1+ \cos \theta) +O(1)  \right)   \qquad 0 \leq \theta < \theta_* \\  \text{diag} \left(  -1 , \;  r (1 - \cos \theta) +O(1)  \right)   \qquad \theta_*  < \theta \leq \pi \end{array} \right.  \;  \label{4dbarPsiinfty}
\ee
as $ r\to \infty$.

More generally, 
any four-dimensional asymptotically flat spacetimes in polar coordinates (\ref{polarcoords}) must take the form
\be
g =  \left( \begin{array}{cc} -1+ \frac{2M}{r}+ O(r^{-2})  & - \frac{2J\sin^2\theta}{r} (1+O(r^{-1})) \\  - \frac{2J\sin^2\theta}{r} (1+O(r^{-1})) & r^2\sin^2\theta (1+ O(r^{-1}))\end{array} \right)  \; ,  \label{4dAF}
\ee
as $r \to \infty$, where $M, J$ are the ADM mass and angular momentum.  It follows that the corresponding matrices $S, T$ in the linear system (\ref{LPpolar}) are now given by
\be
S -\bar{S}=  \left( \begin{array}{cc}   O(r^{-1})  & O(r^{-3}) \\  O(r^{-1}) & O(r^{-1})\end{array} \right)  \; , \qquad 
T-\bar{T}=\left( \begin{array}{cc}  O(r^{-2})  & O(r^{-3}) \\ O(r^{-2}) &  O(r^{-1})\end{array} \right)  \;  \; ,
\ee
which together with (\ref{4dbarPsiinfty}) imply that the RHS of  equations (\ref{Deltaeq}) are
\be
\Upsilon_{r+} = \left( \begin{array}{cc}   O(r^{-2})  & O(r^{-3}) \\  O(r^{-3}) & O(r^{-2})\end{array} \right) , \qquad\Upsilon_{\theta +} = \left( \begin{array}{cc}  O(r^{-1})  & O(r^{-2}) \\ O(r^{-2}) &  O(r^{-1})\end{array} \right)   \; 
\ee
for all $0\leq \theta \leq \pi$.
This justifies the claim (\ref{Deltainfty}). Thus we may compute $C$ from (\ref{C}) using (\ref{4dbarPsiinfty}), which gives
\be
\label{C4d}
C= -I_2  \; .
\ee
Note that from (\ref{Qtilde}) we deduce $\tilde{Q}_1(k)= - Q_1(k)$ and hence  that $Q_1(k)$ must be a symmetric matrix.

As a simple example, consider the rod structure of Minkowski spacetime, which is given by a single rod consisting of the whole $z$-axis. Thus the right and left semi-infinite axes are identified $I_L=I_R$ and there are no continuity conditions to be imposed.  Then combining (\ref{MLR}) with (\ref{bpaxis}) gives
\be
\mathring{g}(z)= X(z,z)
\ee
which using (\ref{4dXaxis}) is equivalent to
\be
h(z)=-1, \qquad b(z)=0  \; .
\ee
This is indeed the data for Minkowski spacetime (\ref{4dflat}).  In itself this a nontrivial result: it shows that any asymptotically flat stationary and axisymmetric vacuum solution with the same rod structure as Minkowski spacetime is isometric to Minkowski spacetime on the axis.  This of course follows from the well known no-soliton theorems.

Given a solution $(h(z), b(z))$ on $I_L$ or $I_R$ we can compute the  mass and angular momentum. Comparing to (\ref{4dAF}) we find as $|z|\to \infty$ 
\be
h(z) = -1+\frac{2M}{|z|} +O(z^{-2}) \; ,   \qquad b(z) = - \frac{\text{sign}(z) 2 J}{z^2} + O(z^{-3})   \label{4daxisMJ}
\ee
where $b(z)$ is determined using (\ref{baxis}) and we have fixed the integration constant so that it vanishes at infinity.  

Finally, given the solution on a horizon rod, the surface gravity can be computed from (\ref{4dkappa}), which in principle may impose a nontrivial constraint on the parameters.

\subsection{Asymptotics of general solution}

We now confirm our general solution  (\ref{gensolaxis}) is asymptotically flat and compute the asymptotic charges.  In particular, the metric and Ernst potential on $I_L$ are given by the components of $F_L(k)= Q_1(k)^T$. Using the  decomposition of $Q_1(k)$ given in equation (\ref{Q1decomp}) we find
\bea
&&F_L(k) = \left( \begin{array}{cc} R_0^{~0}(k) - \frac{R_1^{~0}(k) b(z_1)}{2(k-z_1)}  & -\frac{R_1^{~0}(k)}{2(k-z_1)}  \\  \tilde{F}_{L10}(k) &  \frac{ R_{1}^{~0}(k) b(z_n)+ 2(k-z_n) R_1^{~1}(k)}{2(k-z_1)} \end{array} \right) \; , \\
&&\tilde{F}_{L10}(k) = 2(k-z_n) \left( -R_0^{~1}(k)+\frac{R_1^{~1}(k) b(z_1)}{2(k-z_1)} \right) - b(z_n) \left(R_0^{~0}(k) - \frac{R_1^{~0}(k) b(z_1)}{2(k-z_1)} \right)   \nonumber
\eea
and $R_A^{~B}(k)$ denote the components of the matrix (\ref{Rdef}) in the standard basis.
Hence using (\ref{gensolaxis}) we find that the solution on $I_L$ is
\bea
&&h(z)= -\frac{2(z-z_1)}{ R_{1}^{~0}(z) b(z_n)+ 2(z-z_n) R_1^{~1}(z)},  \\
&&b(z)= -\frac{R_1^{~0}(z)}{ R_{1}^{~0}(z) b(z_n)+ 2(z-z_n) R_1^{~1}(z)} \;,  
\eea
where we have used the fact that $h(z) = - \tilde{F}_{L11}(z)^{-1}$ (see Remark 2 below Theorem \ref{th:main}). We may now compute the asymptotics of the solution as $z\to -\infty$. Using (\ref{S}) we find
\be
h(z) = -1- \frac{S_0^{~0}}{z} +O(z^{-2}), \qquad b(z) = - \frac{S_1^{~0}}{2 z^2} + O(z^{-3})  \label{4dinfty}
\ee
where $S_A^{~B}$ denote the components of the matrix $S$ defined in (\ref{S}). 

We can evaluate these relations more explicitly using (\ref{Sdef}).  We find that 
\bea
&&S_a = \left( \begin{array}{cc} 0& - \tfrac{1}{2}b^a(z_a) \\ 0 & \ell_a  \end{array} \right) \;,  \qquad I_{a\neq L, R} \subset \hat{A} \; , \\
&& S_a=  \left( \begin{array}{cc} \ell_a + \tfrac{1}{2} \Omega^a \chi^a(z_a) & \Omega^a(\ell_a + \tfrac{1}{2} \Omega^a \chi^a(z_a) )  \\  - \tfrac{1}{2}\chi^a(z_a) &  - \tfrac{1}{2}\Omega^a\chi^a(z_a)  \end{array} \right) \; , \qquad I_a \subset \hat{H} \; .
\eea
Therefore, from the asymptotics of the general solution derived above we deduce
\bea
&&M =  \sum_{I_a \subset \hat{H}} M_a \;,  \qquad M_a= \tfrac{1}{2}(\ell_a+ \tfrac{1}{2} \Omega^a \chi^a(z_a)), \label{4dMa} \\ &&J=  \sum_{I_a \subset \hat{H}} J^a \; , \qquad J^a= \tfrac{1}{8}\chi^a(z_a) \; . \label{4dJa} 
\eea
Observe that expressions for the angular momenta are the well-known relations (\ref{Jhorizon}).  The expressions for the mass (\ref{4dMa}) together with (\ref{area}) imply the Smarr relation (for multi-black holes).

On the other hand, suppose instead we use the alternate form of the solution where $F_L(k)$ is replaced by $F_L(k)^T$. Then the only change in the solution is that now $b(z)= \tilde{F}_{L10}(z)/\tilde{F}_{L11}(z)$. Working to first order in the expansion for $R(z)$ as in (\ref{S}) allows us only to determine the $O(1)$ term,
\begin{equation}
    b(z) = b(z_1)-2 S_0^{~1}- b(z_n) + O(z^{-1}).
\end{equation}
Therefore $b(z) \to 0$ implies
\bea
b(z_n)- b(z_1) = -2S_0^{~1} =  \sum_{\substack{I_{a} \subset \hat{A}\\a \neq L,R}} b^a(z_a) - 4 \sum_{I_a \subset \hat{H}} \Omega^a M_a  \label{4dbinfty}  \; .
\eea
We provide an alternate derivation of this relation in Appendix \ref{apen:ernstHorizon} (the same relation was  also found in~\cite{Varzugin:1997ee}).  It is worth emphasising that the coefficient $q_n$ appearing in Lemma \ref{lemma} can be deduced from the above to be 
\be
q_n = -\left( \begin{array}{cc} 1 & b(z_1)- b(z_n) -2 S_0^{~1} \\  0& 1 \end{array} \right)   \qquad
\ee
and thus symmetry of this is equivalent to (\ref{4dbinfty}). Therefore this  asymptotic analysis solves the $p=n$ moduli space equation (\ref{modspace}). 

We can similarly consider the asymptotics of the solution on $I_R$ which is given by the matrix $F_R(k)= Q_1(k)^{-1}$. The computation is essentially the same as above and one finds the formulas (\ref{4dMa}) and (\ref{4dbinfty}). Furthermore, imposing symmetry of $F_R(k)$ one now gets (\ref{4dJa}).

\subsection{Kerr solution}

We now consider solutions with the same rod structure as Kerr. Namely, we assume there are three rods $I_L= (-\infty, z_1), I_H=(z_1, z_2)$, $I_R=(z_2, \infty)$ where $I_H$ is a horizon rod.  The solution is given by Theorem \ref{th:main} in terms of  the matrices $F_a(k)$  given by  (\ref{Faalt}), which in this case are simply 
\be
F_L(k)=Q_1(k)^T\; , \qquad  F_H(k)= P_1(k)^{-1} P_2(k)^T \; , \qquad  F_R(k)=Q_1(k)^{-1} \; ,
\ee
where $Q_1(k)=P_1(k) P_2(k)$ and the $P_a(k)$ are defined by (\ref{Padef}).

First, let us consider $z<z_1$.  It is convenient to use the alternate form of the solution $(h(z), b(z))$ on $I_L$ where $F_L(k)$ is replaced with $F_L(k)^T= Q_1(k)$ in (\ref{gensolaxis}). We  then compute the mass and angular momentum by comparing to the asymptotic expansions (\ref{4daxisMJ}), which in fact  also fixes $b(z_1), b(z_2)$. We find
\bea
&&M = \tfrac{1}{2} [\ell_H + \tfrac{1}{2} \Omega (\chi(z_2)-\chi(z_1)) ]  \; , \label{KerrM} \\
&&b(z_1)=- b(z_2) = 2\Omega M  \label{biKerr}  \; , \\
&&J = \Omega M^2 [ 4M - \tfrac{1}{2}\Omega (\chi(z_2)-\chi(z_1)) ]  \; , \label{KerrJ}
\eea
where $\ell_H=z_2-z_1$ and we have written the latter quantities in terms of $M$.   

On the other hand, from our general asymptotic analysis, (\ref{4dJa}) reduces to
\be
\chi(z_2)- \chi(z_1)= 8 J \;  , \label{4ddeltaY}
\ee
while the expressions (\ref{4dMa}) and (\ref{4dbinfty}) already follow from (\ref{KerrM}) and (\ref{biKerr}). We can use (\ref{4ddeltaY}) to eliminate $\chi(z_2)-\chi(z_1)$. Then (\ref{KerrM})  gives\footnote{Combining this with (\ref{area}) leads to the standard Smarr relation.}
\be  
\ell_H = 2(M- 2 \Omega J)  \label{ellHKerr}
\ee
and (\ref{KerrJ}) can be solved for $J$
\be
J = \frac{4 \Omega M^3}{1+ 4 \Omega^2 M^2}  \; .  \label{KerrJ2}
\ee
One can now check that $Q_1(k)$ is symmetric and therefore we have fully solved the moduli space equations (\ref{modspace}).

Substituting (\ref{KerrJ2}) back into (\ref{ellHKerr}) we find
\be
\ell_H= \frac{2 M(1- 4 \Omega^2 M^2)}{1+ 4\Omega^2 M^2}
\ee
and hence positivity of the horizon rod length $\ell_H>0$ and of the mass $M>0$ implies
\be
|\Omega| < \frac{1}{2M} \; .  \label{Kerrmodspace}
\ee
This determines the full moduli space of nonextreme Kerr black hole solutions. Indeed, the relation (\ref{KerrJ2}) now implies the well-known inequality
\be
|J|<M^2   \label{KerrMS}  \;  . 
\ee
In terms of the physical quantities the solution simplifies a little. We find for $z<z_1$:
\bea
&&h(z) = \frac{-(z-z_1)(z-z_2)}{(z-z_2)^2 - 4 \Omega J (z-z_2) + 4 M \Omega J},  \\  &&b(z) = \frac{2J}{(z-z_2)^2 - 4 \Omega J (z-z_2) + 4 M \Omega J}    \; .
\eea
It is worth noting that the relation for $b(z_1)$ in (\ref{biKerr}) is automatically satisfied by this solution (as it must be by Remark 2 below Proposition \ref{prop:consistency}). Thus from the above analysis we see that the solution is naturally parameterised by $(M, \Omega)$\footnote{Eq (\ref{KerrJ2}) can be solved for $\Omega$, yielding $\Omega J = M- \sqrt{M^2-\frac{J^2}{M^2}}$.  Using this, the solution can be equivalently uniquely parameterised in terms of $(M, J)$.}.  It is interesting to note that we have fully determined the moduli space (\ref{KerrMS}) of non-extremal Kerr solutions by only analysing one semi-infinite axis (this was also found in~\cite{Neugebauer:2003qe}).

A similar analysis can be performed for the other semi-infinite axis $z>z_2$. One again finds (\ref{KerrM})-(\ref{KerrJ}) and the solution for $z>z_2$:
\bea
&&h(z) = \frac{-(z-z_1)(z-z_2)}{(z-z_1)^2 +4 \Omega J (z-z_1) + 4 M \Omega J},  \\  &&b(z) = \frac{-2J}{(z-z_1)^2 +4 \Omega J (z-z_1) + 4 M \Omega J}.
\eea
Again, the relation (\ref{biKerr}) is  automatically satisfied by this $b(z_2)$ (as it must be).
Thus, the analysis of this semi-infinite axis yields equivalent results.

Finally, consider the horizon rod $z_1< z<z_2$.  We find that (\ref{gensolhor}) gives
\bea
&&\gamma(z)=\frac{-4(z-z_1)(z-z_2)}{1+4\Omega^2(z-(z_1-M))(z-(z_2+M))},  \\
&&\chi(z)= \frac{-8\Omega(z-z_1)^2(z-(z_2+M))}{1+4\Omega^2(z-(z_1-M))(z-(z_2+M))},
\eea
where we have used (\ref{biKerr}) and (\ref{KerrM}).  The solution for $\chi(z)$ can be shown to  automatically satisfy (\ref{4ddeltaY}) as a consequence of the above relations (as guaranteed by Proposition \ref{prop:consistency}). Furthermore, it can be checked that ${\gamma}'(z_1)=-\gamma'(z_2)$ automatically so (\ref{4dkappa}) implies that the metric on the horizon has no conical singularities and the surface gravity simplifies to
\be
\kappa= \frac{1- 4 \Omega^2M^2}{4M}  \; .
\ee
Notice that (\ref{Kerrmodspace}) is equivalent to  the non-extremality condition $\kappa>0$.

To summarise, we have fully determined the metric on the whole $z$-axis for any solution with the same rod structure as Kerr and computed all asymptotic and horizon physical quantites. We find this reproduces the full moduli space of nonextremal Kerr black holes, as it must from the no-hair theorem. It is interesting to note that our analysis does this without knowledge of the full spacetime metric.

\section{Five dimensions}
\label{sec:5d}

\subsection{General solution and physical parameters}

For $D=5$ the general solution for the metric data $(h^a_{\mu\nu}(z), b^a_\mu(z))$ on any axis rod  $I_a$ takes the explicit form (\ref{gensolaxis}), with an analogous expression for the data $(\gamma_{ij}^a(z), \chi_i^a(z))$ on any horizon rod (\ref{gensolhor}). The solution is given in terms of components of the matrices $F_a(k)$, which depend on the moduli (\ref{moduli}) and a matrix $C$.  The matrix $C$  arises in the asymptotic solution (\ref{Psiinfty}), (\ref{Deltainfty}), in particular it relates the solution in the left and right segments (\ref{Cdef}). Therefore, to fully fix the general solution on the axis and horizon rods we  need to find the solution to the linear system in Minkowski spacetime which matches onto our axis solution (\ref{barPsiaxis}) and compute the corresponding matrix $C$ using (\ref{C}).

Five-dimensional Minkowski spacetime in polar coordinates (\ref{polarcoords}) is
\be
\bar{g} =  \text{diag} \left( -1 , \;  r(1-\cos\theta) , \;  r(1+\cos\theta)  \right), \qquad e^{2\bar{\nu}}=\frac{1}{2r}   \; ,   \label{5dflat}
\ee 
which gives
\be
\bar{S} = \text{diag}( 0, 1, 1), \qquad \bar{T}= \text{diag}(0, 1+\cos\theta, -(1-\cos\theta)) \; .
\ee
For $r>|k|$, the solution to (\ref{LPpolar}) on Minkowski space which agrees with the axis solution (\ref{barPsiaxis})  is
\be
\bar{\Psi}_+ = \text{diag} \left( -1, \; r(1-\cos\theta)-\mu_+  , \; r(1+\cos\theta)+ \mu_+ \right) N(r, \theta, k)  \; , \label{5dPsiflat}
\ee
where the matrix
\be
N(r, \theta, k)= \left\{ \begin{array}{cc} \text{diag}(1,\; -1, \; -2k)^{-1}  & \quad 0 \leq \theta < \theta_* \\ \text{diag}(1,\; 2k, \; 1)^{-1}  & \quad  \theta_* <\theta \leq \pi \end{array} \right. 
\ee
is needed to ensure the solution matches with the one on the axes.
In particular,  using the  asymptotic expansions
 in Section \ref{subsec:infinity} we find that 
 \be
\bar{\Psi}_+ (r, \theta, k) =\left\{ \begin{array}{c}  \text{diag} \left(-1 , \; -2 r+ O(1) , \; -\tfrac{1}{2}(1+\cos\theta)+ O(r^{-1})   \right)   \qquad 0 \leq \theta < \theta_* \\  \text{diag} \left(  -1 , \; -\tfrac{1}{2}(1-\cos\theta)+O(r^{-1}) , \; 2r +O(1)  \right)   \qquad \theta_*  < \theta \leq \pi \end{array} \right. \; ,  \label{5dPsiinfty}
\ee 
 as $ r\to \infty$.

Now, any five-dimensional asymptotically flat spacetimes in polar coordinates must take the form~\cite{Harmark:2004rm}
\be
g =  \left( \begin{array}{ccc} -1+ \frac{4M}{3\pi r}+ O(r^{-2})  & - \frac{J_1(1-\cos\theta)}{\pi r} (1+O(r^{-1})) & - \frac{J_2(1+\cos\theta)}{\pi r} (1+O(r^{-1}))\\  - \frac{J_1(1-\cos\theta)}{\pi r} (1+O(r^{-1})) & r(1-\cos\theta) (1+ O(r^{-1})) & \frac{\zeta \sin^2\theta }{r}(1+ O(r^{-1})) \\  - \frac{J_2(1+\cos\theta)}{\pi r} (1+O(r^{-1})) &\frac{\zeta \sin^2\theta}{r}(1+ O(r^{-1})) &  r(1+\cos\theta) (1+ O(r^{-1}))\end{array} \right)  \; ,  \label{5dAF}
\ee
as $r \to \infty$, where $M, J_i$ are the ADM mass and angular momenta and $\zeta$ is a gauge invariant constant.  From this one can show that $S,T$ appearing in the linear system in polar coordinates (\ref{LPpolar}) satisfy
\be
S- \bar{S}= \left( \begin{array}{ccc} O(r^{-1}) & O(r^{-2}) & O(r^{-2}) \\ O(r^{-1}) & O(r^{-1}) & O(r^{-2}) \\ O(r^{-1}) & O(r^{-2}) & O(r^{-1}) \end{array} \right), \qquad 
T -\bar{T}= \left( \begin{array}{ccc} O(r^{-2}) & O(r^{-2}) & O(r^{-2}) \\  O(r^{-2}) &   O(r^{-1}) &  O(r^{-2}) \\  O(r^{-2}) &  O(r^{-2}) & O(r^{-1}) \end{array} \right) \; .
\ee
Then using (\ref{5dPsiinfty}) we find that the matrices $\Upsilon$ defined in (\ref{Deltaeq}) satisfy\footnote{In fact one obtains different fall-offs for $0\leq \theta< \theta_*$ and $\theta_*< \theta \leq \pi$ where some components have faster fall-offs. We will not need these in our analysis. }
\be
\Upsilon_{r +} = O(r^{-2})
, \qquad \Upsilon_{\theta +} = O(r^{-1}) \; ,
\ee
for all $0\leq \theta \leq \pi$, 
thus justifying (\ref{Deltainfty}).   Finally, from (\ref{C}) we find
\be
\label{C5d}
C= \left( \begin{array}{ccc} -1 & 0 & 0 \\ 0 & 1 & 0 \\ 0 & 0 & -1 \end{array} \right) \; .
\ee
We therefore have fully fixed the general solution.

We now relate the parameters of the solution to the asymptotic quantities. Given a solution $(h^L_{\mu\nu}, b^L_{\mu})$ on $I_L$ we can compute the mass and angular momenta.  From (\ref{5dAF}) we deduce that 
\bea
&&h^L_{\mu\nu}(z) = \left( \begin{array}{cc} -1 - \frac{4M}{3\pi z}+ O(z^{-2})&   \frac{2J_1}{\pi z} +O(z^{-2}) \\  \frac{2J_1}{\pi z} +O(z^{-2}) & -2z +O(1) \end{array} \right),  \nonumber \\
&&b^L_{\mu}(z)= \left( \begin{array}{c} -\frac{2J_2}{\pi z}+ O(z^{-2}) \\  \frac{4 \zeta}{z}+O(z^{-2}) \end{array} \right)  \; ,  \label{5dLinfty}
\eea
as $z\to -\infty$,  where $b^L_\mu$ is determined using (\ref{baxis}) and we have fixed the integration constant so that $b^L_\mu \to 0$ at infinity.
Similarly, given a solution on $I_R$  we can compute the  asymptotic quantities again from (\ref{5dAF}) which in this case implies that 
\bea
&&h^R_{\mu\nu}(z) = \left( \begin{array}{cc} -1 + \frac{4M}{3\pi z}+ O(z^{-2})&  - \frac{2J_2}{\pi z} +O(z^{-2}) \\ - \frac{2J_2}{\pi z} +O(z^{-2}) & 2z +O(1) \end{array} \right), \nonumber  \\
&&b^R_{\mu}(z)= \left( \begin{array}{c} -\frac{2J_1}{\pi z}+ O(z^{-2}) \\  \frac{4 \zeta}{z}+O(z^{-2}) \end{array} \right),    \label{5dRinfty}
\eea
as $z\to \infty$, again using (\ref{baxis}) and fixing constants so $b^R_\mu \to 0 $ at infinity.  

On the other hand, given a solution on a horizon rod $I_a$ we may compute the surface gravity from (\ref{5dkappa}), which in principle may provide one constraint on the parameters. Similarly, given a solution on an axis rod $I_a$, smoothness requires that there are no conical singularities at any endpoint of the rod, the conditions for which are given by (\ref{regupper}) and (\ref{reglower}).

The moduli (\ref{moduli}) that appear in the general solution  in Theorem \ref{th:main} are constrained by the symmetry of $\tilde{Q}_1(k)$, which is equivalent to the moduli space equations (\ref{modspace}). As noted at the end of Section \ref{sec:modspace}, these equations give a redundant description of the moduli space. On the other hand, Proposition \ref{prop:consistency} and  Remark 1 that follows it show that the consistency conditions on the potentials $b_\mu^a(z), \chi_i^a(z)$ for  the {\it finite} rods generically provide $(D-3)(n-1)$ constraints on the parameters for the general solution. Thus supplementing these with the asymptotic conditions for the potentials (\ref{bLRinfty})  gives $(D-3)(n+1)$ constraints on the parameters as required.  This leads to the following conjecture.

\begin{conj} 
\label{conj}  Given the $D=5$ solution (\ref{gensolaxis}), (\ref{gensolhor})   with $F_a(k)$ replaced by $F_a(k)^T$ for $a=1, \dots, n$ and the inequalities $\ell_a>0$, (\ref{FNNpos}) and (\ref{F00neg}), the moduli space equations (\ref{modspace}) are satisfied if and only if (\ref{balower}), (\ref{chialower}) and (\ref{bLRinfty}) are imposed.
\end{conj}

The converse statement for finite rods is true by Remark 3  below Proposition \ref{prop:consistency}.
The motivation for using the alternate form of the general solution in Theorem \ref{th:main}  with $F_a(k)$ replaced by $F_a(k)^T$ for $a=1, \dots, n$ comes from the following observations. Firstly, Proposition \ref{prop:consistency} shows that imposing (\ref{balower}) on the solution $b_\mu^a(z)$  for the finite axis rods gives an equation for $b^a_\mu(z_a)$ (and similarly for the horizon rods). To see this note that since (\ref{baupper}) is automatic we can write $b^a_\mu(z)= b^a_\mu(z_a)+(z-z_a) f_a(z)$ for some smooth function $f_a(z)$. Then, evaluating at $z=z_{a-1}$ gives $b^a_\mu(z_a)=\ell_a f_a(z_{a-1})$, which can be taken as an equation for $b_\mu^a(z_a)$ (of course, $f_a(z_{a-1})$ will typically be a function of the moduli including $b^a_\mu(z_a)$ itself, so this is a nonlinear equation). Secondly, the asymptotic analysis in the next section shows that $b^L_\mu(z)|_{z\to -\infty}= 0$ fixes $b^L_\mu(z_1)$ and $b_\mu^R(z) |_{z\to \infty}= 0$ fixes $b_\mu^R(z_n)$.

In any case, in all the examples that we study below we find this conjecture is valid and is a convenient way of solving the moduli space equations (more precisely, for single black holes we use a slightly modified version of this conjecture, given below).  In particular, in practice it is easier to impose the $(D-3)(n+1)$ relations listed in the conjecture, rather than the symmetry of $\tilde{Q}_1(k)$ which as argued earlier must include redundancies.

\subsection{Asymptotics of general solution}

We now confirm our general solution is asymptotically flat and deduce the asymptotic charges.  First, consider the solution (\ref{gensolaxis}) on $I_L$ which is given by the components of $F_L(k)= - C^{-1} Q_1(k)^T$ defined by (\ref{Facompaxis}).  Using (\ref{Q1decomp}) to write $Q_1(k)$ in terms of $R(k)$ defined in  (\ref{Rdef}) and then using the asymptotic expansion (\ref{S}) gives
\bea
&&h^L_{\mu \nu}(z) =\left( \begin{array}{cc} -1-\frac{S_0^{~0}}{z} +O(z^{-2}) & - \frac{S_{1}^{~0}}{z} + O(z^{-2}) \\- 2S_0^{~1} - b^R_0(z_n) +O(z^{-1}) & -2 z +2 (z_n-S_1^{~1})+ O(z^{-1})\end{array} \right),  \label{hLinfty} \\
&&b^L_\mu(z) = \left( \begin{array}{c} \frac{S_2^{~0}}{z} + O(z^{-2})  \\  2S_2^{~1}+b^R_1(z_n)+ O(z^{-1}) \end{array} \right),  \; 
\eea
as $z \to -\infty$. Thus comparing to the asymptotics (\ref{5dLinfty}) we deduce that
\be
M = \frac{3\pi S_0^{~0}}{4}, \qquad J_i =- \frac{\pi S_i^{~0}}{2} , \qquad b^R_\mu(z_n)= -2 \left( \begin{array}{c} S_0^{~1} \\ S_2^{~1}\end{array} \right).
\ee
Using (\ref{Sdef}) we can evaluate  these expressions more explicitly.  We find that
\bea
&&S_a = \left( \begin{array}{ccc} 0& -  \tfrac{1}{2} v_a^1 b^a_0(z_a) &  - \tfrac{1}{2} v_a^2  b^a_0(z_a)  \\ 0 & -\tfrac{1}{2} \epsilon_a v_a^1 ( 2u_a^2 \ell_a + v_a^2 b^a_1(z_a)) &  -\tfrac{1}{2} \epsilon_a v_a^2 ( 2u_a^2 \ell_a + v_a^2 b^a_1(z_a))  \\ 0  & \tfrac{1}{2} \epsilon_a v_a^1 ( 2u_a^1 \ell_a + v_a^1 b^a_1(z_a)) &  \tfrac{1}{2} \epsilon_a v_a^2 ( 2u_a^1 \ell_a + v_a^1 b^a_1(z_a))  \end{array} \right) \;, \quad I_{a\neq L, R} \subset \hat{A}  \; ,\\
&& S_a=  \left( \begin{array}{ccc} \ell_a + \tfrac{1}{2} \Omega_i^a \chi_i^a(z_a) & \Omega_1^a(\ell_a + \tfrac{1}{2} \Omega_i^a \chi_i^a(z_a) )  & \Omega_2  ^a(\ell_a + \tfrac{1}{2} \Omega_i^a \chi_i^a(z_a) )\\  - \tfrac{1}{2}\chi_1^a(z_a) &  - \tfrac{1}{2}\Omega_1^a\chi_1^a(z_a) & - \tfrac{1}{2}\Omega_2^a\chi_1^a(z_a)  \\  - \tfrac{1}{2}\chi_2^a(z_a) & - \tfrac{1}{2}\Omega_1^a\chi_2^a(z_a) & - \tfrac{1}{2}\Omega_2^a\chi_2^a(z_a) \end{array} \right)  \;, \quad I_a \subset \hat{H} \; ,
\eea
where $(u_a,v_a)$ is a basis of $U(1)^2$-Killing fields such that $v_a$ is the rod vector and $\epsilon_a = (u_a^1v_a^2 - u_a^2 v_a^1)^{-1}$.
Therefore, from the asymptotics of the general solution derived above we deduce
\bea
&&M =  \sum_{I_a \subset \hat{H}} M_a \;,  \qquad M_a= \tfrac{3\pi}{4}(\ell_a+ \tfrac{1}{2} \Omega_i^a \chi_i^a(z_a)), \label{5dM}\\
&&J_i=  \sum_{I_a \subset \hat{H}} J_i^a \; , \qquad J_i^a= \tfrac{\pi}{4}\chi_i^a(z_a), \label{5dJ}\\
&&b^R_\mu(z_n)= \sum_{I_a \subset \hat{H}} \frac{4\Omega^a_1}{3\pi} \left( \begin{array}{c} -2  M_a  \\   3 J_2^a\end{array}\right)
+ \sum_{\substack{I_a \subset \hat{A}\\ a\neq L,R}} v^1_a  \left( \begin{array}{c} b^a_0(z_a) \\ - \epsilon_a(2u_a^1 \ell_a + v_a^1 b^a_1(z_a))\end{array}\right).\label{5dbR}
\eea
Note that again we reproduce the well-known relations between angular momenta and the change in twist potential across a horizon rod (\ref{Jhorizon}).  Combining these with the above formulae for the mass, together with (\ref{area}), gives the Smarr relation.   Notice that in the absence of a black hole $M=0$ and $J_i=0$, in line with the no-soliton theorem.

If instead we use the alternate general solution with $F_L(k)$ replaced by $F_L(k)^T$ we find
\be
b^L_\mu(z) = b^L_\mu(z_1) - 2 \left( \begin{array}{c} S_0^{~2}  \\  S_1^{~2} \end{array} \right) + O(z^{-1}) ,
\ee 
with $h^L_{\mu\nu}(z)$ given by the transpose of (\ref{hLinfty}).
 Thus the asymptotics  of $b^L_\mu$ now give
\be
b^L_\mu(z_1)= 2 \left( \begin{array}{c} S_0^{~2} \\ S_1^{~2}\end{array} \right)  =\sum_{I_a \subset \hat{H}} \frac{4\Omega^a_2}{3\pi} \left( \begin{array}{c} 2  M_a  \\   -3 J_1^a\end{array}\right)
- \sum_{\substack{I_a \subset \hat{A}\\ a\neq L,R}} v^2_a  \left( \begin{array}{c} b^a_0(z_a) \\ \epsilon_a(2u_a^2 \ell_a + v_a^2 b^a_1(z_a))\end{array}\right) \label{5dbL} \; .
\ee
The $O(z^{-1})$ term, and hence $J_2$, requires a higher order calculation than the one given by (\ref{S}).  On the other hand, the asymptotics of $h^L_{\mu\nu}(z)$ give the mass (\ref{5dM}), the angular momentum $J_1$ (\ref{5dJ}) and the Ernst potential $b_0^R(z_n)$ (\ref{5dbR}).

We may perform an analogous calculation on $I_R$. Using the general solution (\ref{gensolaxis}) and the explicit form $F_R(k)=-\tilde{Q}_1(k)^{-1}$, together with (\ref{Q1decomp}) and (\ref{S}), the asymptotics yield the same expressions as above for $M, J_2, b^R_\mu(z_n), b^L_{0}(z_n)$; here $J_1$ requires the $O(z^{-1})$ term in $b^R_\mu(z)$ which needs a higher order calculation. If one instead considers the solution with $F_R(k)$ replaced by $F_R(k)^T$ one obtains $M, J_i, b^L_\mu(z_1)$ to this order.  

We remark that if one uses the alternate solution as in Conjecture \ref{conj}, then the above shows that the solution on $I_L$ fixes $M, J_1, b_\mu^L(z_1)$ (and $b^R_0(z_n)$) and the solution on $I_R$ fixes $M, J_2, b^R_\mu(z_n)$ (and $b^L_0(z_1)$), so taken together these give all the asymptotic quantities. Indeed, this partially motivates our conjecture. On the other hand using the solution on $I_L$ or $I_R$, together with symmetry of $F_L(k)$ or $F_R(k)$, also gives all the asymptotic quantities. 
In Appendix \ref{apen:ernstHorizon} we show how to derive these expression for $b^R_\mu(z_n), b^L_\mu(z_1)$   from general properties of the Ernst potentials.

It is also worth noting that the leading coefficients appearing in Lemma \ref{lemma}  and (\ref{Qtildeasympt}) can be deduced from the above analysis and are $q_{n+1}= -2 \text{diag}(0,1,0)$ and
\be
q_n = \left(  \begin{array}{ccc} -1 & b^R_0(z_n)+2 S_{0}^{~1}  & 0 \\  0 & 2(z_n +\sum_{a=1}^{n} z_a -  S_1^{~1})  & 0 \\ 0  & 0 & 0 \end{array} \right)   \; .
\ee
Therefore, the $p=n$ moduli space equation (\ref{modspace}) only gives the $\mu=0$ component of (\ref{5dbR}).

In the case of a single black hole the above formulas simplify.  In particular, if say $I_H= (z_1, z_2)$, the mass and angular momenta become
\bea
&&M= \tfrac{3\pi}{4}(\ell_H+ \tfrac{1}{2} \Omega_i \chi_i(z_2)),  \label{5dMBH} \\
&&J_i= \tfrac{\pi}{4}\chi_i(z_2),   \label{5dJBH}
\eea
where $\ell_H= z_2-z_1$ and we have chosen a gauge in which $\chi_i(z_1)=0$ (for a single horizon one is always free to do this). In this case we find it convenient to work with the dimensionless parameters
\be
\label{jwdimless}
j_i = J_ i M^{-3/2} \left( \frac{27 \pi}{32} \right)^{1/2},\qquad  \omega_i = \Omega_i M^{1/2}  \left( \frac{8}{3\pi} \right)^{1/2}, \qquad \lambda_H = \frac{3\pi}{4 M} \ell_H,
\ee
where we are of course now assuming $M>0$.   Then (\ref{5dMBH}) gives
\be
\lambda_H=1 - \omega_i j_i \; . \label{lH}
\ee
For any finite axis rods $I_a$, $a=2, \dots, n$  we also define the associated dimensionless parameters
\be
\label{fdimless}
f^a_0 = b^a_0(z_a) \left( \frac{3 \pi}{8 M} \right)^{1/2},\qquad  f^a_1 = b^a_1(z_a) \left( \frac{3 \pi}{8 M} \right), \qquad \lambda_a = \frac{3 \pi}{4 M} \ell_a \; ,
\ee
and $f^L_\mu$, $f^R_\mu$ are similarly defined with $b_\mu^a(z_a)$ replaced by $b^L_\mu(z_1), b^R_\mu(z_n)$ respectively.

For the single black hole cases we study below, we in fact use a method based on a slightly modified version of Conjecture \ref{conj}, which appears  to be more convenient. Given a single horizon rod, the condition (\ref{chialower})  can be thought of as an equation for $J_i$, as follows. Since (\ref{chiaupper}) is automatic we can write $\chi_i^a(z) = \chi^a_i(z_a)+ (z-z_a) g_a(z)$ for some smooth function $g_a(z)$. Then evaluating at $z=z_{a-1}$ and using (\ref{Jhorizon}) we deduce that $J_i = \pi \ell_a g_a(z_{a-1})/4$ which gives a nonlinear equation for $J_i$ (typically $g_a(z_{a-1})$ depends on all the moduli, including $J_i$).  On the other hand, from the asymptotics (\ref{5dLinfty}) and (\ref{5dRinfty}), the $O(z^{-1})$ term in $b^L_0(z)$ and $b^R_0(z)$  gives $J_2$ and $J_1$ respectively. The above asymptotic analysis showed that, for the alternate solution, the computation of these $O(z^{-1})$ terms requires a higher order calculation, which in general will give different formulas for $J_i$ than (\ref{5dJ}). Thus, one can take these asymptotic equations as new equations for $J_i$, instead of those from (\ref{chialower}) described above.  Thus we  reformulate Conjecture \ref{conj} as follows.

\begin{conj} \label{conj2}
Given the $D=5$ single black hole solution (\ref{gensolaxis}), (\ref{gensolhor})  with $F_a(k)$ replaced by $F_a(k)^T$ for $a=1, \dots, n$ and the inequalities $\ell_a>0$, (\ref{FNNpos}) and (\ref{F00neg}),  the moduli space equations (\ref{modspace}) are satisfied if and only if  (\ref{balower}), (\ref{5dLinfty}) and (\ref{5dRinfty}) are imposed (the latter two conditions also determine $M, J_i$ and $\zeta$).
\end{conj}

Of course, for the case of no black hole the two conjectures are equivalent. On the other hand, for multi-black holes, Conjecture \ref{conj2} would need to be revisited. We will not consider this here.

\subsection{Minkowski spacetime}

First consider the rod structure of Minkowski spacetime as in Figure \ref{fig:flat}.
\begin{figure}[h!]
\centering
\subfloat{
\begin{tikzpicture}[scale=1.5, every node/.style={scale=1}]
\draw[very thick](-4,0)--(-0.1,0)node[black,left=3cm,above=.2cm]{$(0,1)$};
\draw[very thick](.1,0)--(4,0)node[black,left=3cm,above=.2cm]{$(1,0)$};
\draw[fill=black] (0,0) circle [radius=.1] node[black,font=\large,below=.1cm]{};
\end{tikzpicture}}
\caption{Rod structure for Minkowski spacetime.}
\label{fig:flat}
\end{figure}
Thus we have two rods $I_L=(-\infty, z_1)$ and $I_R= (z_1, \infty)$.  In this case the matrices which give the general solution in Theorem \ref{th:main} are  $F_L(k)=-C^{-1} P_1(k)^T$ and $F_R(k)= - (C P_1(k))^{-1}$ where  $P_1(k)= X_L(z_1, k)^{-1} X_R(z_1, k) $. 

Let us first consider $z<z_1$. We use the alternate form of the solution obtained by replacing $F_L(k)$ with $F_L(k)^T$ in (\ref{gensolaxis}) as described in Conjecture \ref{conj}. Explicitly, we find
\be
F_L(k)^T  = \left(
\begin{array}{ccc}
 1 & b^R_{0}(z_1) -\frac{b^L_{0}(z_1) b^R_{1}(z_1)}{2(k-z_1) }  & - \frac{b^L_{0}(z_1)}{2(k-z_1) } \\0 & 2(k-z_1)- \frac{b^L_{1}(z_1) b^L_{1}(z_1) }{2(k-z_1) } & -\frac{b^L_{1}(z_1)}{2(k-z_1) }\\ 0 & - \frac{b^R_{1}(z_1)}{2(k-z_1) } & -\frac{1}{2(k-z_1) }
\end{array}
\right)  \; ,
\ee
which gives
\be
h^L_{\mu\nu}(z)  =\left( \begin{array}{cc} - 1& - b^R_0(z_1) \\ -b^R_0(z_1)+ \frac{b^L_0(z_1) b^R_1(z_1)}{2(z-z_1)} & - 2 (z-z_1) \end{array} \right) , \qquad b^L_{\mu}(z)= \left( \begin{array}{c} b^L_0(z_1) \\ b^L_{1}(z_1)\end{array} \right), \qquad z<z_1 \; .
\ee
Imposing our boundary condition $b_\mu^L(z)\to 0$ as $z\to -\infty$ then implies 
\be
b^L_\mu(z_1)=0 \; ,
\ee
which then immediately fixes $b^L_\mu(z)=0$ for $z<z_1$.

The analysis for $z>z_1$ is analogous.  One gets
\be
F_R(k)  = \left(
\begin{array}{ccc}
 1 & -b^L_{0}(z_1) +\frac{b^R_{0}(z_1) b^L_{1}(z_1)}{2(k-z_1) }  & \frac{b^R_{0}(z_1)}{2(k-z_1) }  \\ 0 & -2(k-z_1)+ \frac{b^R_{1}(z_1) b^L_{1}(z_1) }{2(k-z_1) } & \frac{b^R_{1}(z_1)}{2(k-z_1) }\\ 0 &  \frac{b^L_{1}(z_1)}{2(k-z_1) } & \frac{1}{2(k-z_1) }
\end{array}
\right)  \; ,
\ee
and hence using the general solution (\ref{gensolaxis})  the metric data reads
\be
h^R_{\mu\nu}(z)  =\left( \begin{array}{cc} - 1&  b^L_0(z_1) \\ b^L_0(z_1)- \frac{b^L_1(z_1) b^R_0(z_1)}{2(z-z_1)} &  2 (z-z_1) \end{array} \right) , \qquad b^R_{\mu}(z)= \left( \begin{array}{c} b^R_0(z_1) \\ b^R_{1}(z_1)\end{array} \right), \qquad z>z_1 \; .
\ee
Imposing the boundary condition $b^R_\mu(z)\to 0$ as $z \to \infty$ implies
\be
b^R_{\mu}(z_1)=0 \; 
\ee
and thus $b_\mu^R(z)=0$ for $z<z_1$.

We have now fixed all nontrivial parameters. Indeed, given the above parameter conditions the matrices $F_a(k)$ are automatically symmetric in line with Conjecture \ref{conj}. Also notice that the asymptotic conditions for $h_{\mu\nu}^L, h_{\mu\nu}^R$ are both satisfied automatically with $M=J_1=J_2=\zeta=0$. The final solution is simply
\bea
&&h^L_{\mu\nu}(z)  =\left( \begin{array}{cc} - 1& 0\\ 0 & 2 (z_1-z) \end{array} \right) , \qquad b^L_{\mu}(z)= 0, \qquad z<z_1 \\ 
&& h^R_{\mu\nu}(z)  =\left( \begin{array}{cc} - 1&  0 \\ 0 &  2 (z-z_1) \end{array} \right) , \qquad b^R_{\mu}(z)=0, \qquad z>z_1\;.
\eea
This of course is the metric data on axis for Minkowski spacetime (\ref{5dflat}).  As in four dimensions this is a nontrivial result, showing that the only asymptotically flat spacetime in this symmetry class with the same rod structure as Minkowski spacetime is Minkowski spacetime itself. Of course, this is expected and follows from the more general no-soliton theorem for vacuum gravity.

\subsection{Eguchi-Hanson soliton}

Let us now attempt to construct a soliton solution, i.e. a non-trivial solution with no horizon. Of course, we know from the no-soliton theorem for asymptotically flat vacuum solutions that there can be no smooth solution in this case. Nevertheless, it is interesting to see how this emerges from our formalism.

The simplest rod structure without a horizon which is not flat space is given by three axis rods $I_L=(-\infty, z_1)$, $I_B=(z_1, z_2)$ and $I_R= (z_2, \infty)$ with rod vectors $(0,1)$, $v_B=(p,q)$ and $(1,0)$ respectively, where $(p,q)$ are coprime integers. The finite axis rod $I_B$ corresponds to a 2-cycle, or bolt, in the spacetime. The admissibility condition (\ref{admissible}) between adjacent axis rods fixes $p=\pm 1$ and $q=\pm 1$ and without loss of generality we can fix $p=1$ (since $v_B$ is only defined up to a sign). We   also fix $q=1$ which can always be arranged since $v_R$ is only defined up to a sign. Thus we take the  rod vector for $I_B$  to be  $v_B= (1, 1)$.  The rod structure is depicted in Figure \ref{fig:EH}.
\begin{figure}[h!]
\centering
\subfloat{
\begin{tikzpicture}[scale=1.5, every node/.style={scale=1}]
\draw[very thick](-4,0)--(-1,0)node[black,left=2cm,above=.2cm]{$(0,1)$};
\draw[very thick](-.8,0)--(.3,0)node[black,left=.8cm,above=.2cm]{$(1,1)$};
\draw[very thick](.5,0)--(3.5,0)node[black,left=2cm,above=.2cm]{$(1,0)$};
\draw[fill=black] (-.9,0) circle [radius=.1] node[black,font=\large,below=.1cm]{};
\draw[fill=black] (.4,0) circle [radius=.1] node[black,font=\large,below=.1cm]{};
\end{tikzpicture}}
\caption{Rod structure for the simplest soliton spacetime.}
\label{fig:EH}
\end{figure}
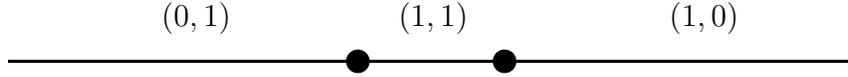
We choose the other independent axial vector to be $u_B=(1,0)$, so the change of basis matrix (\ref{Laxisrod}) is
\be
L_B = \left( \begin{array}{ccc} 1 & 0  & 0 \\ 0 &1 & 0 \\ 0 & -1 & 1  \end{array} \right)  \; .
\ee
The general solution in this case is determined by (\ref{gensolaxis}) where 
\be
F_L(k)= -C^{-1} Q_1(k)^T, \qquad F_B(k)= - P_1(k)^{-1} C^{-1} P_2(k)^T, \qquad F_R(k)= - (C Q_1(k))^{-1}  \label{Fa3rods}
\ee
and $Q_1(k)$ and $P_a(k)$ are given by (\ref{Qadef}) and (\ref{Padef}).
It is convenient to write the general solution on each rod given in Theorem \ref{th:main} with $F_a(k)$ replaced by $F_a(k)^T$ for $a=L,B$ as described in Conjecture \ref{conj}.

First, imposing that the general solution $(h^L_{\mu\nu}(z), b^L_\mu(z))$ on $I_L$ obeys our boundary condition $b^L_\mu(z)\to 0$ as $z\to -\infty$ fixes the constants
\be
b^L_\mu(z_1)= -b^B_\mu(z_2) \;  ,
\ee
with $b_\mu^L(z)|_{z \to z_1}= b^L_\mu(z_1)$ being automatically satisfied (as guaranteed by (\ref{bLRpoints})). 
Next, imposing that $(h^R_{\mu\nu}(z), b^R_\mu(z))$  on $I_R$ obeys $b^R_\mu(z)\to 0$ as $z\to \infty$ fixes
\be
b^R_\mu(z_2) = \left( \begin{array}{c} b_0^B(z_2) \\ -b_1^B(z_2)+ 2(z_1-z_2) \end{array} \right)
\ee
with $b_\mu^R(z)|_{z \to z_2}= b_\mu^R(z_2)$ being automatically satisfied (again, as guaranteed by (\ref{bLRpoints})).  These relations also follow from our general asymptotic analysis (\ref{5dbL}) and (\ref{5dbR}) respectively. 

Finally, the solution $(h^B_{\mu\nu}(z), b^B_\mu(z))$ on $I_B$ satisfies $b_{\mu}^B(z)|_{z \to z_2}=b^B_{\mu}(z_2)$ automatically (as guaranteed by Proposition \ref{prop:consistency}) and $b_{\mu}^B(z)|_{z \to z_1}=0$ fixes
\be
b^B_{\mu}(z_2)= \left( \begin{array}{c}0\\ z_1-z_2\end{array} \right) \; ,
\ee
where we have used the above to simplify this expression.
All parameters have been now fixed except for the axis rod length $\ell_B= z_2-z_1$.  The matrices $F_a(k)$ are now all symmetric demonstrating the validity of Conjecture \ref{conj} in this case. The resulting solution is 
\bea
\label{EHhbL}
&&h^L_{\mu\nu}(z) = \left( \begin{array}{cc} -1 & 0 \\  0 & -\frac{4(z-z_1)(z-z_2)}{2z-z_1-z_2} \end{array} \right), \qquad b^L_{\mu}(z) = \left( \begin{array}{c} 0 \\  - \frac{(z_2-z_1)^2}{2z-z_1-z_2} \end{array} \right),  \qquad z\in I_L \\
\label{EHhbB}
&&h^B_{\mu\nu}(z) = \left( \begin{array}{cc} -1 & 0 \\  0 & -\frac{(z-z_1)(z-z_2)}{z_2-z_1} \end{array} \right), \qquad b^B_{\mu}(z) = \left( \begin{array}{c} 0 \\ \frac{(z-z_1)(z+z_1-2z_2)}{z_2-z_1} \end{array} \right) \; , \quad z \in I_B  \\
\label{EHhbR}
&&h^R_{\mu\nu}(z) = \left( \begin{array}{cc} -1 & 0 \\  0 & \frac{4(z-z_1)(z-z_2)}{2z-z_1-z_2} \end{array} \right), \qquad b^R_{\mu}(z) = \left( \begin{array}{c} 0 \\  - \frac{(z_2-z_1)^2}{2z-z_1-z_2} \end{array}\right) \; , \qquad z \in I_R.
\eea
From the asymptotics $z\to \pm \infty$, we immediately deduce from (\ref{5dLinfty}) or (\ref{5dRinfty}) that 
\be
M=0, \qquad J_1=J_2=0,  \qquad \zeta=-\tfrac{1}{8} (z_2-z_1)^2  \; .
\ee
This corresponds to the unique unbalanced solution which is guaranteed to exist by Theorem \ref{th2}.

We may now analyse regularity of the solution. The metric induced on the bolt (\ref{axisrodmetric}) is
\be
\mathbf{g}_B = - \td t^2+ \ell_B \left( \frac{c_B^2 \td y^2}{1-y^2}+ \frac{1}{4}(1-y^2) (\td x^1)^2 \right), \qquad y = \frac{2z- z_1-z_2}{z_2-z_1}\;,
\ee
where $(t, x^1)$ are coordinates such that $k=\partial_t, u= \partial_{x^1}$ and recall $(k, u)$ is the adapted basis for $I_B$. Recall that $u=m_1$ and hence $x^1$ is a $2\pi$-periodic angle. Therefore, it is clear that the spatial part of the metric on the bolt is a smooth round metric on $S^2$  iff $c_B=1/2$ (indeed, one can check that the conditions for the removal of the conical singularity (\ref{regupper}) and (\ref{reglower}) at $z=z_1$ and $z=z_2$ are satisfied iff $c_B=1/2$). Although this gives a smooth metric on the bolt, this shows that in this case the balance condition $c_B=1$ (\ref{balance}) is violated so there must be a conical singularity at $I_B$. On the other hand, if we impose the balance condition $c_B=1$, then inspecting the metric on the bolt shows that there must be conical singularities at the endpoints of $I_B$.

We have shown that any asymptotically flat solution with a single bolt must have a conical singularity. This is indeed consistent with the no-soliton theorem mentioned above. In fact, in this case it is easy to write down the full solution off axis. It is given by the Eguchi-Hanson soliton (\ref{EH}) where $(\theta, \psi, \phi)$ are Euler angles on $S^3$. The  rods $I_L, I_B$ and $I_R$ can be identified with  $\theta = \pi, R = a$ and $\theta=0$. It then follows that $v_L= \partial_\psi+\partial_\phi$ and $v_R=\partial_\psi-\partial_\phi$ are the $2\pi$-periodic rod vectors on the semi-infinite axes, which implies $\phi^1=(\psi - \phi)/2$ and $\phi^2=(\psi + \phi)/2$. Weyl coordinates $(t, \phi^1, \phi^2, \rho, z)$ for this metric are
\be
\rho = \tfrac{1}{2}\sqrt{R^4-a^4} \sin \theta, \quad z= \tfrac{1}{2}(z_1+z_2) + \tfrac{1}{2} R^2 \cos\theta
\ee
and the corresponding metric data is
\bea
&&g =- \td t^2+ \tfrac{1}{4} R^2 \left(1- \frac{a^4}{R^4} \right)\left[ (1-\cos\theta) \td \phi^1+(1+\cos\theta) \td \phi^2 \right]^2+ \tfrac{1}{4} R^2 \sin^2\theta (\td \phi^1-\td\phi^2)^2\;, \nonumber \\
&&e^{2\nu} = \frac{R^2}{ R^4 - a^4\cos^2\theta} \; .
\eea
Using $a^2= \ell_B$, it is straightforward to show that $g$ gives the same $(h^a_{\mu\nu}, b^a_\mu)$ on each rod as our general solution above (\ref{EHhbL})-(\ref{EHhbR}). In addition $e^{2\nu}$ on the axes and the bolt agrees with our expressions (\ref{nu0}) with $c_L = c_R = 1$ and $c_B = 1/2$. 

\subsection{Myers-Perry solution}

We now consider the simplest rod structure of a single black hole with $S^3$ topology , i.e., the same rod structure as the Myers-Perry solution, see Figure \ref{fig:MP}.
\begin{figure}[h!]
\centering
\subfloat{
\begin{tikzpicture}[scale=1.5, every node/.style={scale=1}]
\draw[very thick](-4,0)--(-1,0)node[black,left=2cm,above=.2cm]{$(0,1)$};
\draw[thick,dashed](-.8,0)--(.3,0)node[black,left=.8cm,above=.2cm]{$H$};
\draw[very thick](.5,0)--(3.5,0)node[black,left=2cm,above=.2cm]{$(1,0)$};
\draw[fill=white] (-.9,0) circle [radius=.1] node[black,font=\large,below=.1cm]{};
\draw[fill=white] (.4,0) circle [radius=.1] node[black,font=\large,below=.1cm]{};
\end{tikzpicture}}
\caption{Rod structure for the Myers-Perry black hole.}
\label{fig:MP}
\end{figure}
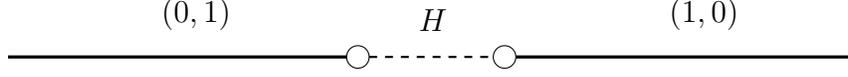
Thus we have three rods $I_L=(-\infty, z_1)$, $I_H= (z_1, z_2)$ and $I_R=(z_2, \infty)$ where $I_H$ is a horizon rod.

The  general solution can be obtained from Theorem \ref{th:main} where the $F_a(z)$ are again given by (\ref{Fa3rods}) (although $X_2(z,k)$ now refers to the horizon rod).  Again, it is convenient to use the alternate form of the solution  with $F_a(k)$ replaced by $F_a(k)^T$ for $a=1,2$ as described in Conjecture \ref{conj2}. The solution depends on the  parameters $(\ell_H, b_\mu^L(z_1), b_{\mu}^R(z_2), \chi_i(z_2), \Omega_i)$ where $\ell_H=z_2-z_1$ and we choose a gauge in which $\chi_i(z_1)=0$.  

From the asymptotics for the solution on $I_L$ and $I_R$ given in (\ref{5dLinfty}) and (\ref{5dRinfty}) we find the mass $M$ and angular momenta $J_i$ are given by (\ref{5dMBH}) and (\ref{5dJBH}), the Ernst potentials are\footnote{Equation (\ref{bLRMP}) also follows from our general asymptotic analysis (\ref{5dbL}) and (\ref{5dbR}). The same result can be established from general considerations using (\ref{bLjumpH}) and (\ref{bRjumpH}), together with the fact that $b_\mu^L(z)=0$ on $I_R$ and $b_{\mu}^R(z)=0$ on $I_L$ (from their definition (\ref{bL}, \ref{bR}) the potentials $b^L_\mu, b^R_\mu$ are constant on $I_R, I_L$ respectively and vanish at infinity).} 
 \be
b^L_\mu(z_1) =\Omega_2  \left( \begin{array}{c} \tfrac{8}{3\pi} M \\  - \frac{4J_1}{\pi} \end{array} \right), \qquad  b^R_\mu(z_2)= \Omega_1  \left( \begin{array}{c} - \tfrac{8}{3\pi} M \\ \frac{4 J_2}{\pi} \end{array} \right)  \; , \label{bLRMP}
\ee
and 
\bea
J_1= \tfrac{16}{9 \pi} M \Omega_1 \left(  M - \tfrac{3}{2} \Omega_2 J_2\right)  \; , \qquad J_2 = \tfrac{16}{9 \pi} M \Omega_2 \left(  M - \tfrac{3}{2} \Omega_1 J_1\right) \; ,  \label{JMP} 
\eea
where we have eliminated $\ell_H$ and $\chi_i(z_2)$ in favour of $M$ and $J_i$ using (\ref{5dMBH}) and (\ref{5dJBH}). It is worth noting that the solutions on $I_L$ and $I_R$ automatically obey $b^L_\mu(z)|_{z\to z_1}= b^L_\mu(z_1)$ and $b^R_\mu(z)|_{z \to z_2}= b^R_\mu(z_2)$ and therefore no further constraints arise from these rods (as guaranteed by (\ref{bLRpoints})).   Observe that (\ref{JMP}) are linear in $J_i$ so we can straightforwardly solve these for $J_i$ and therefore express all parameters in terms of the physical variables $M, \Omega_i$.

 It is convenient to use the dimensionless quantities (\ref{jwdimless}).
Then solving (\ref{JMP})  gives
\be
\label{MPjs}
j_1 = \frac{\omega_1(1-\omega_2^2)}{1-\omega_1^2 \omega_2^2}, \qquad j_2 = \frac{\omega_2(1-\omega_1^2)}{1-\omega_1^2 \omega_2^2}
\ee
and $| \omega_1 \omega_2 | \neq 1$.\footnote{If $|\omega_1\omega_2|=1$ then (\ref{JMP}) imply $\lambda_H=0$ which contradicts our nonextremality assumption.} Thus as promised we can express all quantities in terms of $M, \omega_i$. In particular, eliminating $j_i$ we find that (\ref{lH}) becomes
\be
\lambda_H = \frac{(1-\omega_1^2)(1-\omega_2^2)}{1-\omega_1^2 \omega_2^2}  \; .
\ee
It is now readily verified that the matrices $F_a(k)$ are symmetric in accordance with Conjecture \ref{conj2}.  We have thus fully solved the moduli space equations (\ref{modspace}).

To determine the precise moduli space, we will also need the invariants
\bea
&&\det h^L_{\mu\nu}(z) =  -\frac{2 (z-z_1)(z-z_2)}{\bar{z}_1+z_1 - z}, \qquad \bar{z}_1=\frac{4M}{3\pi} \frac{1-\omega_1^2}{1-\omega_1^2\omega_2^2}\; ,  \qquad z<z_1\; , \\
&&\det h^R_{\mu\nu}(z) = - \frac{2 (z-z_1)(z-z_2)}{z-z_2 +\bar{z}_{2}}, \qquad \bar{z}_{2}=\frac{4M}{3\pi} \frac{1-\omega_2^2}{1-\omega_1^2\omega_2^2} \;  , \qquad z>z_2 \; .
\eea
A smooth Lorentzian metric on $I_L$ requires that the determinant $\det h_{\mu\nu}^L(z)<0$ and is smooth for $z<z_1$ (see (\ref{FNNpos}))
and from the above expression we see this is equivalent to $\bar{z}_1>0$. Similarly, the requirement that $\det h^R_{\mu\nu}(z)$ is smooth and negative on $I_R$ is  equivalent to $\bar{z}_{2}>0$.  The inequalities $\lambda_H>0, \bar{z}_1>0, \bar{z}_2>0$ are equivalent to
\be 
|\omega_i |<1  \;  . \label{omMPmodspace}
\ee
This fully constrains the moduli space of solutions which is simply given by (\ref{omMPmodspace}) and $M>0$.  One can show (\ref{omMPmodspace}) implies 
\be
|j_1|+ |j_2|<1  \; ,
\ee
which is a well-known inequality for the Myers-Perry black holes.

Now we turn to the solution $(\gamma_{ij}(z), \chi_i(z))$ on the horizon rod $z_1<z<z_2$ which can be deduced from (\ref{gensolhor}). 
Writing the parameters in terms of $M, \omega_i$ as above, we find that both $\chi_i(z)|_{z\to z_1}=0$ and (\ref{5dJBH}) are automatically satisfied (as they must be).  Furthermore, using (\ref{5dkappa}), we find that removal of the conical singularities of the horizon metric at the endpoints $z=z_1, z_2$ imposes no further constraints and fixes the surface gravity to be
\be
\kappa = \sqrt{  \frac{3\pi}{8M} (1- \omega_1^2)(1-\omega_2^2)}  \; .
\ee
The horizon topology is of course $S^3$ with $m_2=0$ at $z=z_1$ and $m_1=0$ at $z=z_2$. Notice that the moduli space (\ref{omMPmodspace}) is equivalent to the nonextremality condition $\kappa>0$.

 It is straightforward to check that the metric data for above solution agrees precisely with the Myers-Perry solution restricted to the $z$-axis, and the parameter region $|\omega_i|<1$ we have derived agrees with the full moduli space of non-extremal Myers-Perry black holes (of course, this includes 5d Schwarzschild for $\omega_i=0$).  It is interesting to note that by combining (\ref{JMP}) we obtain the thermodynamic identity recently obtained by integrating the sigma model equation over the boundary of the orbit space~\cite{Kunduri:2018qqt}. Thus our present method  leads to a refinement of these identities.

\subsection{Black ring}
We now consider the rod structure of the black ring as depicted in Figure \ref{fig:BR}. 
 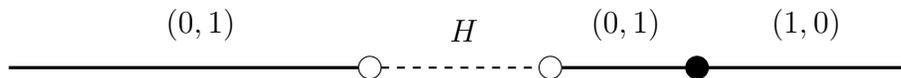
\begin{figure}[h!]
\centering
\subfloat{
\begin{tikzpicture}[scale=1.5, every node/.style={scale=1}]
\draw[very thick](-4,0)--(-0.9,0)node[black,left=2.1cm,above=.2cm]{$(0,1)$};
\draw[thick,dashed](-.7,0)--(.7,0)node[black,left=1.0cm,above=.2cm]{$H$};
\draw[very thick](0.9,0)--(2.0,0)node[black,left=.8cm,above=.2cm]{$(0,1)$};
\draw[very thick](2.2,0)--(4,0)node[black,left=1.4cm,above=.2cm]{$(1,0)$};
\draw[fill=white] (-.8,0) circle [radius=.1] node[black,font=\large,below=.1cm]{};
\draw[fill=white] (.8,0) circle [radius=.1] node[black,font=\large,below=.1cm]{};
\draw[fill=black] (2.1,0) circle [radius=.1] node[black,font=\large,below=.1cm]{};
\end{tikzpicture}}
\caption{Rod structure for the black ring}
\label{fig:BR}
\end{figure}
Thus we have four rods $I_L=(-\infty, z_1)$, $I_H=(z_1, z_2)$, $I_D=(z_2, z_3)$ and $I_R=(z_3, \infty)$, where $I_H$ is a horizon rod and $I_D$ is an axis rod with rod vector $v_D=(0,1)$. The topology of the horizon is $S^2\times S^1$ and the finite axis rod $I_D$ lifts to a noncontractible 2-disc in spacetime. We use the adapted basis $\tilde{E}_A=(k,m_1,m_2)$ for $I_D$, i.e. $u_D=(1,0)$, so the change of basis matrix $L_D$ (\ref{Laxisrod}) is simply the identity matrix.

The general solution is given by Theorem \ref{th:main} and once again it is convenient to use the alternate form of the solution where $F_a$ is replaced with $F_a^T$ for $a=L,H,D$ as described in Conjecture \ref{conj2}. The solution depends on the parameters $(\ell_H, \ell_D, b_\mu^L(z_1), \chi_i(z_2), b^D_\mu(z_3), b^R_\mu(z_3), \Omega_i)$ where $\ell_H= z_2-z_1, \ell_D= z_3-z_2$ and we choose a gauge in which $\chi_i(z_1)=0$.

From the asymptotics for the solution on $I_L$ and $I_R$ given in (\ref{5dLinfty}) and (\ref{5dRinfty}) we find the mass $M$ and angular momenta $J_i$ are given by (\ref{5dMBH}) and (\ref{5dJBH}) and the Ernst potentials are
\be
b^L_\mu(z_1)+b_\mu^D(z_3) =\Omega_2  \left( \begin{array}{c} \tfrac{8}{3\pi} M \\  - \frac{4J_1}{\pi} \end{array} \right), \qquad  b^R_\mu(z_3)=  \Omega_1  \left( \begin{array}{c} - \tfrac{8}{3\pi} M \\ \frac{4 J_2}{\pi} \end{array} \right) , \;  \label{bLRring}
\ee
where we have eliminated $\ell_H$ and $\chi_i(z_2)$ in favour of $M$ and $J_i$ using (\ref{5dMBH}) and (\ref{5dJBH}). The solution on $I_L$ and $I_R$ automatically obeys $b^L_\mu(z)|_{z\to z_1}= b^L_\mu(z_1)$ and $b^R_\mu(z)|_{z \to z_2}= b^R_\mu(z_2)$ and therefore no further constraints arise from these rods (as must be from (\ref{bLRpoints})).

The asymptotics of the solution also give nontrivial equations for $J_i$ which in terms of dimensionless variables introduced in (\ref{jwdimless}) and (\ref{fdimless}) are given by
\begin{equation}
    \begin{gathered}
    \label{jeqnsBR}
    j_1 = \omega_1( 1 + \lambda_D + j_2(f_0^D - \omega_2))\;,\\
    j_2 = \omega_2 - \omega_1(j_1 \omega_2 + f^D_1) + f^D_0 (\omega_1 j_1  - 2) \; ,
    \end{gathered}
\end{equation}
where the previous relations have been used to eliminate variables in favour of $j_i, \omega_i, f^D_\mu, \lambda_D$. These equations correspond to (\ref{JMP}) for the Myers-Perry solution. 

Next consider $I_D$. We find that $b^D_\mu(z)|_{z \to z_3} = b^D_\mu(z_3)$ is automatically satisfied (see Proposition \ref{prop:consistency}), however $b^D_\mu(z)|_{z \to z_2} = 0$ gives the new constraints 
\be
\label{fDBR}
f^D_\mu = \frac{j_2 \lambda_D}{D}\begin{pmatrix} 1-\omega_1^2 \\ -\omega_1 \lambda_D \end{pmatrix}, \quad D \equiv (1-\omega_1 j_1)^2 -j_2(f^D_0 + \omega_1 (\omega_1 j_2 + f^D_1)).
\ee
For these equations to be well-defined $D$ must be nonzero. These relations were obtained by evaluating $\lim_{z\to z_3} b^D_\mu(z)$, so if $D=0$ then the numerators $j_2 \lambda_D (1-\omega_1^2)$ and $-\omega_1 j_2\lambda_D^2$ would have to vanish as well to ensure that $b^D_\mu(z_3)$ was well-defined. This implies that $j_2 = 0$ (since $\lambda_D > 0$), which combined with (\ref{lH}) and $D=0$ implies that $\lambda_H=0$. Since $\lambda_H>0$ this means that $D\neq0$ and so (\ref{fDBR}) are well-defined.

Equations (\ref{jeqnsBR}) and (\ref{fDBR}) are significantly more complicated than the corresponding parameter constraints for the Myers-Perry solution (\ref{MPjs}). Therefore it is instructive to first consider the $S^1$ singly spinning case.

\subsubsection{Singly spinning black ring}
The $S^1$ spinning black ring corresponds to setting $j_2 = 0$ in the above equations. In this case (\ref{fDBR}) simply gives that $f^D_\mu = 0$. Substituting this back into the equations for $j_i$ (\ref{jeqnsBR}) gives
\begin{equation}
\label{BRS1j1}
\begin{gathered}
    j_1 = \omega_1 (1+\lambda_D) \\
    \omega_2(1 - \omega_1j_1) = 0.
\end{gathered}
\end{equation}
The first of these two equations gives $j_1$ and the second implies that $\omega_2 = 0$ since $1-\omega_1 j_1 = \lambda_H \neq 0$.

This gives the solution for the general unbalanced $S^1$ spinning black ring parameterised in terms of $(M,\omega_1,\lambda_D)$. Note that the matrices $F_a(k)$ are now automatically symmetric in accordance with Conjecture \ref{conj2}. The horizon rod length $\lambda_H = 1-\omega_1^2(1+\lambda_D)$ and so $\lambda_H>0$ gives the constraint
\begin{equation}
    \omega_1^2< \frac{1}{1+\lambda_D},
\end{equation}
which together with the conditions $M>0$ and $\lambda_D>0$ determines the moduli space of unbalanced solutions. It can then be checked that (\ref{FNNpos}) and (\ref{F00neg}) are satisfied automatically and so impose no further constraints.

Next consider conical singularities on $I_D$. The balance condition (\ref{balance}) and regularity condition at $z=z_3$ (\ref{regupper}) is equivalent to
\begin{equation}
    (1-\omega_1^2)^2 - \lambda_D \omega_1^2(2-\omega_1^2)=0
\end{equation}
which implies  $\omega_1^2>0$ and
\begin{equation}
    \lambda_D = \frac{(1-\omega_1^2)^2}{\omega_1^2(2-\omega_1^2)} \; .
\end{equation}
Substituting this back in we obtain
\be
\lambda_H = \frac{1-\omega_1^2}{2-\omega_1^2},
\ee
which gives the moduli space of the balanced solution as
\begin{equation}
    M>0, \quad 0<\omega_1^2<1.  \label{BRS1moduli}
\end{equation}
In addition, the expression for $j_1$ (\ref{BRS1j1}) now takes the simple form
\begin{equation}
    j_1 = \frac{1}{\omega_1(2-\omega_1^2)}.
\end{equation}
Extremising this over the moduli space (\ref{BRS1moduli}) gives the well-known inequality $|j_1|\geq \sqrt{27/32}$.
Finally, condition (\ref{5dkappa}) for the removal of conical singularities on $I_H$ imposes no further constraints and fixes the surface gravity to be
\begin{equation}
    \kappa = \sqrt{\frac{3\pi}{8M}}\frac{\sqrt{1-\omega_1^2}}{|\omega_1|}.
\end{equation}

\subsubsection{Doubly spinning black ring}
Now we consider the doubly spinning solution corresponding to $j_2 \neq 0$. In this case it is no longer straightforward to solve (\ref{jeqnsBR}) and (\ref{fDBR}) in terms of any of the variables already defined.  Firstly, using (\ref{jeqnsBR}) and (\ref{fDBR}), together with the  balance condition (\ref{balance}) on $I_D$ and the condition for removal of the conical singularity on $I_D$ at $z=z_3$  (\ref{regupper}), one can show that $\omega_2=0$ implies $j_2=0$ (here we are also assuming $\lambda_H,\lambda_D>0$). Thus we deduce $\omega_2 \neq 0$.

It turns out it is convenient to define a new parameter $t$, using the denominator $D$ defined in (\ref{fDBR}), by
\be
\label{tBR}
t = \frac{\omega_2 D}{j_2 \lambda_D}.
\ee
Note that $t\neq 0$. This gives
\begin{equation}
    f^D_\mu = \frac{\omega_2}{t}\begin{pmatrix} 1-\omega_1^2 \\ -\omega_1 \lambda_D \end{pmatrix}.
\end{equation}
Now we can solve (\ref{jeqnsBR}) for $j_i$\footnote{Using (\ref{jeqnsBR}) and (\ref{tBR}) one can show that the denominator of (\ref{jssolBR}) being zero is incompatible with $\lambda_H,\lambda_D>0$ and the conditions for the removal of conical singularities (\ref{balance}) and (\ref{regupper}).}
\begin{equation}
\begin{gathered}
\label{jssolBR}
j_1 = \frac{\omega_1(t^2(1+\lambda_D) - \omega_2^2(t-1+\omega_1^2)(t-2+\omega_1^2(2+\lambda_D))}{t^2 - \omega_1^2 \omega_2^2 (t-1+\omega_1^2)^2}  \; , \\
j_2 = \frac{t \omega_2(t-2+\omega_1^2)(1-\omega_1^2(1+\lambda_D))}{t^2 - \omega_1^2 \omega_2^2 (t-1+\omega_1^2)^2}  \; ,
\end{gathered}
\end{equation}
and then (\ref{tBR}) for $\omega_2$\footnote{The denominator of (\ref{w2BR}) can never vanish since the denominator of (\ref{jssolBR}) is nonzero and $j_2 \neq 0$.}
\be
\label{w2BR}
\omega_2^2 = \frac{t^2(1-\omega_1^2 - \lambda_D(t-2+2\omega_1^2))}{(t-2 + 2\omega_1^2)(1-\omega_1^2(1+\lambda_D))}.
\ee
This gives two branches of solutions corresponding to either $\omega_2> 0$ or $\omega_2<0 $. We have now solved for the generic\footnote{As explained above,  a couple of possible special cases were  ruled out using the balance condition.} unbalanced doubly spinning black ring solution parameterised in terms of $(M,\omega_1,t,\lambda_D)$. The matrices $F_a$ are now indeed symmetric as expected from Conjecture \ref{conj2}.

Now consider the possible conical singularities on $I_D$. To remove this the balance condition (\ref{balance}) and the regularity condition  (\ref{regupper}) at $z=z_3$ must be satisfied, which in this case reduces to
\be
\label{conicalBR}
\omega_1^4 \lambda_D + \omega_1^2(1+t \lambda_D) - 1 = 0.
\ee
Note that this implies that $\omega_1 \neq 0$. Solving this for $\lambda_D$ one finds\footnote{If $t+\omega_1^2=0$, using (\ref{conicalBR}) and (\ref{w2BR}) one can show that the denominator of (\ref{jssolBR}) is zero which is a contradiction. Therefore (\ref{ldBR}) is the unique solution of (\ref{conicalBR}).}
\be
\label{ldBR}
\lambda_D = \frac{1-\omega_1^2}{\omega_1^2(t+\omega_1^2)}.
\ee
The expressions (\ref{jssolBR}), (\ref{w2BR}) and (\ref{lH}) can be simplified with this result and one finds
\be
j_1 = \frac{1+(t-1+2\omega_1^2)(t-1+\omega_1^2)}{\omega_1(t+\omega_1^2)^2}, \qquad
j_2 = \frac{\omega_2(t-1+\omega_1^2)(t-2+2\omega_1^2)}{t(t+\omega_1^2)^2},
\ee
\be
\omega_2^2 = \frac{t^2(\omega_1^4 - (t-2)(1-\omega_1^2))}{\omega_1^2(t-1+\omega_1^2)(t-2+2\omega_1^2)},
\ee
\be
\label{lHlDBR}
\lambda_H = \frac{(1-\omega_1^2)(t-2+\omega_1^2)}{\omega_1^2(t+\omega_1^2)} = \lambda_D (t-2+\omega_1^2).
\ee
This gives the balanced doubly rotating solution, however one still needs to find the bounds on the parameters $(M,\omega_1,t)$. These turn out to be given by $M>0$,
\begin{equation}
    \label{BRmoduli}
   0<1-\omega_1^2<t-1, \qquad t < ((1-\omega_1^2)+(1-\omega_1^2)^{-1}).
\end{equation}
Positivity of the rod lengths $\lambda_H,\lambda_D>0$ is equivalent to the first condition and the second condition then corresponds to $\omega_2^2> 0$. The conditions (\ref{FNNpos}) and (\ref{F00neg}) are then automatically satisfied and impose no further constraints.

Note that the limit curve given by $t \to ((1-\omega_1^2)+(1-\omega_1^2)^{-1})$ corresponds to the $\omega_2 \to 0$ (or equivalently $j_2 \to 0 $) singly spinning limit. It turns out that taking this limit one recovers the results of the previous section on the $S^1$ spinning ring as one might expect. Therefore, although the original definition of $t$ (\ref{tBR}) only holds when $j_2 \neq0$, this parameterisation can be extended to cover the singly spinning case as well.

Finally consider the horizon rod $I_H$. Using the parameters $(M,\omega_1,t)$, we find that both $\chi_i(z)|_{z \to z_1}=0$ and (\ref{5dJBH}) are automatically satisfied (as they must be). There are no further constraints from removing conical singularities at the endpoints of $I_H$ since (\ref{5dkappa}) is also satisfied automatically for a surface gravity given by
\be
\kappa = \sqrt{\frac{3 \pi}{8 M (t-1+\omega_1^2)}}\frac{(1-\omega_1^2)(t-2 + \omega_1^2)(t+\omega_1^2)}{|\omega_1|(t-2 + 2\omega_1^2)}.
\ee
From this one can explicitly see that the limit curve $\omega_1\to \sqrt{2-t}$, which is a boundary of the moduli space of solutions, corresponds to extremal solutions as one might expect. On the other hand although $\kappa = 0$ as $\omega_1 \to 1$, this corresponds to a singular solution since $\lambda_D\to0$ in this limit.

We have now constructed the most general regular solution on the axes and horizon with the given rod structure. We will now show that our solution maps exactly to the Pomeransky-Sen'kov solution for the balanced doubly rotating black ring. Chen, Hong and Teo ~\cite{Chen:2011jb} present the solution for $\omega_1 >0,\omega_2>0$ in terms of the parameters $(\chi,\mu,\nu)$, satisfying 
\begin{equation}
    \label{CHTmoduli}
    \chi>0,\qquad 0<\nu<\mu<1.
\end{equation}
Note that we take $\nu \neq \mu$ since we are considering non-extremal solutions and $\nu \neq0$ since we are considering $\omega_2\neq0$. To find an expression for $t$ in terms of these variables, first use (\ref{lHlDBR}) to give 
\be
t = (2-\omega_1^2)+\frac{\lambda_H}{\lambda_D}.
\ee
Using this, combined with the expressions for $M, \Omega_1, \ell_H, \ell_D$ from the known solution gives
\be
M = \frac{3 \pi \chi^2(\mu+\nu)}{(1-\mu)(1-\nu)},\qquad 
\omega_1^2 = \frac{2(\mu + \nu)}{(1+\mu)(1+\nu)}, \qquad
t = \frac{2(1+\mu^2)(1-\nu)}{(1-\mu^2)(1+\nu)}.
\ee
Inverting these relations for $\chi^2,\mu$ and $\nu$ gives
\be
\chi^2 = \frac{2M}{3 \pi}\frac{(1-\omega_1^2)}{\omega_1^2}, \qquad
\mu = \frac{x-(1-\omega_1^2)}{x+(1-\omega_1^2)}, \qquad
\nu = \frac{1-x}{1+x},
\ee
where
\be
    x = \sqrt{(1-\omega_1^2)(t-1+\omega_1^2)}.
\ee
A short calculation also demonstrates that these expressions give bijections between the subspaces defined by (\ref{CHTmoduli}) and (\ref{BRmoduli}) restricted to $\omega_1> 0, \omega_2> 0$.  One can also show that the metric data on the axis and horizon rods agrees precisely under this map. Therefore, we deduce that the Pomeransky-Sen'kov black ring is the most general regular solution within this class of rod structures (for $\omega_2=0$ see the singly spinning case above).

\section{Discussion}
\label{sec:disc}

In this paper we have considered the classification of $D=4,5$ asymptotically flat stationary vacuum black hole spacetimes that admit $D-3$ commuting axial Killing fields. We have   developed a general method based on integrability of the Einstein equations for this class of spacetimes. In particular, we have presented a general solution for the metric and associated Ernst and twist potentials on each axis and horizon component, see Theorem \ref{th:main}. This solution depends on a number of geometrically defined moduli which obey a set of algebraic equations and inequalities. Generically the solutions possess conical singularities on the axes and  correspond to the moduli space of solutions guaranteed to exist in Theorem \ref{th2}.  However, by imposing that the axis and horizon metric is free of conical singularities we obtain, at least in principle, the moduli space of {\it regular} black hole solutions in this class for any given rod structure (which may be empty depending on the rod structure).

In practice the  equations which define the moduli spaces increase in complexity as one increases the number of rods. Therefore an analysis of the general solution remains out of reach. To this end, it would be interesting to prove Conjecture \ref{conj} and \ref{conj2}, as this may lead to a better understanding of the moduli space equations. Nevertheless, we have studied various special cases in which it is possible to fully solve the moduli space equations. In particular, for rod structures corresponding to the Kerr black hole, the Myers-Perry black holes and the known doubly spinning black rings, we find that the resulting moduli space of regular solutions coincide precisely with that of the known solutions. Thus our analysis, together with Theorem \ref{th2}, provides a proof of uniqueness of these solutions within their class of rod structures (of course, for the Kerr case we recover the classic no-hair theorem). These proofs are constructive in the sense that we also obtain the metric and associated Ernst or twist potentials  on the axes and horizon.

More interestingly, our general solution can be used to determine the (non)existence of new types of regular black hole solution in this context.   For $D=5$ an open question is whether a regular vacuum black lens exists. We are currently investigating this question for the simplest rod structure compatible with a $L(n,1)$ horizon topology (i.e. a single finite axis rod).  It turns out that the analysis of the moduli space equations is much more complicated than the black ring case. We have proven that the singly spinning case $J_2=0$ must always possess a conical singularity on axis. This explains why the previously constructed singly spinning solutions did not lead to regular black lens spacetimes~\cite{Chen:2008fa, Tomizawa:2019acu}. The analysis of the doubly spinning case is far more involved and details will be presented in a forthcoming paper.

By construction, we have obtained the general solution only on the boundary of the orbit space, i.e. on the axis and horizon rods. On the other hand, Theorem \ref{th2} shows that for given boundary data, there exists a unique solution that is smooth everywhere away from the axes.  An interesting question is to write down this full solution explicitly, given our boundary solution. We expect that further methods from integrability theory will be required for this, e.g., by employing the technique used for the Ernst equations~\cite{MS}. In particular, this would be useful to analyse  regularity of the full solution at the axes (i.e. to show that the metric components are even functions of $\rho^2$).   In any case, we anticipate that this regularity issue will likely be satisfied automatically as in four-dimensions (even for conically singular solutions). Therefore, given the general regular boundary solution, we expect a unique spacetime that is regular everywhere on and outside the axes and horizon to exist. Thus the analysis in this paper should be sufficient to determine the full moduli space of regular black hole solutions.

It would be interesting to develop our method to study the analogous classification problem for other types of boundary conditions. In particular, for $D=5$ one can have asymptotically Kaluza-Klein (KK) or Taub-NUT (TN) vacuum solutions.  This could be of interest, as in these cases, the space of regular solutions is richer since one can have regular soliton spacetimes (e.g. $\mathbb{R}\times$ Euclidean Schwarzschild and the KK monopole, for KK and TN asymptotics respectively). Presumably our analysis can be adapted to these cases, although clearly one would have to revisit the solution of the spectral equations near infinity.

Our method is based on the existence of an auxiliary linear system whose integrability condition is the vacuum Einstein equations for spacetimes in this symmetry class. 
It seems likely that this method could be employed in other theories of gravity which are integrable for spacetimes with $D-2$ commuting Killing fields.  For example, it is well-known that this is the case for $D=4$ Einstein-Maxwell equations and an analogous inverse scattering method has been developed~\cite{Neugebauer:2003qe}. This was recently used to construct the general charged, rotating, double-black hole solution~\cite{Hennig:2019knn}. 

More generally, any theory which reduces to a two-dimensional sigma-model with coset target space is integrable in this sense.  A notable example is $D=5$ minimal supergravity (Einstein-Maxwell-CS theory)~\cite{Figueras:2009mc}. This theory could be particularly interesting  to study as it is already known to contain a rich class of regular spacetimes with these symmetries.  Besides the well-known charged versions of the Myers-Perry black holes and black rings, this theory also admits positive energy soliton solutions (a.k.a microstate geometries)~\cite{Bena:2007kg}, supersymmetric black lenses, and  black holes with nontrivial topology in the DOC (2-cycles)~\cite{Kunduri:2014iga, Kunduri:2014kja, Tomizawa:2016kjh, Horowitz:2017fyg, Breunholder:2017ubu, Breunholder:2018roc}. Recently a complete classification of supersymmetric spacetimes in this class  was obtained revealing an infinite class of new black holes, black lenses and rings in spacetimes with nontrivial 2-cycles~\cite{Breunholder:2017ubu}.  It would be very interesting to provide a complementary classification based on integrability as this would also capture the much larger moduli space of nonsupersymmetric solitons and black holes.
\\
 
\noindent {\bf Acknowledgements}. JL is supported by a Leverhulme Trust Research Project Grant. FT is supported by an EPSRC studentship. We would like to thank Harry Braden and Hari Kunduri for helpful discussions.

\appendix 

\section{Rod structure of Gibbons-Hawking solitons}
\label{apen:GH}

In Section \ref{sec:prelim} we showed that the Eguchi-Hanson soliton can be interpreted as an asymptotically Minkowski solution which is regular everywhere except for a conical singularity on its bolt. In particular, it gives a rod structure which satisfies the admissibility condition (\ref{admissible}) and hence gives the corresponding solution that is guaranteed to exist in Theorem \ref{th2}. It is natural to wonder whether the more general Gibbons-Hawking solitons can be similarly interpreted. In fact, we find that within this class of solutions, the only case which gives an admissible rod structure is the Eguchi-Hanson soliton.

The Gibbons-Hawking solitons are
 \be
 \td s^2_{\text{GH}}=   -\td t^2+ H^{-1} (\td \tau + \chi_i \td x^i)^2 +H \td x^i \td x^i \; , \qquad H= \sum_{a=1}^n \frac{1}{|x - p_a|}  \; ,
 \ee
 where $x^i$ are Cartesian coordinates on $\mathbb{R}^3$,  $p_a \in \mathbb{R}^3$ are constants and $\chi$ is determined by $\td \chi= \star_3 \td H$. We assume $n>1$ and note that  for $n=2$ this is the Eguchi-Hanson soliton (\ref{EH}) in different coordinates (for $n=1$ this of course Minkowski spacetime).  If we take the $p_a= (0,0,z_a)$  collinear then the metric has biaxial symmetry and in cylindrical coordinates reads
 \bea
 &&\td s^2_{\text{GH}}=   -\td t^2+ H^{-1} (\td \tau + \chi \td \phi)^2+ H \rho^2 \td \phi^2 +H ( \td \rho^2+ \td z^2)\;, \nonumber \\
 && H= \sum_{a=1}^n \frac{1}{\sqrt{\rho^2 +(z-z_a)^2}}, \qquad \chi = \sum_{a=1}^n \frac{z-z_a}{\sqrt{\rho^2 +(z-z_a)^2}}  \; .
 \eea
 Observe that this metric is also in Weyl coordinates.
 As is well-known, if $(\tau, \phi)$ are identified as Euler angles on $S^3$ (i.e. such that the orbits of $\partial_\phi\pm \partial_\tau$ are independently $2\pi$-periodic) this gives a smooth ALE metric with $S^3/\mathbb{Z}_n$  topology at infinity and any curve between the centres $p_a$ corresponds to a 2-cycle (or bolt).   
 
 On the other hand, one can identify $(\tau, \phi)$ such that the topology at infinity is $S^3$ resulting in an asymptotically Minkowski spacetime.  Explicitly, as $r= |x| \to \infty$ we have $H\sim n/r$ and $\chi \sim n \cos\theta$, where $(r, \theta)$ are standard polar coordinates on $\mathbb{R}^3$, so
 \be
 \td s^2_{\text{GH}} \sim - \td t^2+  \td R^2+ \frac{1}{4}R^2\left[ ( \td \psi+ \cos \theta \td \phi)^2+ \td \theta^2+ \sin^2\theta \td \phi^2\right] \; ,
 \ee
 where we have defined coordinates $\psi= \tau/n$ and  $R^2= 4n r$. Thus identifying $(\theta, \psi, \phi)$ to be Euler angles on $S^3$ gives an asymptotically Minkowski spacetime. In particular, the rod vectors with $2\pi$-periodic orbits on the two semi-infinite axes $\theta=0$ and $\theta=\pi$ are $v_R= \partial_\phi- \partial_\psi$ and $v_L= \partial_\phi+ \partial_\psi$ respectively. Let us compute the rod structure for this asymptotically flat vacuum solution.

 It is clear there are $n+1$ axis rods $I_1= (-\infty, z_1)$,  $I_a=( z_{a-1}, z_a)$ for $a=2, \dots, n$ and $I_{n+1}= (z_n, \infty)$. The rod vector on each rod is a multiple of
 \be
\tilde{v}_a = \partial_\phi- \chi_a\partial_\tau
 \ee
 where
 \be 
 \chi_a \equiv \chi|_{I_a} = \sum_{b=1}^n \text{sign}(z-z_b) = 2 (a-1)-n
 \ee
 for $a=1, \dots, n+1$. For $a=1$ and $a=n+1$ this expression reduces to $v_L$ and $v_R$ respectively and hence is correctly normalised.  With respect to the $2\pi$-periodic basis $(v_R, v_L)$ the rod vectors are
 \be
 \tilde{v}_a= \left( \frac{a-1}{n} ,1- \frac{a-1}{n} \right) 
 \ee
 so $\tilde{v}_1= v_1= (0,1)$ and $\tilde{v}_{n+1}= v_{n+1}= (1,0)$ as previously noted. However, for $a=2, \dots, n$ rod vectors must be rescaled to ensure they have integer entries with respect to a $2\pi$-periodic basis.  
Thus for $a=2, \dots, n$ the rod vectors are 
 \be
 v_a= \frac{1}{\text{gcd}(a-1, n)}\left( a-1, n- a+1 \right) \;,
 \ee
 where the prefactor is included to ensure the components are coprime and hence $v_a$ has $2\pi$-periodic orbits.  
 
 We will now examine whether this rod structure satisfies the admissibility condition (\ref{admissible}).
 In general we have
 \be
 v_2= (1, n-1) \; , \qquad  v_n= (n-1, 1) \; ,
 \ee
 so $\det (v_1, v_2)=-1$ and $\det( v_n, v_{n+1})=-1$ satisfy (\ref{admissible}).  Therefore, if $n=2$, we have an admissible rod structure $v_1=(0,1)$, $v_2=(1,1)$, $v_3=(1,0)$. This is the Eguchi-Hanson soliton discussed in the main text (\ref{EH}).  However, for $n>2$ and $a=2, \dots, n-1$ we have
 \be
 \det (v_a, v_{a+1} ) =  -\frac{n}{\text{gcd}(a-1, n)\text{gcd}(a, n)} \;,
 \ee
 which is never equal to $\pm 1$ and hence the admissibility condition (\ref{admissible}) is always violated for $n>2$. Instead, for these cases the corners of the orbit spaces $z_2, \dots, z_{n-1}$ are orbifold singularities.
 
\section{Geometry near corners of orbit space}
\label{app:corners}

\subsection{Intersection of axes}

Here we consider the geometry of a $D=5$ spacetime near a fixed point of the $U(1)^2$-action, i.e., we consider the geometry near a corner of the orbit space $z=z_a$ where two consecutive axis rods $I_a$ and $I_{a+1}$ meet. 

Then, as shown in Section \ref{sec:axes}, smoothness of the metric on $I_{a}$ at $z=z_a$ requires (\ref{regupper}), whereas smoothness of the metric on $I_{a+1}$ at $z=z_a$ requires (\ref{reglower}) with $a$ replaced by $a+1$, i.e., ${h^{a+1}}'(z_{a})^2/h^{a+1}_{00}(z_{a})=-4c_{a+1}^2$.  On the other hand, for any axis rod 
\be
h^a_{00}(z)=  \mathring{g}_{AB}k^A k^B
\ee
is simply the squared norm of the stationary Killing field $k$ on the axis. Therefore, $h_{00}^a(z_a)=h^{a+1}_{00}(z_a)$ and hence eliminating the norm of $k$ between the aforementioned regularity conditions we deduce that
\be
c_a^{-1} {h^a}'(z_a)= - c_{a+1}^{-1} {h^{a+1}}'(z_a) \; , \label{regcorner}
\ee
where in order to fix the sign we have used the fact that ${h^a}'(z_a)>0$ and ${h^{a+1}}'(z_a)<0$ (these follow from $h^a<0$ in the interior of $I_a$).  

Finally, observe that using (\ref{nu0}) the condition (\ref{regcorner}) is equivalent to continuity of $| z-z_a| e^{2\mathring{\nu}}$ at $z=z_a$.  In fact this continuity condition for the conformal factor $e^{2\nu}$ has been previously proven in~\cite{Tod:2013ska}.

\subsection{Intersection of horizon and axis}

We now consider the geometry where a horizon rod $I_a$ meets an axis rod $I_{a+1}$. In particular, the geometry on the axis  corresponding to $I_{a+1}$ (\ref{axisrodmetric}) is a $(D-2)$-dimensional Lorentzian spacetime that must have a regular $(D-3)$-dimensional horizon as $z\to z_a$ corresponding to where the full horizon intersects the axis corresponding to $I_{a+1}$. We will  now compute the surface gravity of this `axis horizon' $z=z_a$, which must of course coincide with the surface gravity of the full horizon.   

For $D=5$,  the Killing field null on the horizon $\xi$ restricted to the axis rod $I_{a+1}$ is $\xi= k+\Omega u_{a+1}$ where $(k, u_{a+1})$ is the adapted basis of $I_{a+1}$ and $\Omega$ is a constant angular velocity. Therefore, the metric on this component of the axis (\ref{axisrodmetric}) must be of the form
\bea
\mathbf{g}_{a+1} &=& -\frac{c_{a+1}^2\td z^2}{h^{a+1}(z)} + (p_1 (z-z_a) + O((z-z_a)^2) (\td x^0)^2 \nonumber  \\ &+& O(z-z_a) \td x^0  (\td {x}^1- \Omega  \td x^0) + (p_2+O(z-z_a)) (\td {x}^1- \Omega  \td x^0)^2 \; ,  \label{axisnearhorizon}
\eea
 as $z\to z_a^+$, 
where we choose adapted coordinates such that $k= \partial/ \partial {x^0}, u_{a+1}=\partial/\partial {x^1}$. The expansions of the metric components follow from smoothness, together with $\xi$ being null on the axis horizon and $u_{a+1}$ being tangent to the axis horizon. Here $p_1<0$, $p_2>0$ are constants related to the metric components ($p_1=0$ would correspond to an extremal horizon which we do not consider here).  It follows that the determinant $h^{a+1}(z) = p_1p_2 (z-z_a) + O((z-z_a)^2)$ and hence defining $\epsilon^2= z-z_a$, the first two terms in (\ref{axisnearhorizon}) approach the Rindler metric 
\be
-\frac{4c_{a+1}^2}{p_1 p_2} \left( \td \epsilon^2- \kappa^2 \epsilon^2 (\td x^0)^2 \right) \; ,
\ee
as $\epsilon \to 0$, with surface gravity
\be
\kappa^2 = \frac{p_1^2 p_2 }{4 c_{a+1}^2}
= \frac{{h^{a+1}}'(z_a)^2}{4c_{a+1}^2 h^{a+1}_{11}(z_a) }  \; .  \label{5Dkappaaxis}
\ee
The second equality  follows from the relations $p_2 = h_{11}^{a+1}(z_a)$ and $p_1 p_2 = {h^{a+1}}'(z_a)$.  A similar analysis for $D=4$ (which effectively can be obtained from dropping the $\td x^1$ terms above)  gives
\be
\kappa^2= \frac{{h^{a+1}}'(z_a)^2}{4 c_{a+1}^2}  \; . \label{4dkappaaxis}
\ee
This analysis confirms the axis geometry on $I_{a+1}$ has a smooth non-degenerate horizon at $z=z_a$ with surface gravity (\ref{5Dkappaaxis}) for $D=5$ and (\ref{4dkappaaxis}) for $D=4$.  

On the other hand, as shown above, smoothness of the horizon metric at the corner $z=z_a$ leads to a different expression for $\kappa$. For $D=4$ this is given by (\ref{4dkappa}) and combining this with (\ref{4dkappaaxis}) implies
\be
\kappa^2 \gamma'(z_a)= c_{a+1}^{-1}{h^{a+1}}'(z_a)  \;, \label{contHA}
\ee 
where the signs are fixed from the fact that $\gamma'(z_a)<0$ and ${h^{a+1}}'(z_a)<0$. For $D=5$,  the expression for the surface gravity (\ref{5dkappa}), written in  coordinates  $\hat{\phi}^i, i=1,2$, adapted to the horizon rod $I_a$ so that $u_{a+1}= \partial_{\hat{1}}$ and $v_{a+1}= \partial_{\hat{2}}$, becomes
\be
\kappa^{-2} = \frac{\gamma'(z_a)^2}{4 \gamma_{\hat{1} \hat{1}}(z_a)} \; ,\label{5dkappasimp}
\ee
where we used $ \gamma'(z_a)= \gamma_{\hat{1} \hat{1}}(z_a) \gamma'_{\hat{2} \hat{2}}(z_a)$. Next, note that 
\be
h^{a+1}_{11}(z)= \mathring{{g}}_{AB}u_{a+1}^Au_{a+1}^B, \qquad  \gamma_{\hat{1} \hat{1}}(z)= \mathring{{g}}_{AB}u_{a+1}^Au_{a+1}^B \; ,
\ee
on the  rods $I_{a+1}$ and $I_a$ respectively,  are both equal to the norm squared of $u_{a+1}$, so in particular $h^{a+1}_{11}(z_a)= \gamma_{\hat{1} \hat{1}}(z_a)$. Hence eliminating the norm of $u_{a+1}$ between (\ref{5Dkappaaxis}) and (\ref{5dkappasimp}) we deduce that  (\ref{contHA}) also holds for $D=5$. The analysis for a horizon rod $I_a$ meeting an axis rod $I_{a-1}$ is entirely analogous and similarly to (\ref{contHA}) one can derive that
\begin{equation}
\label{contHAlower}
    \kappa^2 \gamma'(z_{a-1})= c_{a-1}^{-1} {h^{a-1}}'(z_{a-1})
\end{equation}
for $D=4,5$.

Finally, using (\ref{nu0}) and (\ref{nu0hor}) we see that (\ref{contHA}) is equivalent to the continuity of $|z-z_a|e^{2\mathring{\nu}}$ at $z=z_a$ (with a similar condition at $z=z_{a-1}$ for (\ref{contHAlower})), just as in the case of a corner separating two axis rods. 

\section{Ernst potential identities\label{apen:ernstHorizon}}

Consider a component of the horizon $H$ with corresponding rod $I_a$ and we drop rod labels when convenient and unambiguous. First, recall the well-known identity
\be
\int_{H} \star \td \xi = - 2\kappa A \;,
\ee
where $\xi$ is the horizon Killing field (\ref{corotKVF}), $\kappa$ is the surface gravity and $A$ is the area of $H$. Therefore, using (\ref{Hint}) we deduce that
\be
\zeta(z_a)-\zeta(z_{a-1})= - \frac{2 \kappa A }{(2\pi)^{D-3}}  \; , \label{cjump}
\ee
where we have defined a new potential $\zeta$ by
\be
\td  \zeta= \star (m_1 \wedge \dots m_{D-3}\wedge \td \xi)   \; .  \label{zetapot}
\ee
Also, we will need the following fact: in coordinates adapted to the horizon rod (\ref{ghoradapted}) implies that the 1-form dual to the corotating Killing field is 
\be
\xi_A=\tilde{g}_{A D-3} =O(\rho^2) 
\label{vnearh}
\ee
near the horizon. Thus, in particular, $\xi= 0$ on the horizon (although $\td \xi \neq 0$ since $\rho$ is not a good coordinate on the horizon). 

For $D=4$ we can write (\ref{4dernst}) in terms of the corotating Killing field
\be
\td b = -\star ( \xi \wedge \td \xi) +\Omega \star (\xi \wedge \td m)+ \Omega  \td \zeta- \Omega^2 \td \chi \;,
\ee
where we have used the definition of the twist potential (\ref{twist}) and (\ref{zetapot}).  Evaluating this on the horizon we see that the first two terms must vanish due to (\ref{vnearh}). Thus we find that on the horizon
\be
\td b = \Omega ( \td \zeta - \Omega \td \chi) 
\ee
and integrating this over the horizon rod $I_a$ gives
\be
b(z_a)- b(z_{a-1})
=  -\Omega \left(  \frac{\kappa A}{\pi} + 8\Omega J \right)= -4 \Omega M \;,
\ee
where in the first equality we used (\ref{cjump}) and (\ref{Jhorizon}) and in the final equality the standard  Smarr relation for the Komar mass of the horizon $M=\frac{1}{8\pi}\int_{H} \star \td \xi$.  This implies the identity (\ref{4dbinfty}).

For $D=5$, one can show again using (\ref{vnearh}) that on the horizon
\be
\td b^L_{\mu} = \left( \begin{array}{c} - \Omega_2 \Omega_i \td \chi_i+ \Omega_2 \td \zeta \\  \Omega_2 \td \chi_1 \end{array} \right)
\ee
and hence integrating this over the horizon rod
\be
b^L_\mu(z_a)- b^L_\mu(z_{a-1}) =\Omega_2 \left( \begin{array}{c} - \frac{4}{\pi} \left( \Omega_i J_i + \frac{\kappa A}{8\pi} \right) \\  \frac{4  J_1}{\pi} \end{array} \right)= \Omega_2 \left( \begin{array}{c}-  \frac{8M}{3\pi}  \\  \frac{4 J_1}{\pi} \end{array} \right) \;, \label{bLjumpH}
\ee
where in the first equality we used (\ref{Jhorizon}) and (\ref{cjump}) and in the second the Smarr relation. Similarly, 
one finds that on the horizon
\be
\td b^R_{\mu} = \left( \begin{array}{c}  -\Omega_1 \Omega_i \td \chi_i+  \Omega_1 \td \zeta \\  \Omega_1 \td \chi_2 \end{array} \right)
\ee
and hence
\be
b^R_\mu(z_a)- b_\mu^R(z_{a-1}) =\Omega_1 \left( \begin{array}{c}  -\frac{8M}{3\pi}  \\ \frac{4 J_2}{\pi} \end{array} \right)  \; .  \label{bRjumpH}
\ee
In a similar manner, one can also evaluate the change in Ernst potential  associated to any other axis rod over a horizon rod. Formulae for $b^L_\mu(z_a) - b^L_\mu(z_{a-1})$ and $b^R_\mu(z_a) - b^R_\mu(z_{a-1})$ across axis rods can also be derived, which combined with (\ref{bLjumpH}) and (\ref{bRjumpH}) imply the identities (\ref{5dbL}) and (\ref{5dbR}).

\section{Proof of Proposition \ref{prop:consistency}}
\label{apen:ghneq0}
First we observe that for an axis rod  $\tilde{G}_{a NN}(k) = v_a^T G_a(k) v_a$ where in the standard basis $v_a^T= (0, v_a^1, \dots, v_a^{D-3})$ is the rod vector.  Similarly, for a horizon rod we can write  $\tilde{G}_{a 00}(k) = v_a^T G_a(k) v_a$ where $v_a^T=(1, \Omega_1^a, \dots, \Omega^a_{D-3})$ denotes the horizon null vector.  Similar statements hold for the matrices $H_a(k)$.  Thus, to complete the proof of Proposition \ref{prop:consistency} we need to establish 
\begin{equation}
\label{gneq0}
    \lim_{k\to z_{a-1}}v_a^T G_a(k) v_a \neq 0
\end{equation}
and
\begin{equation}
\label{hneq0}
    \lim_{k\to z_{a}}v_a^T H_a(k) v_a \neq 0,
\end{equation}
for each finite rod $I_a$, for generic values of the parameters. We will only explicitly prove (\ref{gneq0}), though (\ref{hneq0}) can be proved in an almost identical fashion. 

Writing out $G_a$ explicitly in terms of the $P_a$ matrices using the expression for $F_a$ (\ref{Faalt}) gives
\begin{equation}
\begin{gathered}
    G_2(k) = -X_L(z_1,k) C^{-1} P_n(k)^T \cdots P_2(k)^T,\\
    G_a(k) = -X_{a-1}(z_{a-1},k) P_{a-2}(k)^{-1} \cdots P_1(k)^{-1} C^{-1} P_n(k)^T \cdots P_a(k)^T,
\end{gathered}
\end{equation}
where $a = 3,\dots, n$. Consider a fixed, but arbitrary set of axis rod vectors $v_a$ (this is of course only relevant for $D=5$). Then, from the definition of the matrices $G_a(k)$  it is clear that the LHS of (\ref{gneq0}) is a rational function $\mathcal{R}_a( \vec{\varphi })$ where the vector $\vec{\varphi}$ denotes the continuous moduli in (\ref{moduli}) (i.e. excluding the axis rod vectors).  For the purposes of the proposition we need to prove $\mathcal{R}_a( \vec{\varphi })\neq 0$ for {\it generic} values of the moduli $\vec{\varphi}$, i.e. the zero set of $\mathcal{R}_a$ is lower-dimensional.   A simple strategy to prove this is to find an explicit value of the moduli $\varphi_0$ for which $\mathcal{R}_a(\varphi_0) \neq 0$, since when combined with analyticity of the numerator of $\mathcal{R}_a$,  implies that the zero-set of $\mathcal{R}_a$ does not contain an open set. It is worth noting that for this argument  the value $\varphi_0$ does not need to belong to the actual moduli space of solutions (defined by (\ref{modspace})).

It is convenient to choose $\varphi_0$ for each rod $I_a$ such that $P_b(z_{a-1})=I_{D-3}$ for all $b \ne a-1$  and $1\le b \le n$. This is achieved by setting $b_\mu^{b}(z_{b})=0$ or $\chi_i^{b}(z_{b})=0$, depending on whether $I_{b}$ is an axis or horizon rod, and $z_b = z_{a-1} +1/2$. The result of this is that for any finite rod $I_a$
\begin{equation}
\label{gneq02}
    \lim_{k\to z_{a-1}}v_a^T G_a(k) v_a \to
    - v_a^T X_{a-1}(z_{a-1},z_{a-1}) C^{-1} v_a
\end{equation}
under these parameter identifications. Therefore in order to prove $(\ref{gneq0})$ all that remains is to show that the right hand side of (\ref{gneq02}) is generically nonzero.

First consider $D=4$. Using the explicit expression for $C$ (\ref{C4d}) one finds that
\begin{equation}
    - v_a^T X_{a-1}(z_{a-1},z_{a-1}) C^{-1} v_a = 
    \begin{cases}
    -1 + \Omega^{a-1}\chi^{a-1}(z_{a-1}), \quad & \text{$I_{a-1}$ horizon rod, $I_a$ axis rod},\\
    -1 + \Omega^a b^{a-1}(z_{a-1}), \quad & \text{$I_{a-1}$ axis rod, $I_a$ horizon rod},
    \end{cases}
\end{equation}
which are indeed generically nonzero.

Now consider $D=5$, in which case $C$ is explicitly given by (\ref{C5d}). If $I_{a-1}$ is an axis rod and $I_a$ is a horizon rod, we can also set $\Omega^a_i = 0$ which implies that the right hand side of $(\ref{gneq02})$ is simply given by $-1$. If $I_{a-1}$ is a horizon rod and $I_a$ is an axis rod then the right hand side of $(\ref{gneq02})$ is given by 
\begin{equation}
    [v_a^i \chi^{a-1}_i(z_{a-1})][\tilde{v}^T_a v_{a-1}] - v_a^T \tilde{v}_a,
\end{equation}
where $\tilde{v}_a^T = \begin{pmatrix}0 & - v_a^1 & v_a^2 \end{pmatrix}$, which is generically nonzero. Finally, if both $I_{a-1}$ and $I_a$ are axis rods then the right hand side of $(\ref{gneq02})$ is given by
\begin{equation}
    (\det A_{a-1})^{-1} \det \left(\begin{array}{cc} v^1_a & v_a^2 \\ v^1_{a-1} & v^2_{a-1}\end{array} \right) \tilde{v}_a^T (b_1^{a-1}(z_{a-1}) v_{a-1} - u_{a-1}),
\end{equation}
where the matrix $A_{a-1}$ and the axial Killing field $u_{a-1}$ are introduced in (\ref{uAdef}). The first factor is nonzero since $A_{a-1}\in GL(2, \mathbb{Z})$, the second factor is nonzero since $v_a$ and $v_{a-1}$ must be linearly independent (in particular see (\ref{admissible})), and the third factor is generically nonzero since $\tilde{v}_a$ cannot be orthogonal to both $v_{a-1}$ and $u_{a-1}$. This establishes the claim.


\begin{thebibliography}{99}

\bibitem{Chrusciel:2012jk}
P.~T.~Chrusciel, J.~Lopes Costa and M.~Heusler,
Living Rev. Rel. \textbf{15} (2012), 7
doi:10.12942/lrr-2012-7
[arXiv:1205.6112 [gr-qc]].

\bibitem{Emparan:2001wn}
  R.~Emparan and H.~S.~Reall,
  Phys.\ Rev.\ Lett.\  {\bf 88} (2002) 101101
  doi:10.1103/PhysRevLett.88.101101
  [hep-th/0110260].
  
\bibitem{Myers:1986un}
  R.~C.~Myers and M.~J.~Perry,
  Annals Phys.\  {\bf 172} (1986) 304.
  doi:10.1016/0003-4916(86)90186-7
  
\bibitem{Emparan:2008eg}
  R.~Emparan and H.~S.~Reall,
  Living Rev.\ Rel.\  {\bf 11} (2008) 6
  doi:10.12942/lrr-2008-6
  [arXiv:0801.3471 [hep-th]].

\bibitem{Hollands:2012xy}
  S.~Hollands and A.~Ishibashi,
  Class.\ Quant.\ Grav.\  {\bf 29} (2012) 163001
  doi:10.1088/0264-9381/29/16/163001
  [arXiv:1206.1164 [gr-qc]].


\bibitem{Emparan:2001wk}
  R.~Emparan and H.~S.~Reall,
  Phys.\ Rev.\ D {\bf 65} (2002) 084025
  doi:10.1103/PhysRevD.65.084025
  [hep-th/0110258].
  
  
\bibitem{Harmark:2004rm}
  T.~Harmark,
  Phys.\ Rev.\ D {\bf 70} (2004) 124002
  doi:10.1103/PhysRevD.70.124002
  [hep-th/0408141].
  
\bibitem{Hollands:2007aj}
  S.~Hollands and S.~Yazadjiev,
  Commun.\ Math.\ Phys.\  {\bf 283} (2008) 749
  doi:10.1007/s00220-008-0516-3
  [arXiv:0707.2775 [gr-qc]].
  
\bibitem{Hollands:2008fm}
  S.~Hollands and S.~Yazadjiev,
  Commun.\ Math.\ Phys.\  {\bf 302} (2011) 631
  doi:10.1007/s00220-010-1176-7
  [arXiv:0812.3036 [gr-qc]].
  
\bibitem{Figueras:2009ci}
  P.~Figueras and J.~Lucietti,
  Class.\ Quant.\ Grav.\  {\bf 27} (2010) 095001
  doi:10.1088/0264-9381/27/9/095001
  [arXiv:0906.5565 [hep-th]].
  
   \bibitem{Weinstein1990}
  Weinstein, G. (1990), 
  Comm. Pure Appl. Math., 43: 903-948. doi:10.1002/cpa.3160430705
  
  \bibitem{Weinstein1992}
  Weinstein, G. (1992), 
  Comm. Pure Appl. Math., 45: 1183-1203. doi:10.1002/cpa.3160450907
  
 
  
  \bibitem{Weinstein1994}
  Weinstein, Gilbert, 
  Trans. Amer. Math. Soc. 343 (1994), no. 2, 899906.
  
   \bibitem{TY}
  Tian, Gang, and Li, Yan Yan, 
  Manuscripta mathematica 73.1 (1991): 83-90. http://eudml.org/doc/155656.

  
   
  \bibitem{KN}
  Kramer, D. and Neugebauer, G., 
  Phys. Lett. A \textbf{75}, 259
(1980)
  
  \bibitem{N}
  G. Neugebauer 1980 J. Phys. A: Math. Gen. \textbf{13} L19
  
  \bibitem{BZ2}
  V.~A.~Belinsky and V.~E.~Zakharov, 
  Sov. Phys. JETP, {\bf 50}, 1, (1979).
  
\bibitem{Tomimatsu:1981bc}
A.~Tomimatsu and M.~Kihara,
Prog. Theor. Phys. \textbf{67} (1982), 1406
doi:10.1143/PTP.67.1406
  
  \bibitem{MR}
   V. S. Manko and E. Ruiz,  Class. Quantum Grav. \textbf{18} L11  (2001)
   
\bibitem{Neugebauer:2011qb}
  G.~Neugebauer and J.~Hennig,
  J.\ Geom.\ Phys.\  {\bf 62} (2012) 613
  doi:10.1016/j.geomphys.2011.05.008
  [arXiv:1105.5830 [gr-qc]].
  
\bibitem{Varzugin:1997ee}
G.~G.~Varzugin,
Theor. Math. Phys. \textbf{111} (1997), 667
doi:10.1007/BF02634055
[arXiv:gr-qc/0004073 [gr-qc]].

\bibitem{Neugebauer:2003qe}
  G.~Neugebauer and R.~Meinel,
  J.\ Math.\ Phys.\  {\bf 44} (2003) 3407
  doi:10.1063/1.1590419
  [gr-qc/0304086].


\bibitem{Hennig:2008yw}
J.~Hennig, M.~Ansorg and C.~Cederbaum,
Class. Quant. Grav. \textbf{25} (2008), 162002
doi:10.1088/0264-9381/25/16/162002
[arXiv:0805.4320 [gr-qc]].


  
\bibitem{Dain:2011pi} 
  S.~Dain and M.~Reiris,
  Phys.\ Rev.\ Lett.\  {\bf 107}, 051101 (2011)
  doi:10.1103/PhysRevLett.107.051101
  [arXiv:1102.5215 [gr-qc]].
  
\bibitem{Chrusciel:2011iv}
  P.~T.~Chrusciel, M.~Eckstein, L.~Nguyen and S.~J.~Szybka,
  Class.\ Quant.\ Grav.\  {\bf 28} (2011) 245017
  doi:10.1088/0264-9381/28/24/245017
  [arXiv:1111.1448 [gr-qc]].

  
\bibitem{Khuri:2017xsc}
  M.~Khuri, G.~Weinstein and S.~Yamada,
  Diff.\ Eq.\  {\bf 43} (2018) 1205
  [arXiv:1711.05229 [gr-qc]].
  
\bibitem{Khuri:2018udf}
  M.~Khuri, G.~Weinstein and S.~Yamada,
  PTEP {\bf 2018} (2018) no.5,  053E01
  doi:10.1093/ptep/pty052
  [arXiv:1802.02457 [hep-th]].
  
\bibitem{Elvang:2007rd}
  H.~Elvang and P.~Figueras,
  JHEP {\bf 0705} (2007) 050
  doi:10.1088/1126-6708/2007/05/050
  [hep-th/0701035].
  
\bibitem{Iguchi:2007is}
  H.~Iguchi and T.~Mishima,
  Phys.\ Rev.\ D {\bf 75} (2007) 064018
   Erratum: [Phys.\ Rev.\ D {\bf 78} (2008) 069903]
  doi:10.1103/PhysRevD.78.069903, 10.1103/PhysRevD.75.064018
  [hep-th/0701043].
  
\bibitem{Izumi:2007qx}
  K.~Izumi,
  Prog.\ Theor.\ Phys.\  {\bf 119} (2008) 757
  doi:10.1143/PTP.119.757
  [arXiv:0712.0902 [hep-th]].
  
\bibitem{Elvang:2007hs}
  H.~Elvang and M.~J.~Rodriguez,
  JHEP {\bf 0804} (2008) 045
  doi:10.1088/1126-6708/2008/04/045
  [arXiv:0712.2425 [hep-th]].

  
\bibitem{BZ1}
  V.~A.~Belinsky and V.~E.~Zakharov,
  Sov.\ Phys.\ JETP {\bf 48} (1978) 985
   [Zh.\ Eksp.\ Teor.\ Fiz.\  {\bf 75} (1978) 1953].
  
\bibitem{Evslin:2008gx}
J.~Evslin,
JHEP \textbf{09} (2008), 004
doi:10.1088/1126-6708/2008/09/004
[arXiv:0806.3389 [hep-th]].

  
\bibitem{Chen:2008fa}
  Y.~Chen and E.~Teo,
  Phys.\ Rev.\ D {\bf 78} (2008) 064062
  doi:10.1103/PhysRevD.78.064062
  [arXiv:0808.0587 [gr-qc]].
  
\bibitem{Tomizawa:2019acu}
  S.~Tomizawa and T.~Mishima,
  Phys.\ Rev.\ D {\bf 99} (2019) no.10,  104053
  doi:10.1103/PhysRevD.99.104053
  [arXiv:1902.10544 [hep-th]].
  
\bibitem{Kunduri:2018qqt}
  H.~K.~Kunduri and J.~Lucietti,
  Class.\ Quant.\ Grav.\  {\bf 36} (2019) no.7,  07LT02
  doi:10.1088/1361-6382/ab0982
  [arXiv:1810.13210 [hep-th]].

\bibitem{Pomeransky:2006bd}
  A.~A.~Pomeransky and R.~A.~Sen'kov,
  hep-th/0612005.

\bibitem{Varzugin:1998wf}
G.~G.~Varzugin,
Theor. Math. Phys. \textbf{116} (1998), 1024
doi:10.1007/BF02557144
[arXiv:gr-qc/0005035 [gr-qc]].
  
   

  
   
    
   
\bibitem{Gibbons:2002bh}
  G.~W.~Gibbons, D.~Ida and T.~Shiromizu,
  Prog.\ Theor.\ Phys.\ Suppl.\  {\bf 148} (2003) 284
  doi:10.1143/PTPS.148.284
  [gr-qc/0203004].
  
\bibitem{Morisawa:2004tc}
  Y.~Morisawa and D.~Ida,
  Phys.\ Rev.\ D {\bf 69} (2004) 124005
  doi:10.1103/PhysRevD.69.124005
  [gr-qc/0401100].
  
\bibitem{Chen:2011jb}
  Y.~Chen, K.~Hong and E.~Teo,
  Phys.\ Rev.\ D {\bf 84} (2011) 084030
  doi:10.1103/PhysRevD.84.084030
  [arXiv:1108.1849 [hep-th]].
  
   
\bibitem{Tod:2013ska}
P.~Tod, N.~Metzner and L.~Mason,
Class. Quant. Grav. \textbf{30} (2013), 095002
doi:10.1088/0264-9381/30/9/095002
[arXiv:1303.0849 [gr-qc]].

  
   
   \bibitem{MY}
  A. V. Mikhailov and A. I. Yaremchuk, Nucl. Phys. B \textbf{202}, 508 (1982)
  
  \bibitem{BZM}
  Burtsev, S.P., Zakharov, V.E. and Mikhailov, A.V., Theor. Math. Phys. (1987) 70: 227. https://doi.org/10.1007/BF01040999
  
    

   
    
  

\bibitem{MS}
V S Manko and N R Sibgatullin, Class. Quantum Grav. \textbf{10} (1993)  1383

\bibitem{Hennig:2019knn}
J.~Hennig,
Class. Quant. Grav. \textbf{36} (2019) no.23, 235001
doi:10.1088/1361-6382/ab4f41
[arXiv:1906.04847 [gr-qc]].

\bibitem{Figueras:2009mc}
P.~Figueras, E.~Jamsin, J.~V.~Rocha and A.~Virmani,
Class. Quant. Grav. \textbf{27} (2010), 135011
doi:10.1088/0264-9381/27/13/135011
[arXiv:0912.3199 [hep-th]].

\bibitem{Bena:2007kg}
I.~Bena and N.~P.~Warner,
Lect. Notes Phys. \textbf{755} (2008), 1-92
[arXiv:hep-th/0701216 [hep-th]].

\bibitem{Kunduri:2014iga}
H.~K.~Kunduri and J.~Lucietti,
JHEP \textbf{10} (2014), 082
doi:10.1007/JHEP10(2014)082
[arXiv:1407.8002 [hep-th]].

\bibitem{Kunduri:2014kja}
H.~K.~Kunduri and J.~Lucietti,
Phys. Rev. Lett. \textbf{113} (2014) no.21, 211101
doi:10.1103/PhysRevLett.113.211101
[arXiv:1408.6083 [hep-th]].

\bibitem{Tomizawa:2016kjh}
S.~Tomizawa and M.~Nozawa,
Phys. Rev. D \textbf{94} (2016) no.4, 044037
doi:10.1103/PhysRevD.94.044037
[arXiv:1606.06643 [hep-th]].

\bibitem{Horowitz:2017fyg}
G.~T.~Horowitz, H.~K.~Kunduri and J.~Lucietti,
JHEP \textbf{06} (2017), 048
doi:10.1007/JHEP06(2017)048
[arXiv:1704.04071 [hep-th]].

\bibitem{Breunholder:2017ubu}
V.~Breunh\"older and J.~Lucietti,
Commun. Math. Phys. \textbf{365} (2019) no.2, 471-513
doi:10.1007/s00220-018-3215-8
[arXiv:1712.07092 [hep-th]].

\bibitem{Breunholder:2018roc}
V.~Breunh\"older and J.~Lucietti,
JHEP \textbf{03} (2019), 105
doi:10.1007/JHEP03(2019)105
[arXiv:1812.07329 [hep-th]].
  
\end{thebibliography}
\end{document}